\documentclass[]{article}
\usepackage[utf8]{inputenc}
\usepackage{amsmath}
\usepackage{amsthm, amssymb}
\usepackage{shadethm}
\usepackage[nothm]{thmbox}
\usepackage{graphicx}
\usepackage{thmtools}
\usepackage{subcaption}
\usepackage{xcolor}
\usepackage{algpseudocode, algorithm, algorithmicx}
\usepackage{psfrag}
\usepackage{pstool}
\usepackage{comment}
\usepackage{epstopdf}
\usepackage[affil-it]{authblk}
\usepackage{geometry}
\usepackage{mathtools}
\usepackage{tikz}
\usepackage{bbm}
\usepackage{bm}
\usepackage[normalem]{ulem}
\usepackage{cancel}
\usetikzlibrary{shapes,arrows,3d}
\definecolor{tumor}{HTML}{E37222}
\definecolor{tumred}{RGB}{205,32,44}
\definecolor{tumblue}{RGB}{00,115,207}
\definecolor{tumgreen}{RGB}{162,173,000}
\newcommand{\red}[1]{{\color{tumred} #1}}
\newcommand{\blue}[1]{{#1}}

\numberwithin{equation}{section}

\usepackage{pgfplots}
\usepackage{pgfplotstable}
\newtheorem{theorem}{Theorem}

\newtheorem{lemma}[theorem]{Lemma}
\newtheorem{definition}[theorem]{Definition}
\newtheorem{remark}[theorem]{Remark}
\newcommand\R{\mathbb{R}}
\newcommand\ep{\varepsilon}
\renewcommand\L{\mathcal{L}}

\newcommand\T{\mathcal{T}}
\newcommand\B{\mathcal{B}}
\newcommand\St{\mathcal{S}}
\newcommand\1{\mathbf{1}}
\newcommand\0{\boldsymbol{0}}
\newcommand\A{\mathbf{A}}
\renewcommand\a{\mathbf{a}}
\newcommand\g{\mathbf{g}}
\renewcommand{\u}{\mathbf{u}}
\newcommand\w{\mathbf{w}}
\newcommand\x{\mathbf{x}}
\newcommand\y{\mathbf{y}}
\newcommand\z{\mathbf{z}}
\newcommand\qc{\mathbf{q}_{\text{corr}}}

\newcommand\q{\mathbf{q}}

\newcommand\ttau{\boldsymbol{\tau}}
\newcommand{\nnu}{\boldsymbol{\nu}}
\newcommand\sign{\operatorname{sign}}
\newcommand\conv{\operatorname{conv}}
\newcommand{\inner}[2]{\langle #1, #2 \rangle}
\newcommand{\norm}[2]{\left\Vert #1 \right\Vert_{#2}}
\newcommand{\relu}[1]{\left[#1 \right]_+}
\newcommand{\la}{\lambda}
\newcommand\unif{\operatorname{Unif}}
\newcommand\N{\mathcal{N}}

\renewcommand\S{\mathcal{S}}
\renewcommand\P{\mathbb{P}}
\newcommand\E{\mathbb{E}}

\renewcommand\gg{\gtrsim}
\renewcommand\ep{\varepsilon}
\newcommand\eps{\varepsilon}

\usepackage{hyperref}
\usepackage{cleveref}
\usepackage{authblk}
\renewcommand\L{\mathcal{L}}

\author[1]{Hans Christian Jung}
\affil[1]{ \emph{DeepL, Cologne, Germany}}
\author[2]{Johannes Maly}
\affil[2]{ \emph{Department of Scientific Computing, KU Eichstaett/Ingolstadt, Germany}}
\author[3]{Lars Palzer}
\affil[3]{ \emph{Signal Iduna group, Hamburg, Germany}}
\author[4]{Alexander Stollenwerk}
\affil[4]{ \emph{ICTEAM Institute, ISPGroup, UCLouvain, Belgium}}
\begin{document}

%\large

%\label{eq:1}
%\eqref{eq:1}

%\maketitle

\title{Quantized Compressed Sensing by Rectified Linear Units}
\maketitle

\begin{abstract}
This work is concerned with the problem of recovering high-dimensional  signals $\x\in \R^n$ which belong to a convex set of low complexity from a small number of quantized measurements.
We propose to estimate the signals via a convex program based on
rectified linear units (ReLUs) for two different quantization schemes, namely one-bit and uniform multi-bit quantization.
Assuming that the linear measurement process can be modelled by a sensing matrix with i.i.d. subgaussian rows, we 
obtain for both schemes near-optimal
uniform reconstruction guarantees by adding well-designed noise to the linear measurements prior to the quantization step. In the one-bit case, 
we show that the program is robust against adversarial bit corruptions as well as additive noise on the linear measurements.	%The program can be viewed as a convex proxy for the untractable Hamming distance minimization program. 
Further, our analysis quantifies precisely how the
rate-distortion relationship of the program
changes depending on whether we seek reconstruction accuracies above or below \blue{the level of additive noise.}
%the noise floor.
The proofs rely on recent results by Dirksen and Mendelson on non-Gaussian hyperplane tessellations.
Finally, 
we complement our theoretical analysis with numerical experiments which compare our method to other state-of-the-art methodologies. \\

\end{abstract}

%%%%%%%%%%%%%%%%%%%%%%%%%%%%%%%%%%%
%%% Introduction
%%%%%%%%%%%%%%%%%%%%%%%%%%%%%%%%%%%

\section{Introduction}
The compressed sensing paradigm provides methods to infer accurate information about high-dimensional signals $\x \in \R^n$ from few linear measurements
\begin{align} \label{eq:CS}
    \y = \A\x \in \R^{m},
\end{align}
where $\A \in \R^{m \times n}$ models a specific
%the -> a specific
measurement process. %In general impossible if $m < n$, 
The last decade showed that using prior knowledge on the unknown signal (e.g., sparsity) enables unique identification of $\x$ from $m \ll n$ measurements $\y$ via efficient recovery methods. In fact, under the assumption that $\x$ is $s$-sparse, 
meaning that at most $s$ entries are non-zero, unique recovery of $\x$ from $\y$ is possible, whenever we are given at least
\begin{align} \label{eq:CSmeas}
    m \ge C
    s \log\left( \frac{en}{s} \right)
\end{align}
measurements. Here, and in the following, $C > 0$ denotes an absolute constant. Starting with the works of \cite{candes_robust_2006,donoho_compressed_2006} compressed sensing has developed into a lively field of research which inspired new solutions 
%and methods 
for problems in various applied sciences \cite{haldar2010compressed,murphy2012fast,herman2008high}. We refer to \cite{foucart_mathematical_2013} for a comprehensive discussion of compressed sensing and its applications.\\
Though the linear model \eqref{eq:CS} is powerful enough to encompass many import models of measurement processes,
%may be assumed to behave in a linear way, 
it is blind to the fact, that in real world scenarios the measurements have to be 
quantized to a finite number of bits before the signal reconstruction can be performed.
%ignores an important factor, namely digital quantization. 
The process of projecting the infinite precision measurements (captured as a real number) onto a finite alphabet
$\mathcal{A} \subset \R$ is called \textit{quantization}.
%Since it is impossible to store real numbers with infinite precision, one needs to quantize the analog measurements when digitalizing them, i.e., the measurements are projected onto a finite quantization alphabet $\mathcal{A} \subset \R$. 
\blue{ Adapting \eqref{eq:CS} accordingly leads to the \textit{quantized compressed sensing} model
\begin{align} \label{eq:QCS}
    \q = Q(\A\x) \in \mathcal{A}^m,
\end{align}
where the quantizer $Q\colon \R^m \rightarrow \mathcal{A}^m$ maps the linear measurements $\A\x$ to quantized measurements $\q$. Quantization in general leads to a loss of information which makes exact signal recovery impossible.
Therefore, in the quantized compressed sensing model we are interested in designing quantizers $Q$ which permit efficient approximation of $\x$ from $\q$ using as few measurements as possible. \\
Although there exist more sophisticated quantization schemes, e.g., noise-shaping (for an overview see \cite{boufounos_quantization_2015, dirksen2019quantized}), our focus is on memoryless scalar quantization where the quantizer acts component-wise, i.e., $Q:\R \to \mathcal{A}$, and each linear measurement $\langle \a_i,\x \rangle$ is quantized 
independently of all other measurements.}
In this context, we call $Q$ a $B$-\textit{bit quantizer} if $|\mathcal{A}| = 2^B$ and restrict ourselves to uniform quantization, which  
%its performance approaches optimality for increasing bit rates 
admits a rather simple structure. Let us mention that uniform quantizers approach optimality for increasing bit rates, cf.\ \cite[p. 2332 et sqq.]{gray1998quantization}). If $B$ is large, the measurement defect caused by $Q$ could be treated as noise and classical compressed sensing results would apply. However, in this case the reconstruction error cannot be smaller than the resolution of the quantizer.
Assuming knowledge on the quantization process (which is most often the case in applications) this is suboptimal and does not use all available information \cite{jacques_dequantizing_2011}. Moreover, modern applications require as well a treatment of coarse quantization, i.e., $B$ is small, \cite{bennett2007netflix} and measurement devices become considerably cheaper in this regime, cf.\ \cite{boufounos_1-bit_2008}.
\\
In its coarsest form, $Q$ quantizes every measurement $\inner{\a_i}{\x}$ to one single bit $\q_i \in \{\pm 1\}$. Following \cite{boufounos_1-bit_2008}, several works \cite{boufounos_greedy_2009,jacques_robust_2013,plan_one-bit_2013} examined one-bit quantization in compressed sensing and were able to derive recovery conditions which are asymptotically equivalent to \eqref{eq:CSmeas}. They could show for different efficient algorithms that it is possible to approximate $s$-sparse unit norm signals $\x$ up to reconstruction precision $\rho$ from $m \ge C(\rho)s\log(en/s)$ one-bit compressive measurements of the form
\begin{align} \label{eq:OneBitCS}
    \q = \sign(\A\x),
\end{align}
where $C(\rho) > 0$ is a constant only depending on $\rho$. Since \eqref{eq:OneBitCS} looses any scaling information, for this quantization scheme it is necessary to assume that the signals are normalized
in order to prove approximation guarantees. To circumvent the normalization restriction, subsequent works \cite{knudson_one-bit_2014,jacques2016error} added a random dither $\ttau \in \R^m$ to \eqref{eq:OneBitCS} leading to the dithered one-bit compressed sensing model
\begin{align} \label{eq:OneBitCSdithered}
    \q = \sign(\A\x + \ttau)
\end{align}
allowing signal approximation for general $s$-sparse signals. The origin of these dithering techniques goes back to
the work \cite{roberts_picture_1962} where dithering\footnote{A geometric perspective on \eqref{eq:OneBitCS} and \eqref{eq:OneBitCSdithered} clarifies the role played by the dither. We can associate to each row $\a_i$ of $\A$ the hyperplane $H_{\a_i}:= \{\x \in \R^N : \inner{\a_i}{\x} = 0\}$, which is orthogonal to
$\a_i$ and contains the origin. For $1 \le i \le m$, each measurement $q_i$ in \eqref{eq:OneBitCS} characterizes on which side of $H_{\a_i}$ the signal $\x$ lies. All hyperplanes together yield a random tessellation of the unit sphere into at most $2^m$ cells %, given by the sets $S^{n-1} \setminus \bigcup_{i=1}^m H_{\a_i}$, 
and $\q\in \{\pm 1\}^m$ encodes in which cell $\x$ lies. Adding the dither introduces an offset to the hyperplanes which leads, depending on the dithers $(\tau_i)_{i \in [m]}$, to a random tessellation not only of the sphere but of the whole space $\R^n$. The geometrical intuition is also helpful for multi-bit quantization. In this case, each measurement corresponds not to one single hyperplane but to a parallel bundle of hyperplanes (see Section \ref{sec:MultiBit}).} was introduced in order to remove artefacts from
quantized pictures (see also \cite{gray1998quantization}). 

\blue{In comparison to unquantized compressed sensing where measurement noise is usually modelled as a bounded additive perturbation of $\y$ in \eqref{eq:CS}, one typically considers two different types of measurement noise in the one-bit models $\eqref{eq:OneBitCS}$ and $\eqref{eq:OneBitCSdithered}$: additive (statistical) noise $\nnu \in \R^m$ disturbing the linear measurements before quantization and (adversarial) bit-flips, cf.\ \cite{jacques_robust_2013,plan_robust_2013,dirksen_robust_2018}. For an extended discussion 
on the goals of quantized compressed sensing and its particular challenges, we refer the reader to the recent survey \cite{dirksen2019quantized}.
In this work, we consider both additive statistical noise as well as adversarial bit corruptions.
More specifically, we aim to recover vectors $\x$ from adversarially corrupted one-bit measurements $\qc$ which satisfy
\begin{equation}
\label{eq:hamming_distance_measurements}
    d_H(\qc, q(\x))\leq \beta m,
\end{equation}
where the Hamming distance $d_H(\z,\z') = |\{ i \colon z_i \neq z_i' \}|$ counts the number of entries where $\qc$ differs from $q(\x) = \sign(\A \x + \nnu + \ttau)$, and $\nnu\in \R^m$ models subgaussian additive noise. \\
Finally, let us mention that we do not exclusively focus our work on sparse signals, but allow for general compact and convex sets $\T \subset \R^n$ as a prior. To obtain sparse reconstruction results, $\T$ would be chosen as a properly scaled $\ell_1$-ball and we will later on allude to sparse reconstruction as benchmark and sanity check. 
%In order to estimate $\x$ from $\qc$ we propose a new recovery program, which uses the rectified linear unit as data fidility term and encodes prior information on $\x$ by means of a general compact hypothesis set $\T \subset \R^n$. 

%If signals $\x$ are sparse, 

%which belongs to the class of   

}

%\blue{While sparsity might be the widest known structural prior and served well for introductory purposes, we consider in the rest of the work general compact signal sets. That is, $\x \in \T$ where $\T \subset \R^n$ is known in advance and, though high-dimensional, has low intrinsic complexity (we concretize the vague expression "complexity" below). At the moment, viewing $\T$ as the set of sparse vectors, or more precisely its restriction to a compact ball around the origin, provides sufficient intuition. We will later on allude to sparse reconstruction as benchmark and sanity check.}

\begin{figure}
    %\centering
    \begin{subfigure}{0.45\textwidth}
        \centering
        \includegraphics[scale=0.35]{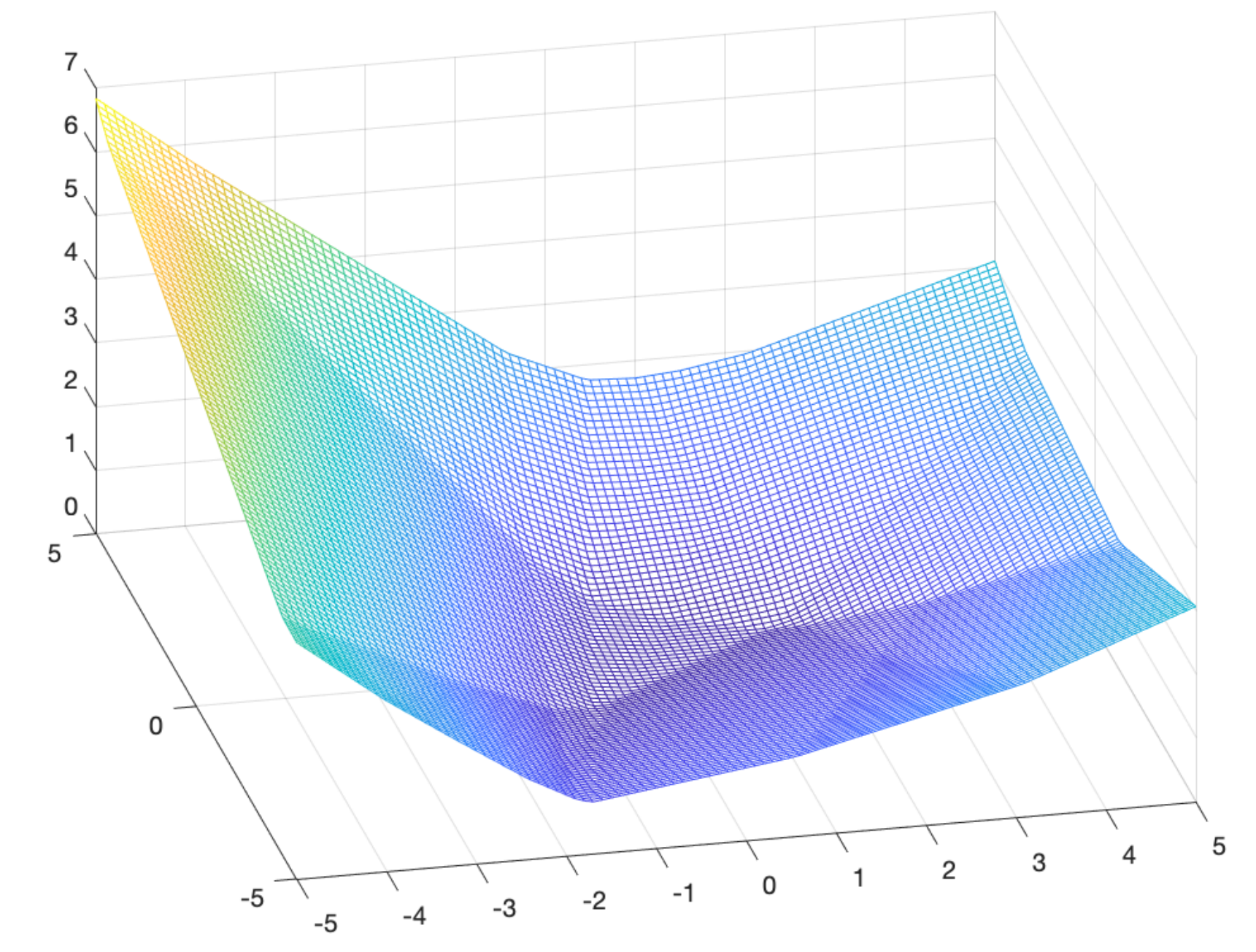}
        \caption{The functional $\L_\q$.}
    \end{subfigure} \quad
    \begin{subfigure}{0.45\textwidth}
        \centering
        \includegraphics[scale=0.35]{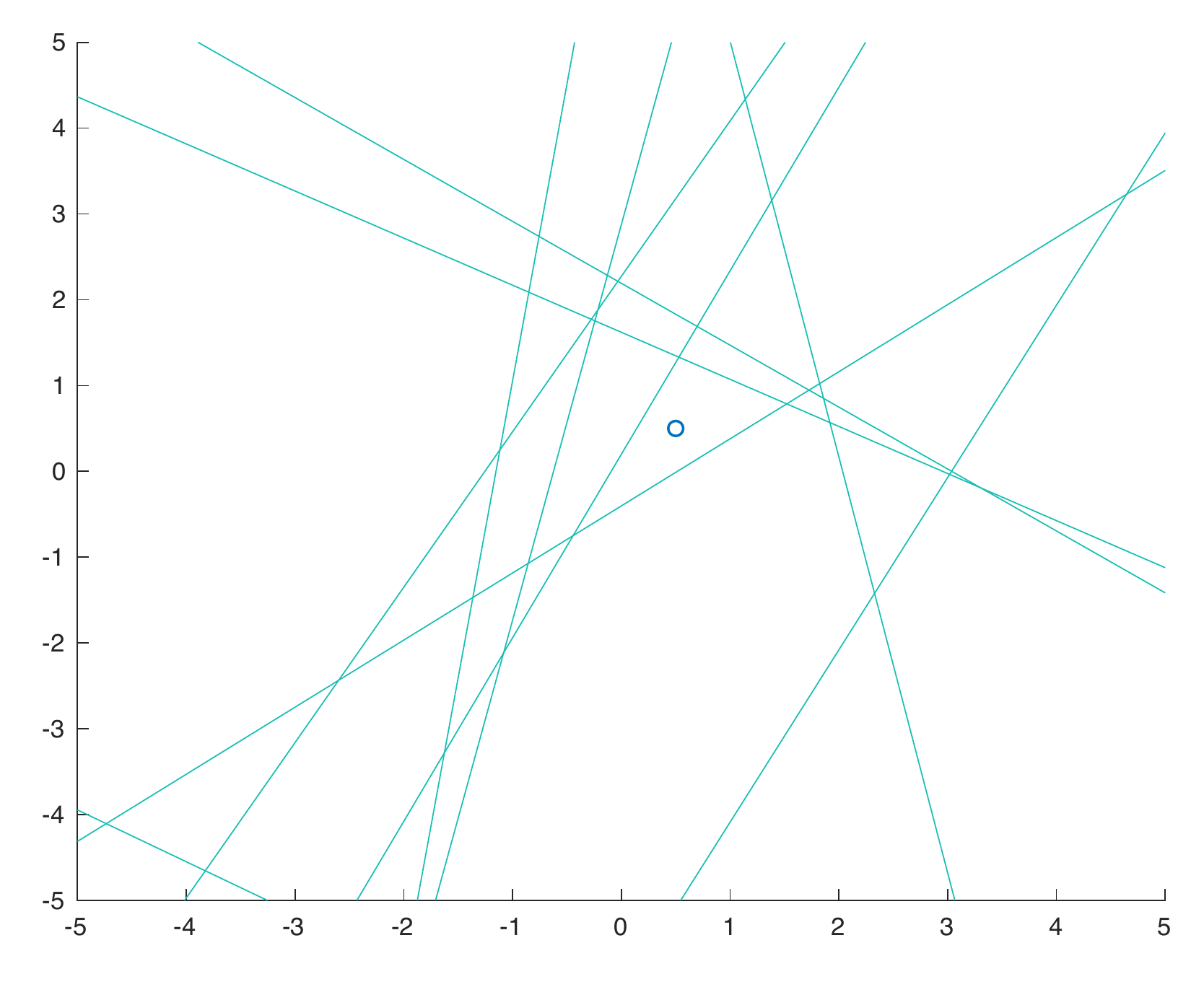}
        \caption{The underlying tessellation and $\x$.}
    \end{subfigure}
    \caption{Example of $\L_\q$ and its underlying tessellation in $\R^2$. The parameters are chosen as $\x = (0.5,0.5)^T$, $m = 10$, $\tau_i \sim \text{Unif}([-0.5,0.5])$, and $\q = \sign(\A\x + \ttau)$. In (a) one can clearly see the Hamming structure from (b) encoded in $\L_\q$.}
    \label{fig:L}
\end{figure}{}

\subsection{Contribution} 
% TODO: Since our contribution starts by mentioning the ReLU functional, the reader might think that we invented this method. Which is 
% not really true. We consider this functional with dithering, instead of using a non-feasible manifold-constraint, e.g., 
% |x|_2 = 1.
\textbf{(i)} 
\blue{We propose to estimate vectors $\x \in \T\subset \R^n$ from $\qc$ in \eqref{eq:hamming_distance_measurements} by minimizing the convex functional
\begin{equation} \label{eq:L}
    \L_{\qc}(\z) := \frac{1}{m} \sum_{i=1}^m \relu{-(q_{\text{corr}})_i(\inner{\a_i}{\z} + \tau_i)} \; ,
\end{equation}
where $\relu{x} = \max\{x,0\}$ denotes the \emph{rectified linear unit (ReLU)}, cf.\ Figure \ref{fig:L}. This amounts to solving the program
\begin{equation} \label{eq:P}
\tag{$P_{\qc}$}
 \min_{\z \in \R^n} \L_{\qc}(\z) \quad \text{subject to}\quad  \z\in \T,
\end{equation}
which is a convex program whenever the set $\T$ is a convex subset of $\R^n$. From a geometric point of view, the functional $\L_{\qc}$ is a convex proxy for the Hamming distance based function $\z \mapsto d_H(\qc, \sign( \A \z + \ttau ))$ and exhibits intuitive relations to established approaches in recent literature, cf.\ Section \eqref{sec:FirstIntuition} below. By designing the dither $\ttau$ to be uniformly distributed in $[-\la, \la]^m$ for a large enough parameter $\la>0$, we provide in Theorem~\ref{thm:noisy_recovery} near-optimal robust uniform reconstruction guarantees for \eqref{eq:P} under the assumption of i.i.d. subgaussian measurement vectors $\a_i$. These guarantees match state-of-the art results established in \cite{dirksen_robust_2018} for reconstruction accuracies below the noise level $\norm{\nu}{L^2}$, and improve them for accuracies above.} 
%In Theorem~\ref{thm:noisy_recovery}, we provide robust reconstruction guarantees for \eqref{eq:P}. When we seek reconstruction accuracies above the noise floor (low-noise regime), the guarantees for \eqref{eq:P} match the near-optimal guarantees established in \cite[Theorem 1.5]{dirksen_robust_2018} for the non-tractable Hamming distance minimization program. For reconstruction accuracies below the noise floor (high-noise regime), the guarantees for \eqref{eq:P} match those given in Theorem \ref{thm:Comparison}. The results show that \eqref{eq:P} is as robust as \eqref{eq:Dirksen} in the high-noise regime and in addition outperforms \eqref{eq:Dirksen} in the low-noise regime. Our results build upon tools developed in \cite{dirksen_robust_2018}. 
%The results improve on Theorem \ref{thm:Comparison} by removing one $(\lambda/\rho)$-factor in \eqref{eqn:Comparison} -- any $\x^\#$ solving \eqref{eq:P} satisfies $\norm{\x - \x^\#}{2} \leq (\lambda w_*(\T)^2/m)^{1/3}$ with high probability -- but work only for an approximation accuracy $\rho$ up to noise variance $\norm{\nnu}{L^2}$ showing that \eqref{eq:P} complements \eqref{eq:Dirksen} for moderate additive noise scenarios. 

\textbf{(ii)} \blue{We extend the estimation scheme proposed in \eqref{eq:L} and \eqref{eq:P} to a memoryless uniform multi-bit quantization model with refinement level $\Delta > 0$, which is the worst-case distortion when $Q$ is applied on a bounded
domain. This leads to a modified version of \eqref{eq:L} and \eqref{eq:P}, see Section~\ref{sec:MultiBit} for details. 
By choosing the dither $\ttau$ to be uniformly distributed in $[-\Delta,\Delta]^m$, we derive
uniform reconstruction guarantees, see Theorem \ref{thm:MultiBitMain}, which exhibit two, for multi-bit quantization
problems characteristic regimes:
if we ask for approximation accuracies above the refinement level $\Delta$, the established guarantee resembles results typical for noisy unquantized compressed sensing models \eqref{eq:CS}. If we ask for accuracies below the refinement level, additional oversampling becomes necessary and the result bears resemblance to the one-bit case.} 
%This overall behavior is expect from the nature of the multi-bit quantization scheme.}
%In particular, we obtain that if $\rho \gtrsim \Delta$, then $m \gtrsim \rho^{-2} w_*((\T-\T)\cap \rho B_2^n)^2$ measurements suffice for approximation up to $\ell_2$-error $\rho$. On the other hand, if $\rho \lesssim \Delta$, then the guarantees resemble the one-bit case, namely $\norm{\x - \x^\#}{2} \leq (\Delta w_*(\T)^2/m)^{1/3}$ (note the independence of the signal set's radius here). 
To keep the presentation concise, we do not consider noise on the measurements for multi-bit quantization (though possible as is evident from the one-bit setting; \blue{the interested reader is referred to the comments following
Theorem~\ref{thm:MultiBitMain}).} \blue{Let us emphasize that both in the one-bit and multi-bit setting our reconstruction guarantees are near-optimal in the context of memory-less scalar quantization (see the discussion on optimality below Theorem~\ref{thm:MultiBitMain} and \cite[Section 5]{dirksen_robust_circulant_2018}). Note that adaptive quantization schemes \cite{baraniuk_2017_exponential} allow better rates. Due to energy consumption and dependencies between single measurements, however, those systems can be of limited use in modern applications like Massive MIMO \cite{haghighatshoar2018low}.}

%\textbf{(iii)} 
Numerical experiments illustrate the performance of \eqref{eq:P} in both settings. \blue{While the idea of using half sided $\ell_1$- and $\ell_2$-norms is not new for quantized compressed sensing and appeared before in \cite{jacques_robust_2013,boufounos2010reconstruction}, to the best of our knowledge, this is the first work unifying one- and multi-bit quantization for compressed sensing in a single tractable program and analysis.} 
%provide the first thorough analysis of the program \eqref{eq:P} so far, which provides a performance-analysis for $1$-bit and multi-bit quantization. 

\subsection{Related Work (One-Bit)}

\blue{Let us recap the state of the art in memoryless scalar quantized compressed sensing before presenting our main results in full detail. Recent results treat measurement noise as well as general signal sets $\T \subset \R^n$. In this context, the Gaussian width is a natural complexity parameter which has proven to accurately capture the effective size of signal sets in many problems from signal processing. The Gaussian width of a set $\T\subset \R^n$ is
defined as 
\begin{displaymath}
 w_*(\T) := \E\sup_{\x \in \T} |\inner{\g}{\x}|,
\end{displaymath}
where $\g\in \R^n$ is a standard Gaussian random vector. As a rule of thumb one may say that the squared Gaussian width corresponds to the information theoretic intrinsic complexity of $\T$.} Further, 
the Gaussian width plays an important role in high-dimensional probability, statistics and geometry. 
For more information on its role in problems from signal recovery, the reader is referred to \cite{plan_robust_2013, vershynin_estimation_2014, amelunxen14}.

\blue{The first result on robust recovery from one-bit quantized compressed sensing measurements using a tractable recovery algorithm appeared in \cite{plan_robust_2013} for an un-dithered model. The authors showed that if $\A$ is an $m\times n$ Gaussian matrix and $m \gtrsim \rho^{-6} w_*(\T)^2$, for $\rho > 0$, then, with high probability, every $\x \in \T \cap \St^{n-1}$ is recovered from $\qc$ as in \eqref{eq:hamming_distance_measurements}, for $\beta > 0$ and $q(\x)$ as in \eqref{eq:OneBitCS}, by any solution $\x^\#$ of the program
\begin{align} \label{eq:Plan}
    \max_{\z \in \T} \frac{1}{m} \sum_{i = 1}^m (q_{\text{corr}})_i \langle \a_i,\z \rangle
\end{align}
up to error $\|\x-\x^\#\|_2^2 = \mathcal{O} (\rho\sqrt{\log(e/\rho)}) + \mathcal{O} (\beta \sqrt{\log(e/\beta)})$.
%the following for $\T \subset B_2^n$
%\cite[Theorem~1.3]{plan_robust_2013}: if $\A$ is a $m\times n$ Gaussian matrix and 
%$m \gtrsim \rho^{-6} w_*(\T)^2$, then, with high probability, every $\x \in \T \cap \St^{n-1}$
%is recovered by any solution $\x^\#$ of the program
%\begin{align} \label{eq:Plan}
%    \max_{\z \in \T} \frac{1}{m} \sum_{i = 1}^m (q_{\text{corr}})_i \langle \a_i,\z \rangle
%\end{align}
%with error bound 
%\begin{equation}
%    \|\x-\x^\#\|_2^2\leq \rho\sqrt{\log(e/\rho)} + 11 \beta \sqrt{\log(e/\beta)}.
%\end{equation}
%\blue{The vector $\qc\in \{-1,1\}^m$ is any sequence of (corrupted) measurements satisfying 
%\begin{equation}
%    d_H(\qc,\q):=|\{i\in [m]\; : \; (q_{\text{corr}})_i\neq q_i\}|\leq \beta m,
%\end{equation}
%where $q_i$ are the one-bit measurements of type \eqref{eq:OneBitCS}.
%Here, the Hamming distance $d_H(\qc,\q)$ counts the number of bit corruptions.}
If $\T$ is convex, then the program \eqref{eq:Plan} is convex as well.}
While the dependence on the intrinsic complexity of $\T$ is optimal in this result, the dependence on the approximation accuracy $\rho$ is highly suboptimal. 
%Moreover, the assumption that $\A$ is Gaussian cannot be easily relaxed to include more general measurement ensembles such as subgaussian matrices. For example, while Bernoulli matrices provably perform optimally in the case of linear measurements, it is straightforward to show that if $\A$ is a Bernoulli matrix, then there are two distinct $2$-sparse vectors that yield the same one-bit measurements of type \eqref{eq:OneBitCS} independent of the number of measurements $m$. Hence, for this combination of linear measurements and quantization scheme there is no hope that the reconstruction error can decrease with the number of observations, no matter which reconstruction method is used. 
%Nevertheless, 
\blue{In \cite{ai_one-bit_2014} it was shown that the program \eqref{eq:Plan} is even capable of estimating a fixed signal $\x\in\St^{n-1}$ from one-bit measurements \eqref{eq:OneBitCS} if $\A$ is a subgaussian measurement matrix, provided that the $\ell_{\infty}$-norm of $\x$ is small enough. Due to the un-dithered measurement set-up, however, the additional assumption on $\x$ is necessary.
%Moreover, the authors could even show that the program is robust to various forms of measurement noise, such as random bit flips.
On the one hand, the result shows that, apart from very sparse vectors, the program \eqref{eq:Plan} still succeeds at recovering signal vectors from one-bit subgaussian measurements. On the other hand, the vectors have to be normalized, the result is non-uniform, and the dependence on the approximation accuracy $\rho$ is still suboptimal.}

The follow-up work \cite{dirksen_robust_2018} massively improved on all of these points by 
considering dithered one-bit measurements as in 
\eqref{eq:OneBitCSdithered} and adding a regularizing term to \eqref{eq:Plan}.
\blue{For the resulting program 
\begin{align} \label{eq:Dirksen}
    \max_{\z \in \T} \frac{1}{m} \left( \sum_{i = 1}^m (q_{\text{corr}})_i \langle \a_i,\z \rangle \right) - \frac{1}{2\lambda} \| \z \|_2^2
\end{align}
the authors showed the following robust reconstruction guarantee (ignoring log-factors):
if $\ttau \in [-\lambda, \lambda]^m$ is uniformly distributed for a sufficiently large parameter $\lambda$, then
the convex program \eqref{eq:Dirksen} is capable of uniformly recovering all signals $\x\in \T\subset B^n_2$ from corrupted one-bit measurements $\qc$ as in \eqref{eq:hamming_distance_measurements} up to accuracy $\rho$
provided that $\norm{\nu}{L^2}\lesssim 1$, $\beta\lesssim \rho$
and the number of measurements satisfies 
\begin{equation}\label{eq:intro_cond_m}
 m \gg \rho^{-4}  w_*((\T-\T)\cap \rho B^n_2)^2 + \rho^{-2}\log \N(\T,\rho/\sqrt{\log( c/\rho)}).
\end{equation}
Here 
$w_*((\T-\T)\cap \rho B^n_2)^2$ is a localized variant of the Gaussian width and
$\N(\T,\ep)$ denotes the $\eps$-covering number 
%(or metric entropy) 
of $\T$ with respect to the Euclidean distance, i.e., the smallest number of Euclidean balls with radius $\eps$ that are needed to cover $\T$.
Ignoring log-factors, Sudakov's inequality shows that
condition 
\eqref{eq:intro_cond_m}
is already satisfied if
$m\gg \rho^{-4}  w_*(\T)^2$.} This massively improves on earlier uniform reconstruction guarantees. 
Moreover, the result shows that by using dithering in the measurement process there is no need of imposing additional structural assumptions on the signal vectors (such as a small $\ell_{\infty}$-norm).
\blue{However, the result has a drawback as well: it is not sensitive to the magnitude of the additive noise on the linear measurements as long as 
$\norm{\nu}{L^2} \le R$.
In particular, $\norm{\nu}{L^2}\approx 0$
and $\norm{\nu}{L^2}\approx R$ lead to the same reconstruction guarantees.} This suboptimal performance in a low-noise setting is clearly observed experimentally, cf.\ Section \ref{sec:Numerics}.

\subsection{Related Work (Multi-Bit)} 

Compared to the extensive theoretical studies on recovery algorithms in one-bit compressed sensing for memoryless scalar quantization, fewer comprehensive results exist for finer quantization. \blue{One has to understand that a multi-bit quantizer $Q$ with refinement level $\Delta > 0$ leads in compressed sensing to two very different recovery regimes: if the local complexity of $\T$ behaves similar to their global complexity, i.e., $w_*((\T - \T) \cap \rho B_2^n) \simeq \rho w_*(\T)$, the sufficient number of measurements to obtain an approximation error $\rho \gtrsim \Delta$ does not depend on $\rho$ (high quantizer resolution $\leftrightarrow$ un-quantized compressed sensing regime); to obtain smaller approximation errors, the number of measurements must behave similar to the one-bit case (low quantizer resolution $\leftrightarrow$ one-bit compressed sensing regime).} As \cite{laska_regime_2012} shows, it is favorable to increase the bit-depth per measurement if high-accuracy is sought, if the expected noise-level is low, or if the number of measurements underlies stronger restrictions than the number of bits. \blue{Though several articles numerically examined recovery algorithms for multi-bit quantized compressed sensing \cite{jacques_quantized_2013,jacques_dequantizing_2011,shi_methods_2016}, to the best of our knowledge no comprehensive theoretical guarantees covering both regimes were derived for tractable algorithms apart from \cite{dirksen2017one,moshtaghpour_consistent_2016}.  Therein Consistent Basis Pursuit (CBP) is theoretically examined, but the analysis is restricted to signal sets corresponding to atomic norm balls and the obtained guarantees have with $m \gtrsim \Delta^{-2} \rho^{-4} w_\ast (\T)^2$ a far worse error dependence than the one-bit results in \cite{dirksen_robust_2018} would suggest.} More important, CBP becomes infeasible under noise on the measurements. The work \cite{jacques2016error} examines consistent reconstruction which is not tractable in general. The tractable Basis Pursuit De-Noising (BPDN) \cite{candes_near_2006,jacques_dequantizing_2011} only covers the high quantizer resolution regime, i.e., the achievable approximation error $\rho$ is lower bounded by $\Delta$. \blue{The work \cite{xu2018quantized} examines an equivalent variant of \eqref{eq:Dirksen} for multi-bit quantization but only reflects the low quantizer resolution regime. For small $\Delta$, the measurement requirements become suboptimal. Moreover, the requirement $m \gtrsim \rho^{-16} w_\ast (\T)^2$ for general bounded, convex, and symmetric sets is rather pessimistic.} Last but not least, \cite{plan2016generalized,plan2016high} are restricted to Gaussian measurements and treat a more general non-linear adaption of \eqref{eq:CS} which covers \eqref{eq:QCS} as a special case but only leads to non-uniform recovery guarantees.

\subsection{A first intuition} 
\label{sec:FirstIntuition}

Let us compare \eqref{eq:P} to the state-of-the-art methods for robust one-bit quantized compressed sensing presented above. Though at first sight,  \eqref{eq:P} appears to be closely related to \eqref{eq:Plan} and \eqref{eq:Dirksen}, the motivation for the program and its geometric meaning is fundamentally different. Both \eqref{eq:Plan} and \eqref{eq:Dirksen} aim at maximizing the alignment of quantized and unquantized measurements while ignoring the concrete geometry defined by the quantization cells. As already mentioned in \cite{krause2017tractable,iwen2018recovery} one can reformulate \eqref{eq:Plan} 
as
\begin{align} \label{eq:EquivalentForm}
\min_{\z \in \T} &\left( \sum_{i \colon (q_{\text{corr}})_i \neq \sign(\langle \a_i,\z \rangle)} \| \a_i \|_2 \left\| \z-\P_{H_{\a_i}} \z \right\|_2 \; - \sum_{i \colon (q_{\text{corr}})_i = \sign(\langle \a_i,\z \rangle)} \| \a_i \|_2 \left\| \z-\P_{H_{\a_i}} \z \right\|_2 \right) ,
\end{align}
where $\P_{H_{\a_i}}$ denotes the orthogonal projection onto the hyperplane defined by $\a_i$. Consequently, maximizing the alignment in \eqref{eq:Plan} corresponds to punishing wrong measurements $\sign(\langle \a_i,\z \rangle) \neq (q_{\text{corr}})_i$ of a point $\z \in \mathcal{T}$ by its Euclidean distance to the corresponding hyperplane $H_{\a_i}$ and rewarding correct measurements by the same amount. If the measurements are trustworthy, i.e., $(q_{\text{corr}})_i=q_i=\sign(\inner{\a_i}{\x})$ for all $i\in [m]$, 
%and when \eqref{eq:Plan} rates the approximation quality of any $\z \in \mathcal{T}$ to $\x$, 
the rewarding term unnecessarily allows hyperplanes $H_{\a_i}$ not separating $\z$ and $\x$ to influence the penalization in \eqref{eq:Plan} and leads to worse approximation. Having \eqref{eq:EquivalentForm} in mind, the regularizer in \eqref{eq:Dirksen} might be interpreted as a compensation for the rewarding part. In contrast, the ReLU-formulation in \eqref{eq:P} completely drops the rewarding part of \eqref{eq:EquivalentForm} and resembles a continuous proxy of the Hamming-distance on the quantization cell structure (see Figure \ref{fig:L}). In particular, in the noiseless case, that is, 
in the case where $(q_{\text{corr}})_i=\sign(\inner{\a_i}{\x} + \tau_i)$ for all $i\in[m]$, we have $\L_{\qc}(\x)=0$ and therefore any solution $\x^{\#}$ to \eqref{eq:P} has to satisfy $\L_{\qc}(\x^\#)=0$ as well. Since this is equivalent to $\sign(\inner{\a_i}{\x} + \tau_i)=\sign(\inner{\a_i}{\x^\#} + \tau_i)$ for all $i\in [m]$, we see that in the noiseless case minimizing \eqref{eq:P}
forces the reconstructed signal to lie on the correct sides of all shifted hyperplanes. Note that \eqref{eq:Dirksen} completely neglects the knowledge about the hyperplane shifts and thus simply treats the dither as additive noise.\\

\subsection{Outline} We state and explain the main results of the paper, Theorem \ref{thm:noisy_recovery} \& \ref{thm:MultiBitMain}, in Section \ref{sec:MainResults}. The proofs of both results are then provided in Section \ref{sec:Proofs}. Section \ref{sec:Numerics} supports our theoretical findings by numerically comparing \eqref{eq:P} to different competing recovery schemes in both the one-bit and multi-bit setting. The proofs of some technical tools are deferred to the Appendix.

\subsection{Notation} 
We will use the following notation throughout the paper:

\begin{enumerate}
    \item For $k\in \mathbb{N}$ we set $[k]:=\{1, \ldots, k\}$.
    Matrices and vectors are denoted by upper- and lowercase boldface letters, respectively.
    \item For $\x\in \R^n$ we set $\|\x\|_0:=|\{i\in [n]\; : \; x_i\neq 0\}|$. A vector $\x\in \R^n$ is called $s$-sparse if $\|\x\|_0\leq s$. The set of all $s$-sparse vectors in $\R^n$
    is $\Sigma_s^n:=\{\x\in \R^n \; : \; \|\x\|_0\leq s\}$.
    Given $p\geq 1$, the $\ell_p$-norm of $\x$ is denoted by $\norm{\x}{p}$ and the associated unit ball is $B^n_p$.
    The Euclidean unit sphere in $\R^n$ 
    is $\S^{n-1}$. Further, for a subset $S\subset \R^n$ we define $d_2(S):=\sup_{\x\in S}\norm{\x}{2}$.
    \item The (unnormalized) Hamming distance between vectors $\x, \y\in \R^n$ is $d_H(\x,\y):=|\{i\in [n]\; : \; x_i\neq y_i\}|$.
%    \item $\B_\eps(\z)$ $\ell_2$-ball of radius $\eps > 0$, centered at $\z$
    \item The $\sign$-function acts componentwise on a vector and we set $\sign(0) := 1$.
    \item The Gaussian width of a set $\T\subset \R^n$ is denoted by 
    \begin{displaymath}
 w_*(\T) := \E\sup_{\x \in \T} |\inner{\g}{\x}|,
\end{displaymath}
where $\g$ is an $n$-dimensional standard Gaussian random vector.

\item For $\eps>0$, the $\eps$-covering number of $\T\subset \R^n$ is denoted by $\mathcal{N}(\T, \eps)$. It is the smallest number of Euclidean balls with radius $\eps$ needed to cover $\T$.  
    \item 
    For $p\geq 1$, the $L^p$ norm of a random variable $X$ will be denoted by $\norm{X}{L^p}=(\E |X|^p)^{1/p}$.
    Further, $X$ is subgaussian if its subgaussian norm 
    \begin{displaymath}
 \norm{X}{\psi_2}:=\inf\{t>0\; : \; \E \exp(X^2/t^2)\leq 2\}
\end{displaymath}
    is finite. \blue{In particular, $X$ satisfies the tail bound
    \begin{equation}
        \Pr(|X| \geq t) \leq 2\exp(-ct^2/\norm{X}{\psi_2}^2)
        \;,
    \end{equation}
    which holds for every $t>0$ and an absolute constant $c>0$.
    }
\item The letters $C,c>0$ (possibly with subscripts, that is, $C_i,c_i$) will always denote constants which may only depend on the subgaussian parameter $L$. We write $A\lesssim B$ if $A\leq C B$ for a constant $C$ (respectively
$A\gtrsim B$ if $A\geq c B$ for a constant $c$). Finally, we use the abbreviation 
$A\sim B$ if both $A\lesssim B$ and $A\gtrsim B$.
\end{enumerate}

%%%%%%%%%%%%%%%%%%%%%%%%%%%%%%%%%%%
%%% Main Results
%%%%%%%%%%%%%%%%%%%%%%%%%%%%%%%%%%%

\section{Main Results} \label{sec:MainResults}

Let us begin by stating the main results of the paper. We split this section into two parts, one containing the results for one-bit quantization and one discussing the more general multi-bit quantization setting. In both settings, we assume that
the linear measurements of a vector $\x$ prior to the quantization step are of the form
\begin{equation}
    \A \x + \ttau+\nnu\in \R^m,
\end{equation}
where
\begin{itemize}
    \item the rows $\a_i^T$ of the measurement matrix $\A$ consist of independent and identically distributed copies of an isotropic, symmetric and $L$-subgaussian random vector $\a\in \R^n$. Recall that
    a random vector $\a\in \R^n$ is said to be isotropic if 
$\norm{\inner{\a}{\x}}{L^2}=\norm{\x}{2}$ for all $\x\in \R^n$. Further, $\a$ is $L$-subgaussian if 
$\norm{\inner{\a}{\x}}{L^p} \leq  L  \sqrt{p} \norm{\inner{\a}{\x}}{L^2}$ for all $\x \in \R^n$ and $p\geq 1$.
Equivalently (up to absolute constants), this means that for every $\x \in \R^n$ the subgaussian norm of $\inner{\a}{\x}$ is bounded by $L \norm{\inner{\a}{\x}}{L^2}$.
    \item $\ttau\in \R^m$ is a random vector with entries $\tau_i$ that are independent copies of a random variable $\tau \sim \unif([-\la, \la])$ for a parameter $\la>0$. 
    \item $\nnu\in \R^m$ denotes a random vector with entries $\nu_i$ that are independent copies of a mean-zero random variable $\nu$ which is $L$-subgaussian, i.e., $\norm{\nu}{L^p}\leq L \sqrt{p} \norm{\nu}{L^2}$ for every $p\geq 1$. Again, equivalently up to absolute constants this means that $\norm{\nu}{\psi_2}$ is bounded by $ L\norm{\nu}{L^2}$. 
\end{itemize}
We assume that the random vectors/matrices $\A, \ttau, \nnu$ are independent. In contrast to $\A$ and $\ttau$, the noise vector $\nnu\in \R^m$ is not known to us.
\vskip 0.1cm
In the following, in order to enhance readability, we will always suppress the dependency of constants on the subgaussian parameter $L$. That is, if we speak of a constant $C$, then it is either a numerical constant or it is a constant which only depends on $L$. In a similar fashion, when we write $\gtrsim$ (or $\lesssim$) then we mean that the inequality holds for a constant that may only depend on $L$.

%%%%%%%%%%%%%%%%%%%%%%%%%%%%%%
%%% One-Bit
%%%%%%%%%%%%%%%%%%%%%%%%%%%%%%

\subsection{One-Bit Quantization} \label{sec:OneBit}

As already mentioned above, we are interested in recovering high-dimensional signal vectors $\x \in \R^n$ from possibly corrupted one-bit measurements $\qc\in \{-1,1\}^m$ which satisfy 
\begin{equation*}
%\label{eq:hamming_distance_measurements}
    d_H(\qc, q(\x))\leq \beta m,
\end{equation*}
where $q(\x)= \sign(\A \x + \ttau+\nnu)\in \{-1,1\}^m$.
Hence, in this model we permit that up to $\beta m$ bits are arbitrarily (possibly adversarially) corrupted.\\

The following theorem is our main recovery result in the one-bit case. 
\begin{comment}
The theorem provides theoretical reconstruction guarantees which nearly match those given in \cite[Theorem 1.5]{dirksen_robust_2018} for the untractable Hamming distance program in the setting of low additive noise, that is, $\norm{\nu}{L^2}\approx 0$ (which is also present in the Hamming distance result in \cite[Theorem 1.5]{dirksen_robust_2018}). In the high-noise regime,
the program performs like the convex program \eqref{eq:Dirksen} and our guarantees deteriorate to
the guarantees in Theorem~\ref{thm:Comparison}.
\end{comment}
\begin{theorem} \label{thm:noisy_recovery}
There are constants $c,c_0, c_1, c_2, c_3>0$ and $C\geq e$ such that the following holds.
Let $\T\subset R B^n_2$ denote a convex set. Fix an approximation accuracy $\rho \in (0, R]$.
\begin{enumerate}
    \item[(i)] Low-noise regime: 
    if $\norm{\nu}{L^2}\leq c_0 \rho/\sqrt{\log(C\lambda/\rho)}$,
    then the following holds.
    Suppose the dithering parameter satisfies $\la \gtrsim R$, the number of measurements satisfies
    \begin{equation}
       m \gg \frac{\lambda}{\rho} \left(\frac{w_*((\T-\T)\cap \rho B^n_2)^2}{\rho^2} 
       +\log \N(\T,c \rho/\sqrt{\log(C\lambda/\rho)})\right),
    \end{equation}
    and $\beta \in (0,1)$ is a parameter such that $\beta\log(e/\beta)\leq c_2 \rho/ \la $. Then,
    with probability exceeding 
    $
     1- 2\exp(- c_3 m \rho / \lambda)
    $,
    the following holds true: for all $\x\in \T$ and all bit sequences $\qc \in \{-1,1\}^m$
    with $d_H(\qc, q(\x)) \leq \beta m$, every minimizer $\x^\#$ of the program \eqref{eq:P} 
    satisfies $$\norm{\x-\x^\#}{2}\leq \rho.$$
    \item[(ii)] High-noise regime: 
    if $\norm{\nu}{L^2}\geq c_0 \rho/\sqrt{\log(C\lambda/\rho)}$, then the following holds.
    Suppose the dithering parameter satisfies $\la\gtrsim (R+\norm{\nu}{L^2})\sqrt{\log(\la/\rho)}$, the number of measurements satisfies
    \begin{equation}
       m \gg \left(\frac{\lambda}{\rho}\right)^2 \left(\frac{w_*((\T-\T)\cap \rho B^n_2)^2}{\rho^2} 
       +\log \N(\T,c \rho/\log(C\lambda/\rho))\right),
    \end{equation}
    and $\beta \in (0,1)$ is a parameter such that $\beta\log(e/\beta)\leq c_2 \rho/ \la $.
    Then, with probability exceeding 
    $
     1- 2\exp(- c_3 m (\rho / \lambda)^2 )
    $,
    the following holds true: for all $\x \in \T$ and all bit sequences $\qc \in \{-1,1\}^m$
    with $d_H(\qc, q(\x)) \leq \beta m$, every minimizer $\x^\#$ of the program \eqref{eq:P} 
    satisfies $$\norm{\x-\x^\#}{2}\leq \rho.$$
\end{enumerate}
\end{theorem}
\paragraph{Two regimes} The result shows that the performance of the program \eqref{eq:P} 
depends on the ratio of the noise level $\norm{\nu}{L^2}$ and the reconstruction accuracy $\rho>0$. As long as (ignoring log-factors)
$\rho \gtrsim \norm{\nu}{L^2}$, the performance of \eqref{eq:P} is comparable to the performance of the (non-tractable) Hamming distance minimization program
(see \cite[Theorem 1.5]{dirksen_robust_2018}),
\begin{equation}\label{eq:hamming_program}
    \min_{\z \in \R^n} d_H(\qc, \sign(\A \z + \ttau)) \quad
    \text{subject to} \quad \z \in \T \; .
\end{equation}
Hence, in this accuracy regime, the program \eqref{eq:P} can be viewed as a convex proxy for the
Hamming distance minimization program and achieves near-optimal reconstruction guarantees (see the discussion on optimality below Theorem~\ref{thm:MultiBitMain}).
\blue{ Moreover, similar to the program \eqref{eq:Dirksen} but in contrast to \eqref{eq:hamming_program}, the program
\eqref{eq:P} also achieves near-optimal reconstruction accuracies well below the noise level (for optimality see \cite[Section 5]{dirksen_robust_circulant_2018})}. However, this comes at the cost of a worse rate-distortion relationship and a worse probability of success.

\paragraph{Sparse Recovery} \blue{For sparse recovery (i.e., there are no sparsity defects on the signal) from non-adaptive one-bit measurements, it has been shown in~\cite{jacques_robust_2013} that the optimal error decay rate is $\mathcal{O}(\tfrac{1}{m})$ while, to the author's knowledge, the best proven rate for a tractable program is $\mathcal{O}(\frac{1}{\sqrt{m}})$, see \cite[Theorem 6]{dirksen2019quantized}. In practice, however, one has to deal with sparsity defects. Here, a commonly considered prior is the set of $s$-compressible signals given by $\T=\sqrt{s}B_{1}^n \cap B_{2}^n$ for which Theorem \ref{thm:noisy_recovery} can be applied. 
%in which case the guarantees provided in~\cite[Theorem 6]{dirksen2019quantized} deteriorate to $\mathcal{O}(m^{-\frac{1}{4}})$.
}
Since (see \cite[Lemma 3.1]{plan_one-bit_2013})
\begin{equation}
    \conv(\Sigma_s^n\cap B^n_2)\subset \sqrt{s}B_{1}^n \cap B_{2}^n\subset 2\conv(\Sigma_s^n\cap B^n_2),
\end{equation}
$\T$ is a proxy for the convex hull of all $s$-sparse vectors in the Euclidean unit ball. Using $w_\ast(\sqrt{s}B_{1}^n \cap B_{2}^n)^2\sim s\log(2n/s)$ and Sudakov's inequality, we can deduce from 
Theorem~\ref{thm:noisy_recovery} that if $\norm{\nu}{L^2}\lesssim 1$ and $\beta=0$, then
any solution $\x^\#$ satisfies $\norm{\x-\x^\#}{2}\leq \rho$ if
\begin{itemize}
    \item $m \gtrsim \rho^{-3}\log(\rho^{-1})s \log(2n/s)$
    provided that 
    $\rho /\sqrt{\log(1/\rho)} \gtrsim \norm{\nu}{L^2}$ and we choose $\la\sim 1$,
    \item $m \gtrsim \rho^{-4}\log^3(\rho^{-1})s \log(2n/s)$
    provided that 
    $\rho /\sqrt{\log(1/\rho)} \lesssim \norm{\nu}{L^2}$ and we choose $\la\sim \sqrt{\log(\rho^{-1})}$.
\end{itemize}
In words, for reconstruction accuracies 
above the noise floor, the reconstruction error essentially decays as
$\mathcal{O}\Big(\big(\frac{s\log(2n/s)}{m}\big)^{1/3}\Big)$
if the dithering random variables $\tau_i$ are uniformly distributed on the interval $[-\la, \la]$ for $\la$ a constant that only depends on $L$. If we want to achieve reconstruction accuracies
$\rho$ below the noise floor, then we have to increase $\la\sim \sqrt{\log(\rho^{-1})}$. In this case, the error decays as  $\mathcal{O}\Big(\big(\frac{s\log(2n/s)}{m}\big)^{1/4}\Big)$ \blue{ which has previously been the best known guarantee for recovery of compressible signals from one-bit measurements even in the noiseless setting (see the discussion after \cite[Theorem 1.7]{dirksen_robust_2018}). Whenever the expected noise level $\norm{\nu}{L^2}$ is unknown, practitioners can choose $\la \sim \sqrt{\log(\rho^{-1})}$ to have guaranteed approximation in both regimes thus paying an additional log-factor in the number of measurements if the noise is small.} %Finally, let us mention that the guarantees \cite[Theorem 6]{dirksen2019quantized} deteriorate to $\mathcal{O}(m^{-\frac{1}{4}})$ when considering sparsity defects.}

\subsection{Multi-Bit Quantization} \label{sec:MultiBit}

%\begin{figure}
%    \centering
%    \includegraphics[width=0.8\textwidth]{Figures/Multi-Bit_Sketch.pdf}
%    \caption{$3$-Bit quantizer \red{(put to TikZ)}}
%    \label{fig:multibit}
%\end{figure}
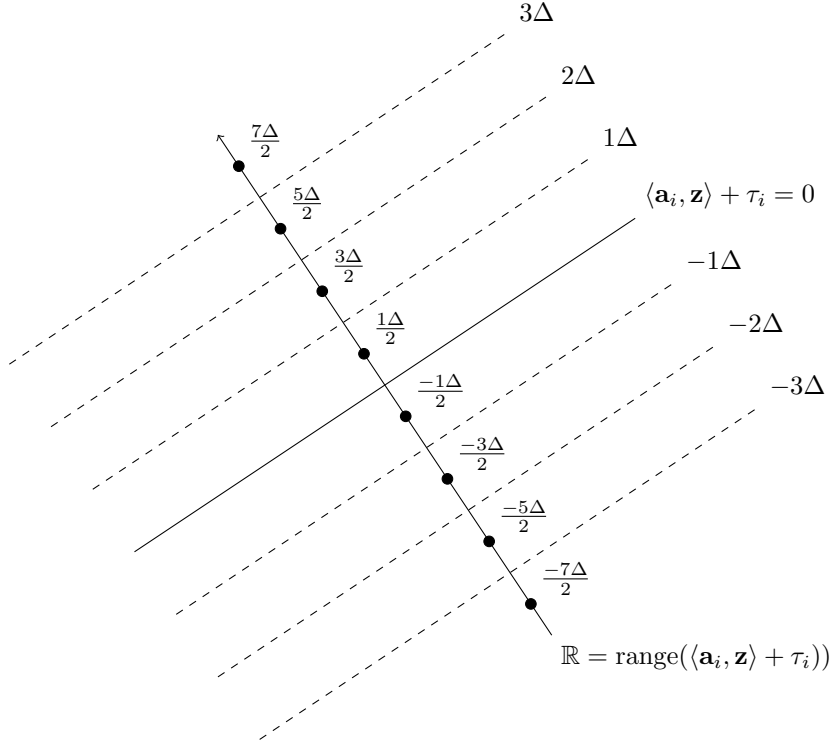
\begin{figure}
    \centering
    	\begin{tikzpicture}
	% planes
	%\draw[tumblue] (-2,.4) -- (2,-.4); % a1 0 plane
	%\draw[tumblue] (-1.8,1.4) -- (2.2,.6); % a1 +1 plane
	%\draw[tumblue] (-2.2,-.6) -- (1.8,-1.4); % a1 -1 plane
	%
	%\draw[tumred] (1.5,1.4) -- (-1.5,-1.4); %a2 0 plane
	%\draw[tumred] (0.8,2.15) -- (-2.2,-0.65); %a2 -1 plane
	%\draw[tumred] (2.2,.65) -- (-0.8,-2.15); %a2 +1 plane
	% aj
	%\draw[->] (0,0) -- (.4,2) node[right]{$\A_1$};
	%\draw[->] (0,0) -- (1.4,-1.5) node[right]{$\A_2$};
	% axes
	\def \aa {0.83205}
    \def \b {0.5547}    % (3,2) vector normalized
    \def \scale {4.0}
    \def \d{1}
	\draw (-\scale*\aa,-\scale*\b) -- (\scale*\aa,\scale*\b) node [above right] {$\inner{\a_i}{\z} +\tau_i = 0$};
	\draw[<-] (-\scale*\b, \scale*\aa) -- (\scale*\b,-\scale*\aa) node [below right] {$\R = \text{range} (\inner{\a_i}{\z} +\tau_i))$};
    \foreach \niter in {-3,-2,-1,1,2,3} {
        \draw[dashed] (-\scale*\aa -\niter*\b ,-\scale*\b + \niter*\aa ) -- (\scale*\aa-\niter*\b,\scale*\b+\niter*\aa) node[above right] {$\niter \Delta$};
      }
    \foreach \xiter in {-7,-5,-3,-1,1,3,5,7} {
        \filldraw (-\xiter/2*\b,\xiter/2*\aa) circle (2pt) node[above right] {$\frac{\xiter \Delta}{2}$};
        }
	% x
	%\filldraw[tumred] (-.36,.9) circle (2pt) node[right] {$x^n$};
	
	\end{tikzpicture}
    \caption{Two ways of characterizing a uniform $3$-Bit quantizer: either by the separating hyperplanes as in \eqref{eq:measMulti} or by the quantization points in $\mathcal{A}_{\Delta,B}^m$ onto which the real measurements $\langle \a_i,\x \rangle + \tau_i$ are projected.}
     \label{fig:multibit}
\end{figure}

In the $B$-bit case we generalize \eqref{eq:OneBitCSdithered} to uniform multi-bit measurements of the form
\begin{align} \label{eq:measMulti}
    Q(\x) \in \{-1,1\}^{m\times (2^B-1)}, \quad Q(\x)_{i,j} = \sign(\langle \a_i,\x \rangle + (\tau_i + j\Delta)),
\end{align}
where we use the index convention $i \in [m], j \in \{-(2^{B-1}-1),...,(2^{B-1}-1)\}$ and $\tau_i \sim \unif([-\Delta,\Delta])$ (cf.\ separating hyperplanes in Figure \ref{fig:multibit}). For simplicity, we do not consider noise on the measurements here. A non-uniform quantizer could be characterized by replacing $j \Delta$ with general shifts $\Delta_j$. As uniform quantizers asymptotically yield optimal performance \cite{gray1998quantization}, we restrict ourselves to this conceptually clearer setting. Our loss function generalizes in a straight-forward way to
\begin{align*}
    \L_{Q(\x)}(\z) = \frac{1}{m} \sum_{i = 1}^m \sum_{j = -(2^{B-1}-1)}^{(2^{B-1}-1)} \left[ -Q(\x)_{i,j} (\langle \a_i,\z \rangle + (\tau_i + j\Delta)) \right]_+,
\end{align*}
leading to the multi-bit ReLU recovery program
\begin{equation} \label{eq:Pmulti}
\tag{$P_{Q}$}
 \min_{\z \in \R^n} \L_{Q(\x)}(\z) \quad \text{subject to} \quad  \z\in \T.
\end{equation}
Apparently, writing the multi-bit measurements as above embeds them into the much larger space $\{-1,1\}^{m\times (2^B-1)}$. Since $Q(\x)_{i,j} \le Q(\x)_{i,k}$, for $j \le k$, only a small part of the possible sequences in $\{-1,1\}^{m\times (2^B-1)}$ are attained, namely those which consist of $-1$ up to some index $j_0$ and stay $1$ for any index $j > j_0$. 
\vskip 0.1cm
As an alternative way of representing the quantized measurements in the multi-bit case we can set
\begin{equation}
    q(\x)_i=q_{\mathcal{A}_{\Delta,B}}(\inner{\a_i}{\x}+\tau_i), \quad i=1, \ldots, m,
\end{equation}
where
$q_{\mathcal{A}_{\Delta,B}}: \R 
\to \mathcal{A}_{\Delta,B}$ denotes the 
one-dimensional quantizer
defined by 
\begin{equation}\label{eq:Alphabet_Quantizer}
    q_{\mathcal{A}_{\Delta,B}}(x):=\text{argmin}_{l\in \mathcal{A}_{\Delta,B}}|x-l|
\end{equation}
and
$\mathcal{A}_{\Delta,B}$ is the $B$-bit quantization alphabet with resolution $\Delta>0$ given by
\begin{equation}\label{eq:Alphabet}
\mathcal{A}_{\Delta,B}=\left\{ -(2^{B}-1)\frac{\Delta}{2},...,-3\frac{\Delta}{2},-\frac{\Delta}{2},\frac{\Delta}{2},3\frac{\Delta}{2},...,(2^{B}-1)\frac{\Delta}{2} \right\}\subset \R.
\end{equation}
For $j\in \mathbb{Z}$ set $q_j=(2j+1)\frac{\Delta}{2}$. Then
$$\mathcal{A}_{\Delta,B}=\{q_j\; : \;
j\in \{-2^{B-1},\ldots, -2, -1,0,1,2,\ldots, 2^{B-1}-1\}\}$$
and for $j\in \{-2^{B-1}+1,\ldots, -2, -1,0,1,2,\ldots, 2^{B-1}-2\}$,
\begin{align*}
    q(\x)_i = q_j=\frac{j\Delta + (j+1)\Delta}{2} \iff (\langle \a_i,\x \rangle + \tau_i) \in (j\Delta,(j+1)\Delta).
\end{align*}
In words, instead of characterizing the multi-bit quantizer by its separating hyperplanes, one characterizes it by the $2^B$ centers of the quantization intervals in between (cf.\ Figure \ref{fig:multibit}). While the first representation allows a nice geometrical and general interpretation of our functional $\L$ (all additional separating hyperplanes have the same influence as the original one-bit hyperplane), the second one allows straight-forward analysis. In any case, both representations can be related via the simple bijection
\begin{align*}
    q(\x) = \left( Q(\x) \cdot \mathbbm{1} \right) \frac{\Delta}{2}
\end{align*}
where $\mathbbm{1} \in \R^{2^B-1}$ is a vector whose entries are $1$. 
\begin{remark}
\blue{ Note that the second quantization representation gives rise to an alternative formulation of $\L_{Q(\x)}$ which does not require an exponentially growing number of summands. In fact, the inner sum of $\L_{Q(\x)}$ adds up the distances between $(\langle \a_i,\z \rangle + \tau_i)$ and the hyperplanes separating $\x$ and $\z$. One can thus deduce, cf.\ Figure \ref{fig:multibit}, that if
\begin{align*}
    K = K(\z,q(\x),\a_i) = 1 + \left\lfloor \frac{|\langle \a_i,\z \rangle - q(\x)_i | - \frac{\Delta}{2}}{\Delta} \right\rfloor
\end{align*}
denotes the number of hyperplanes separating $\x$ and $\z$ in measurement $\a_i$ ($\lfloor r \rfloor$ is the largest number in $\mathbb{Z}$ less or equal to $r \in \mathbb{R}$), then 
\begin{align*}
    \L_{Q(\x)} (\z) &= \L_{q(\x)} (\z) 
    =  \frac{1}{m} \sum_{i=1}^m \sum_{k = 1}^{K} \left( |(\langle \a_i,\z \rangle - q(x)_i| - \left( k - \frac{1}{2} \right) \Delta \right) \\
    &= \frac{1}{m} \sum_{i=1}^m \left( K \left( |(\langle \a_i,\z \rangle - q(x)_i| + \frac{\Delta}{2} \right) - \frac{K (K-1)}{2} \Delta \right).
\end{align*}
This representation of $\L_{Q(x)}$ as a quadratic function of $K(\z,q(\x),\a_i)$ is of special interest for fine quantization, i.e., whenever $B$ becomes large.}
\end{remark}

We are ready to state the main theorem. As mentioned in the introduction, we have to distinguish two regimes in the multi-bit setting. If the aimed for approximation accuracy $\rho > 0$ is above the quantizer resolution $\Delta$, we expect the sufficient number of measurements for signal sets $\T$ whose local complexity behaves similar to their global complexity, \emph{i.e.} $w_*((\T - \T) \cap \rho B_2^n) \simeq \rho w_*(\T)$, to be independent of $\rho$ as in the classical compressed sensing theory. On the other hand, if we ask for $\rho < \Delta$, then the number of measurements should behave more like in the one-bit case and depend explicitly on $\rho$. The following theorem guarantees in both regimes reconstruction of unknown signals from measurements of type \eqref{eq:measMulti} by \eqref{eq:Pmulti}.

\begin{theorem} \label{thm:MultiBitMain} 
There exist constants $\Gamma\geq 1, c, c_1, c_2\in (0,1)$ 
and a numerical constant $C\geq 1$
such that the following holds. 
Let $\T\subset R B^n_2$ be a convex set. Assume that $\Delta > 0$ and $B\in \mathbb{N}$ are chosen such that $\Gamma R\leq (2^{B-1}-1)\Delta - \frac{\Delta}{2}$ (meaning that the quantizer's range includes w.h.p.\ most of the linear measurements). 
\begin{enumerate}{}
    \item[(i)] For any $\rho\geq C\Delta$, if
    \begin{align*}
         m \gtrsim \rho^{-2} w_*((\T-\T)\cap \rho B^n_2)^2 + R^{-2} w_*(\T)^2,
    \end{align*}
    then, with probability exceeding 
    $
     1- 2\exp(- c_1 m)
    $,
    the following holds: for all $\x\in \T$, 
    every minimizer $\x^\#$ of the program \eqref{eq:Pmulti} 
    with $Q(\x)_{i,j} = \sign (\langle \a_i,\x \rangle + \tau_i + j\Delta)$
    satisfies $$\norm{\x-\x^\#}{2}\leq \rho.$$
    \item[(ii)] For any $0 < \rho < \Delta$, if
    \begin{align*}
        m \gtrsim \frac{\Delta}{\rho} \Big(\frac{w_*((\T-\T)\cap \rho B^n_2)^2}{\rho^2} + \log(\N(\T, c\rho/ \sqrt{\log(e\Delta/\rho)})) \Big),
    \end{align*}
    then, with probability exceeding 
    $
    1- 2\exp(-c_2 m(\rho/\Delta))
    $,
    the following holds: for all $\x\in \T$, 
    every minimizer $\x^\#$ of the program \eqref{eq:Pmulti} 
    with $Q(\x)_{i,j} = \sign (\langle \a_i,\x \rangle + \tau_i + j\Delta)$
    satisfies $$\norm{\x-\x^\#}{2}\leq \rho.$$
\end{enumerate}
\end{theorem}

\paragraph{Two regimes} The result highlights the role of the resolution $\Delta$ as a parameter of $q_{\mathcal{A}_{\Delta,B}}$.
The statistical guarantees for the estimator $\x^\#$ computed by \eqref{eq:Pmulti} are different depending on how the desired accuracy $\rho>0$ relates to the resolution $\Delta$. For the sake of simplicity assume for the moment that $\T \subset B_2^n$ and observe that by Sudakov's inequality $\log(\mathcal{N}(\T,\ep)) \lesssim \rho^{-2}
\log(e\Delta/\rho) w_*(\T)^2$
for $\eps\sim \rho/ \sqrt{\log(e\Delta/\rho)}$. Further, we have $ w_*((\T-\T)\cap \rho B^n_2) \leq 2w_*(\T)$.  Theorem~\ref{thm:MultiBitMain} yields the following performance bounds for \eqref{eq:Pmulti} (ignoring log-factors):
\begin{itemize}
    \item If $\rho \geq C \Delta$, then with probability $1-\exp(-c_1 m)$ the estimator $\x^\#$ satisfies
    \begin{equation}
        \norm{\x - \x^\#}{2} \leq c' \Big( \frac{w_*(\T)^2}{m} \Big)^{\frac{1}{2}} \; .
    \end{equation}
    \item If $0<\rho < \Delta$, then with probability $1- 2\exp(-c_2 m(\rho/\Delta))$ the estimator $\x^\#$ 
    satisfies
    \begin{equation}
    \label{eq:mbit_perfomance_simple}
        \norm{\x - \x^\#}{2} \leq c'' \Big( \frac{\Delta w_*(\T)^2}{m} \Big)^{\frac{1}{3}} \; .
    \end{equation}
\end{itemize}
The first bound resembles a classical compressive sensing bound and is optimal in this regard (see the comments
on optimality below). If we try to reach accuracies below the resolution of the quantizer, then the performance of the estimator $\x^\#$ deteriorates to the performance in the one-bit setting (cf. Theorem~\ref{thm:noisy_recovery}).

\paragraph{Bit budgets} Let us  compare the results in Theorem~\ref{thm:noisy_recovery} and Theorem~\ref{thm:MultiBitMain}
with regard to the number of bits necessary in order to achieve a certain accuracy $\rho > 0$ in the noiseless setting. We denote the
smallest number of bits which have to be collected in order to achieve $\norm{\x - \x^\#}{2} \leq \rho$ by $\mathfrak{B}(\rho)$. For
$m$ measurements and
$B$-bits per measurement  
the total number of collected bits is $Bm$.
For the two regimes presented by Theorem~\ref{thm:MultiBitMain} only part (ii) is a
fair comparison to the results obtained by Theorem~\ref{thm:noisy_recovery}, since the bounds for $\rho \geq C \Delta$ in Theorem~\ref{thm:MultiBitMain} part (i) clearly outperform the result in Theorem~\ref{thm:noisy_recovery}. 
Hence, let us assume that $0< \rho < \Delta$. Since we can choose $\Delta\sim \frac{R}{2^B}$,   \eqref{eq:mbit_perfomance_simple} shows that the multi-bit 
estimator $\x_{\text{mbit}}^\#$ satisfies
(up to log-factors)
\begin{equation*}
        \norm{\x - \x_{\text{mbit}}^\#}{2} \leq c'' \Big( \frac{Rw_*(\T)^2}{2^{B} m} \Big)^{\frac{1}{3}} \; .
\end{equation*}
In comparison, the one-bit estimator $\x_{\text{1bit}}^\#$ achieves (up to log-factors)
\begin{equation*}
        \norm{\x - \x_{\text{1bit}}^\#}{2} \leq c' \Big( \frac{R w_*(\T)^2}{m} \Big)^{\frac{1}{3}} \; .
\end{equation*}
Hence, we find
\begin{equation} \label{eq:BitBudgetComparison}
    \mathfrak{B}_{\text{1bit}}(\rho) \gtrsim \rho^{-3} R\, w_*(\T)^2 
    \quad \text{ and } \quad
    \mathfrak{B}_{\text{mbit}}(\rho) \gtrsim B\, 2^{-B}\rho^{-3} R\, w_*(\T)^2  \; .
\end{equation}
%If $B \cdot  2^{-B} \lesssim \rho$, the performance with respect to the number of bits, which need to be aquired in order to have accuracy $\rho$, of $\x_{\text{mbit}}^\#$ becomes close the performance for $\rho \geq C_0 \Delta$. This is not surprising, since in this regime $\Delta \simeq \rho$ (up to logrithmic factors). 
Since the function $B \mapsto B\, 2^{-B}$ is rapidly decreasing, the comparison shows that spending more bits on a single measurement improves the overall bitrate which is necessary to attain accuracy $\rho >0$ (in the noiseless setting).

\paragraph{Near-optimality of Theorem~\ref{thm:MultiBitMain} \blue{for general convex sets}} Above the quantizer resolution, that is, for $\rho\gtrsim \Delta$, 
Theorem~\ref{thm:MultiBitMain} shows that with high probability $\sup_{\x\in B^n_1}\norm{\x-\x^\#}{2}\leq \rho$ if $m\gtrsim \rho^{-2} w_*(B_1^n)^2$. From \cite[Corollary 10.6]{foucart_mathematical_2013} it follows that for any reconstruction map $\mathcal{R}$
and matrix $\A \in \R^{m \times n}$, if
$\sup_{\x\in B^n_1}  \norm{\x - \mathcal{R}(\A \x)}{2} \leq \rho$,
then $m \gtrsim \rho^{-2} \log(en/m)$. Hence, if $m\leq \sqrt{n}$, then the number of measurements is necessarily lower-bounded by $\rho^{-2} w_*(B_1^n)^2$, 
which shows optimality of \eqref{eq:Pmulti} for reconstruction accuracies above the quantizer resolution. To see optimality of the second statement (up to log-factors), consider the convex set $\T=E\cap B^n_2$ where $E\subset \R^n$ is a $k$-dimensional subspace. Here, $R=1$ which implies that  
we can pick $\Delta\sim 2^{-B}$.
For any $\eps\in (0,1)$,  $\log(\N(\T,\ep))\lesssim k \log(1/\eps)$ and 
$(\T-\T)\cap \eps B^n_2\subset \eps (E\cap B^n_2)$. Consequently, for $\eps \sim \rho/\sqrt{\log(e\Delta/\rho)}$
if 
$m\gtrsim 2^{-B} \rho^{-1}\log(\sqrt{\log(e\Delta/\rho)}/\rho) k$ then with high probability $\sup_{\x\in \T}\norm{\x-\x^\#}{2}\leq \rho$. On the other hand, if $\mathcal{R}$ is any reconstruction map and $\mathcal{Q}:\R^m \to \mathcal{A}^m$ is any memoryless scalar $B$-bit quantizer ($|\mathcal{A}|=2^B$) such that $\sup_{\x\in \T}\norm{\x-\mathcal{R}(\mathcal{Q}(\A\x))}{2}\leq \rho$, then $\rho \gtrsim \frac{k}{m 2^B}$ (e.g. see \cite{boufounos_quantization_2015}).

\blue{\paragraph{Sparse recovery} 
Let us briefly comment on the case of sparse recovery for reconstruction accuracies below the quantizer resolution. As can be seen from \eqref{eq:mbit_perfomance_simple}, in the case of $s$-compressible signals $\T=\sqrt{s}B_{1}^n \cap B_{2}^n$, the reconstruction error decays as $\mathcal{O}\Big( \frac{\Delta s\log(2n/s)}{m} \Big)^{\frac{1}{3}}$. To the authors best knowledge this is the best known error decay rate for compressible signals. We remark, however, that in the case of $s$-sparse signals there is a tractable program with reconstruction error decaying as $\mathcal{O}(m^{-1/2})$, see \cite[Section 7.3]{xu2018quantized}. For further discussion we refer the reader to the survey \cite{dirksen2019quantized}.
}

\paragraph{Noise} \blue{If one is interested in treating noise in Theorem \ref{thm:MultiBitMain}, this is surely possible using the tools presented here. However, it requires some additional technical work like clarifying the concept of post-quantization noise and analysing the interplay between quantizer resolution, additive noise, dither range, and reconstruction accuracy in multiple cases.}

%%%%%%%%%%%%%%%%%%%%%%%%%%%%%%%%%%%
%%% Proofs
%%%%%%%%%%%%%%%%%%%%%%%%%%%%%%%%%%%

\section{Proofs} \label{sec:Proofs}

In this section, we provide the proofs omitted in Section \ref{sec:MainResults}. We first focus on the one-bit setting of Section \ref{sec:OneBit} and then turn to the more general multi-bit setting. 

%%%%%%%%%%%%%%%%%%%%%%%%%%%%%%%%%%%%%%%%%%%%%%%%%%%%%%%%%%%%%%%%%%%%%%%%%%%%%%
%%% One Bit
%%%%%%%%%%%%%%%%%%%%%%%%%%%%%%%%%%%%%%%%%%%%%%%%%%%%%%%%%%%%%%%%%%%%%%%%%%%%%%

\subsection{Proof of Theorem \ref{thm:noisy_recovery}}
To facilitate reading, we organize the proof in the following way:
\begin{itemize}
    \item \textbf{Geometric tools.} In Section \ref{sec:ToolsSjoerd} we recap the theory of non-Gaussian hyperplane tesselations as initially developed in \cite{dirksen_robust_2018}.
    \item \textbf{Technical supplement for the low-noise regime.} In Section \ref{sec:ToolsLowNoise} we develop the necessary technical tools to prove the first part of Theorem \ref{thm:noisy_recovery}. They require that $\rho/\sqrt{\log(\la/\rho)} \geq c  \norm{\nu}{L^2}$ but lead to sharper bounds in the end.
    %The low-noise regime is characterized by the fact that the accuracy required in Theorem~\ref{thm:noisy_recovery} satisfies $\rho/\sqrt{\log(\la/\rho)} \geq c  \norm{\nu}{L^2}$. In this regime one can obtain a sharper recovery guarantee. 
    \item \textbf{Technical supplement for the high-noise regime.} In Section \ref{sec:ToolsHighNoise} we develop alternative technical tools which can handle the case where $\rho/\sqrt{\log(\la/\rho)} \leq c  \norm{\nu}{L^2}$. The main tool (Theorem~\ref{thm:concencration_functional}) of this section is developed along the lines of the proof of \cite[Theorem 4.5]{dirksen_robust_2018} with additional technicalities arising from the fact that $\L_{q(\x)}(\z)$ is not linear in $\z$.
    \item \textbf{Proof of Theorem \ref{thm:noisy_recovery}.} In Section \ref{sec:ProofOneBit} we prove Theorem \ref{thm:noisy_recovery} in both regimes.
\end{itemize}

Though we first discuss the technical results in detail, for better understanding it might help to have a look at Section \ref{sec:ProofOneBit} before reading Section \ref{sec:ToolsLowNoise} \& \ref{sec:ToolsHighNoise}.

\subsubsection{Geometric Tools} \label{sec:ToolsSjoerd}

The starting point for the analysis of the functional $\L_{q(\x)}(\z)$ with $q(\x) = \sign(\A \x + \ttau+\nnu)$ is the geometric intuition explained in the beginning. For a pair $\x,\z \in\T$ the value of the functional $\L_{q(\x)}(\z)$ is a weighted counting
of the hyperplanes separating the points $\x,\z \in \R^n$. 
Indeed, if 
\begin{equation*}
    H_{\a,\tau}:=\{\z\in \R^n \; : \;
    \inner{\a}{\z}+\tau=0\}
\end{equation*}
denotes a (inhomogeneous) hyperplane 
with normal vector $\a\in \R^n$, then in the noiseless case ($\nnu=\0$) we can write
\begin{align*}
  \L_{q(\x)}(\z)
  &=  \frac{1}{m} \sum_{i=1}^m
  1_{\sign(\inner{\a_i}{\x} + \tau_i)\neq\sign(\inner{\a_i}{\z} + \tau_i) }|\inner{\a_i}{\z} + \tau_i|\\
  &=  \frac{1}{m} \sum_{i=1}^m
  1_{\{\z \text{ and } \x \text{ lie on different sides of } H_{\a_i, \tau_i} \}}|\inner{\a_i}{\z} + \tau_i|.
\end{align*}
This links the analysis
of $\L_{q(\x)}(\z)$ to results from \cite{dirksen_robust_2018} on random hyperplane tessellations. Let us recall
the notion 
of a \emph{well-separating} hyperplane from
\cite{dirksen_robust_2018}:

\begin{definition}[{\cite[Definition 3.1]{dirksen_robust_2018}}]\label{def:well_sep} Let $i\in [m]$ and $\theta>0$.
The hyperplane $H_{\a_i,\tau_i}$ $\theta$-well-separates $\x$ and $\z$ if
\begin{enumerate}
    \item $\sign(\inner{\a_i}{\x} + \tau_i + \nu_i) \neq 
    \sign(\inner{\a_i}{\z} + \tau_i)$,
    \item $|\inner{\a_i}{\x} + \tau_i + \nu_i| \geq \theta \norm{\x-\z}{2}$, and
    \item $|\inner{\a_i}{\z} + \tau_i| \geq \theta \norm{\x-\z}{2}$.
\end{enumerate}
By $I(\x,\z, \nnu,\theta) \subset [m]$ we denote the set of all indices $i\in [m]$ for which $H_{\a_i,\tau_i}$ $\theta$-well-separates $\x$ and $\z$.
\end{definition}

The next theorem essentially states that if
$\la$ is large enough, then for any two points $\x$ and $\z$ which lie in a bounded set the following holds with high probability: out of the possible $m$ shifted hyperplanes $H_{\a_i, \tau_i}$ (where each $\tau_i$ is uniformly distributed on $[-\la, \la]$), at least $\frac{\norm{\x-\z}{2}}{\la
}m$ of them are well-separating for $\x$ and $\z$.
\begin{theorem}[{\cite[Theorem 3.2]{dirksen_robust_2018}}] \label{thm:noise_well-separating}
There are constants $\kappa, c_1, c_2>0$ for which the following holds.
Set $\lambda \gg \max\{R,\norm{\nu}{L^2}\}$. Then for $\x,\z\in R B^n_2$ with 
probability at least
$1 - 4 \exp(- c_1m \norm{\x - \z}{2}/\lambda )$ 
\begin{equation}
\label{eqn:noise_lower_bound_separating}
|I(\x,\z,\nnu,\kappa)| \geq c_2 m \frac{\norm{\x-\z}{2}}{\lambda} \; .
\end{equation}
\end{theorem}

\blue{Let us point out that the notion of well-seperating hyperplanes together with Theorem~\ref{thm:noise_well-separating} lead to a tight lower bound for $\L_{\qc}$. That this lower bound is optimal can be seen from Lemma~\ref{lem:exp_functional} below, which provides an estimate for the expectation of $\L_{\qc}$.\\}
The main idea behind the notion of a well-separating hyperplane is that if $H_{\a_i,\tau_i}$ well-separates 
$\x'$ and $\z'$ and we are given points $\x\approx \x'$ and $\z\approx \z'$, then also $H_{\a_i, \tau_i}$ well-separates $\x$ and $\z$, provided that $|\inner{\a_i}{\x-\x'}|$ and $|\inner{\a_i}{\z-\z'}|$ are not too large. The next lemma, which makes this observation precise, 
appeared intrinsically in \cite{dirksen_robust_2018}. A proof is provided for convenience in Appendix \ref{app:noise_stability}. 

\begin{lemma}[Stability of well-separating hyperplanes]
\label{lem:noise_stability}
Let $\T\subset \R^n$ be a set, and $\rho,\eps>0$ parameters with $\eps\leq \rho/4$. For every $\kappa>0$, and for all
vectors $\x,\z,\x',\z'\in \T$
with $\norm{\x'-\z'}{2} \geq
\rho$, $\norm{\x-\x'}{2} \leq \ep$ and $\norm{\z-\z'}{2} \leq \ep$ we have
\begin{equation} \label{eq:noise_stability}
    |I(\x,\z,\nnu,\kappa/2)| \geq |I(\x',\z',\nnu,\kappa)|  - 2 \sup_{\y\in (\T-\T)\cap \ep B^n_2}|\{i\in [m]\; : \; |\inner{\a_i}{\y}|>\kappa \rho/4\}|\;.
\end{equation}
\end{lemma}

The right hand side of $\eqref{eq:noise_stability}$ can be bounded using the following lemma, 
which immediately follows from Theorem~\ref{subgaussian_tail} below. For convenience a proof is included in
Appendix \ref{app:noise_stability}.

\begin{lemma} \label{lem:bound_for_spoiled_hyperplanes}
There exist constants $c_1, c_2>0$ such that the following holds.
Let $\rho>0$ and $\mathcal{K} \subset \R^n$. 
If $1\leq k \leq m$ and $u \geq 1$
then with probability at least 
$1-2\exp(-c_1 u^2 k \log(em/k))$,
\begin{displaymath}
 \sup_{\z\in \mathcal{K}}|\{i\in [m]\; : \; |\inner{\a_i}{\z}|>\rho\}|< k \; ,
\end{displaymath}
provided that 
\begin{equation}
    \frac{c_2}{\sqrt{k}}\Big(w_*(\mathcal{K}) + u\; d_2(\mathcal{K}) \sqrt{k \log(em/k)} \Big) \leq  \rho \; .
\end{equation}
\end{lemma}
\begin{theorem}[{\cite[Theorem 2.10]{dirksen_robust_2018}}]
\label{subgaussian_tail}
There exist constants $c_1, c_2>0$ such that the following holds.
Let $\a_1, \ldots, \a_m\in \R^n$ denote independent, isotropic L-subgaussian random vectors and let $\S \subseteq \R^n$ . If $1\leq k \leq m$ and $u \geq 1$ then with probability at least $1 - 2 \exp(-c_1 u^2 k \log(em/k))$, 
\begin{equation}
 \sup_{\x \in \S}  \max_{|I| \leq k} \Big( \sum_{i \in I} \inner{\a_i}{\x}^2 \Big)^{\frac{1}{2}}
 \leq c_2 \Big(w_*(\S) + u d_2(\S) \sqrt{k \log(em/k)} \Big).
\end{equation}
\end{theorem}
The next result from \cite{dirksen_robust_2018} states that the lower bound on the number of well-separating hyperplanes given in Theorem~\ref{thm:noise_well-separating} can be extended to hold simultaneously for all pairs of points from $\T$ which are not too close to each other. 
The proof strategy is to first show the lower bound on the number of well-separating hyperplanes for all pairs of points in a net (which are not too close to each other), and then to extend the estimate to the whole set using the stability property of well-separating hyperplanes (Lemma~\ref{lem:noise_stability}) in combination with Lemma~\ref{lem:bound_for_spoiled_hyperplanes}. 
For convenience the detailed proof is included in Appendix \ref{app:noise_stability}. 
\begin{theorem} \label{thm:noisy_uniform}
There exist constants $c, c'>0$ and $C\geq e$ such that the following holds.
Let $\T\subset R B^n_2$ denote a convex set. 
Choose $\lambda \gg \max\{R, \norm{\nu}{L^2}\}$ and
fix $\rho\leq \lambda$. If
\begin{equation}
    m \gg \frac{\lambda}{\rho} \big(\rho^{-2} w_*((\T-\T)\cap \ep B^n_2)^2 + \log \N(\T,\ep)\big)
\end{equation}
for a number $\eps>0$ with $\ep \lesssim \rho/\sqrt{\log(C\lambda/\rho)}$, 
then with probability at least  
 $1- 2\exp(-
 c m \rho /\lambda)$,
 \begin{equation}
    \forall\, \x, \z \in \T \text{ with } \norm{\x-\z}{2} \geq \rho \; :\;   |I(\x,\z,\nnu,c')| \gg m \frac{\norm{\x-\z}{2}}{\lambda} \; .
\end{equation}
\end{theorem}
Let us lastly state two more results from \cite{dirksen_robust_2018}, which will be important tools for us in order to control the error terms, which appear due to noise corruptions, in the decomposition of the functional $\L_{\qc}(\z)$ (see Section~\ref{sec:ProofOneBit}).
\begin{lemma}[{\cite[Theorem 2.9 and Lemma 4.6]{dirksen_robust_2018}\label{lem:upper}}] 
    There are constants $c_0, c_1>0$ and 
    $C\geq e$ for which the following holds. Let $\T\subset R B^n_2$ denote a convex set. Fix $\rho\leq \lambda$.
    Assume that
    \begin{equation}
        m \gg \frac{\lambda}{\rho} \big(\rho^{-2} w_*((\T-\T)\cap \ep B^n_2)^2 + \log \N(\T,\ep)\big)
    \end{equation}
    for a number $\eps>0$ with $\ep \lesssim \rho/\sqrt{\log(C\lambda/\rho)}$.
    Let $\N_\eps\subset \T$ be a minimal $\eps$-net with respect to the Euclidean metric.
    Then, with probability at least $1-2\exp(-c_0m\rho/\la)$ the following holds: 
    For all $\z\in \N_\eps$ and all $\x\in \T$ such that $\norm{\x-\z}{2}\leq \eps$
    we have
    \begin{equation*}
        |\{i\in [m] \; : \; 
        \sign(\inner{\a_i}{\x} + \nu_i +\tau_i )\neq \sign(\inner{\a_i}{\z} + \nu_i +\tau_i )\}|\leq c_1 \frac{\rho}{\la}m
    \end{equation*}
    and
    \begin{equation*}
        |\{i\in [m] \; : \; 
        \sign(\inner{\a_i}{\x} +\tau_i )\neq \sign(\inner{\a_i}{\z} +\tau_i )\}|\leq c_1 \frac{\rho}{\la}m.
    \end{equation*}
\end{lemma}
\begin{lemma}[{\cite[Theorem 4.4]{dirksen_robust_2018}\label{lem:sup_max_subg}}]
Let $\a_1, \ldots, \a_m\in \R^n$ denote independent, isotropic L-subgaussian random vectors.
There are constants $c_1, c_2>0$ such that the following holds.
Let $\mathcal{K}\subset \R^n$ be a set that is star-shaped around $\0$. 
    For every $\alpha\in (0,1)$ and $\rho>0$, if 
    \begin{equation}\label{eq:lem:sup_max_subg_meas}
       m \geq \alpha^{-1}\log(e/\alpha)^{-1}\rho^{-2} w_*(\mathcal{K}\cap \rho \S^{n-1})^2,
    \end{equation}
then with probability at least $1 - 2 \exp(-c_1 \alpha m \log(e/\alpha))$, 
    \begin{equation*}
    \sup_{\substack{\x\in \mathcal{K}, \, \norm{\x}{2}\geq \rho}} 
    \max_{I\subset [m],\, |I|\leq \alpha m} \frac{1}{m}\sum_{i\in I} 
        \Big|\Big\langle \a_i, \frac{\x}{\norm{\x}{2}^2} \Big\rangle \Big|\leq c_2 \frac{\alpha  \sqrt{\log(e/\alpha)}}{\rho}.
\end{equation*}
\end{lemma}

\subsubsection{Technical Supplement for the Low-Noise Regime} \label{sec:ToolsLowNoise}

The proof of the low-noise regime requires two main ingredients. First, the lower bound on the number of well-separating hyperplanes given in Theorem \ref{thm:noisy_uniform} and, second, a concentration inequality for the
size of the failures produced by the noise term. The necessary concentration inequality, Lemma \ref{lem:concentration_functional}, is derived from the following result.

\begin{lemma}\label{lem:sign_flips_caused_by_noise}
There are constants $c_0, c_1>0$ and 
    $C\geq e$ for which the following holds. Let $\T\subset R B^n_2$ denote a convex set. 
    Choose $\lambda \gg \norm{\nu}{L^2}$ and fix $\rho\leq \lambda$.
    Assume that
    \begin{equation}
       m \gg \frac{\lambda}{\max\{\rho, \norm{\nu}{L^2}\}} \big((\max\{\rho,\norm{\nu}{L^2}\})^{-2} w_*((\T-\T)\cap \ep B^n_2)^2 + \log \N(\T,\ep)\big)
    \end{equation}
    for a number $\eps>0$ with $\ep \lesssim \max\{\rho,\norm{\nu}{L^2}\}/\sqrt{\log(C\lambda/\max\{\rho,\norm{\nu}{L^2}\})}$.
    Then, with probability at least \\
    $1-2\exp(-c_0m\max\{\rho, \norm{\nu}{L^2}\}/\la)$,
    \begin{equation*}
        \sup_{\x\in \T} |\{i\in [m]\; : \; \sign(\langle \a_i,\x \rangle + \nu_i + \tau_i)\neq 
    \sign(\langle \a_i,\x \rangle + \tau_i)\}|\leq c_1 \frac{\max\{\rho, \norm{\nu}{L^2}\} m }{\la} \; . 
    \end{equation*}
\end{lemma}
\begin{proof}
    Let $\N_\ep\subset \T$ denote a minimal $\eps$-net with respect to the Euclidean metric. For every 
$\z \in \R^n$ we have
\begin{align*}
    \Pr(\sign(\langle \a,\z \rangle + \nu + \tau)\neq \sign(\langle \a,\z \rangle + \tau))
    &\leq \Pr(|\nu|\geq |\langle \a,\z \rangle  + \tau|)\\
    &\leq 2\E_{\a,\ttau} \exp(-c |\langle \a,\z \rangle  + \tau|^2/ \norm{\nu}{\psi_2}^2)\\
    &\leq 2\E_{\a,\ttau} \exp(-c |\langle \a,\z \rangle  + \tau|^2/ \norm{\nu}{L^2}^2),
\end{align*}
\blue{where we used the subgaussian tail property of $\nu$} and
that $\norm{\nu}{\psi_2}\lesssim \norm{\nu}{L^2}$.
Moreover, for any $a\in \R$,
\begin{align*}
   \E_{\ttau} \exp(- |a  + \tau|^2/ 2x^2)  = \frac{1}{2\la}\int_{-\la}^{\la} \exp(-(a+s)^2/2x^2)ds
   &=\frac{1}{2\la}\int_{-\la+a}^{\la+a} \exp(-s^2/2x^2)ds\\
   &\leq\frac{\sqrt{2\pi x^2}}{2\la}\int_{\R} \frac{1}{\sqrt{2\pi x^2}}\exp(-s^2/2x^2)ds\\
   &=\sqrt{\frac{\pi}{2}}\frac{|x|}{\la}.
\end{align*}
Therefore, for every $\z\in \R^n$,
\begin{align*}
    \Pr(\sign(\langle \a,\z \rangle + \nu + \tau)\neq \sign(\langle \a,\z \rangle + \tau))\leq c_1\frac{\norm{\nu}{L^2}}{\la}.
\end{align*}
Chernoff's inequality implies that
\begin{equation}
    \label{eqn:upper_bound_set}
    |\{i\in [m]\; : \; \sign(\langle \a_i,\z \rangle + \nu_i + \tau_i)\neq 
    \sign(\langle \a_i,\z \rangle + \tau_i)\}|\leq 2c_1 \frac{\max\{\norm{\nu}{L^2}, \rho\}m}{\la} \; 
\end{equation}
with probability at
least $1-\exp(-c\max\{\rho,\norm{\nu}{L^2}\} m/\la)$.
If $\log|\N_\ep| \leq  c\max\{\rho,\norm{\nu}{L^2}\} m/(2\la)$, a union bound shows that \eqref{eqn:upper_bound_set} holds for all $z \in \N_\ep$ with probability at least $1-\exp(-c \max\{\rho,\norm{\nu}{L^2}\} m/(2\la))$.
Lemma~\ref{lem:upper} implies that with probability at least $1-2\exp(-c'\max\{\rho,\norm{\nu}{L^2}\} m/\lambda)$ the following holds: for all $z\in \N_\eps$ and all $\x\in \T$ such that $\norm{\x-\z}{2}\leq \eps$,
    \begin{align} \label{eqn:upper_bound_set2}
    \begin{split}
        &|\{i\in [m] \; : \; 
        \1_{\{\sign(\langle \a_i,\x \rangle + \nu_i + \tau_i)\neq \sign(\langle \a_i,\x \rangle + \tau_i)\}} \neq \1_{\{\sign(\langle \a_i,\z \rangle + \nu_i + \tau_i)\neq \sign(\langle \a_i,\z \rangle + \tau_i)\}}
        \}|\\
        &\leq 
        |\{i\in [m] \; : \; 
    \sign(\langle \a_i,\x \rangle + \nu_i + \tau_i)\neq \sign(\langle \a_i,\z \rangle + \nu_i + \tau_i)
        \}| \\
        &\quad + |\{i\in [m] \; : \; 
    \sign(\langle \a_i,\x \rangle  + \tau_i)\neq \sign(\langle \a_i,\z \rangle + \tau_i)
        \}|\\
        &\leq c_2 \frac{\max\{\rho,\norm{\nu}{L^2}\}}{\la}m  \; .
        \end{split}
    \end{align}
As all measurements $i \in [m]$ for which $\sign(\langle \a_i,\x \rangle + \nu_i + \tau_i)\neq \sign(\langle \a_i,\x \rangle + \tau_i)$ are either in \eqref{eqn:upper_bound_set} or \eqref{eqn:upper_bound_set2}, we obtain that with probability at least $1-2\exp(-c_0\max\{\rho,\norm{\nu}{L^2}\} m/\lambda)$,
\begin{align}
   \sup_{\x\in \T} |\{i\in [m]\; : \; \sign(\langle \a_i,\x \rangle + \nu_i + \tau_i)\neq 
    \sign(\langle \a_i,\x \rangle + \tau_i)\}|\leq c_3 \frac{\max\{\rho,\norm{\nu}{L^2}\} m }{\la} \; . 
\end{align}
\end{proof}

\begin{lemma} \label{lem:concentration_functional}
    There are constants $c_0, c_1>0$ and 
    $C\geq e$ for which the following holds. Let $\T\subset R B^n_2$ denote a convex set. 
    Choose $\lambda \gtrsim R$ and fix $\rho\leq \lambda$.
    Assume that $\norm{\nu}{L^2}\leq c_0 \rho/\sqrt{\log(C\lambda/\rho)}$ and suppose
    \begin{equation}
       m \gg \frac{\lambda}{\rho} \big( \rho^{-2} w_*((\T-\T)\cap \ep B^n_2)^2 + \log \N(\T,\ep) \big) 
    \end{equation}
    for a number $\eps>0$ with $\ep \lesssim \rho/\sqrt{\log(C\lambda/\rho)}$.
    Then, with probability at least $1- 2 \exp(-c_1 m \rho/ \lambda)$,
    \begin{equation}
        \sup_{\x \in \T} (Y(\x) - \E Y(\x)) \leq \frac{\rho^2}{\lambda} \; ,
    \end{equation}
    where $Y(\x)$ is the random variable defined by
    \begin{equation}
        Y(\x)=\frac{1}{m} \sum_{i=1}^m \1_{\{\sign(\langle \a_i,\x \rangle + \nu_i + \tau_i)\neq \sign(\langle \a_i,\x \rangle + \tau_i)\}}|\nu_i| \; .
    \end{equation}
\end{lemma}

\begin{proof}
Let $\eps_1, \ldots, \eps_m$ denote a sequence of independent symmetric $\{\pm 1\}$-valued Bernoulli random variables that are independent of all other random variables and define the symmetrized random variable 
\begin{align}
    Y'(\x):= 
     \frac{1}{m} \sum_{i=1}^m \ep_i \1_{\{\sign(\langle \a_i,\x \rangle + \nu_i + \tau_i)\neq \sign(\langle \a_i,\x \rangle + \tau_i)\}}|\nu_i|.
\end{align}
By symmetrization (e.g. see \cite[Eq. (6.2)]{ledoux_probability_1991}) for every $t>0$,
\begin{align}
   \inf_{\x \in \T}\Pr(Y(\x) - \E Y(\x)\leq t) \Pr( \sup_{\x \in \T} (Y(\x) - \E Y(\x)) > 2t  ) \leq 2 \Pr( \sup_{\x \in \T} |Y'(\x)|  > t/2  ) \; .
\end{align}
For any $\x\in \R^n$,
\begin{align}\label{eq:mean_Y}
    \E Y(\x)&= \E (1_{\{\sign(\langle \a,\x \rangle + \nu + \tau)\neq \sign(\langle \a,\x \rangle + \tau)\}}|\nu|)\\ \nonumber
    &=\E (1_{\nu\geq 0}1_{\{\sign(\langle \a,\x \rangle + \nu + \tau)\neq \sign(\langle \a,\x \rangle + \tau)\}}|\nu|)
    + \E (1_{\nu< 0}1_{\{\sign(\langle \a,\x \rangle + \nu + \tau)\neq \sign(\langle \a,\x \rangle + \tau)\}}|\nu|)\\ \nonumber
    &=\E (1_{\nu\geq 0}1_{\{ -\langle \a,\x \rangle - \nu\leq \tau < -\langle \a,\x \rangle\}}|\nu|)
    + \E (1_{\nu< 0}1_{\{ -\langle \a,\x \rangle \leq \tau < -\langle \a,\x \rangle - \nu\}}|\nu|)\\ \nonumber
    &\leq \E (1_{\nu\geq 0}\frac{|\nu|^2}{2\lambda})
    + \E (1_{\nu< 0}\frac{|\nu|^2}{2\lambda}) = \frac{1}{2\la}\norm{\nu}{L^2}^2  \; .
\end{align}
Using $\norm{\nu}{L^2}\leq \rho$ we obtain
\begin{equation}
   \Pr \left( Y(\x)- \E Y(\x)\leq  \frac{\rho^2}{\la} \right)\geq \Pr \left( Y(\x)- \E Y(\x)\leq  \frac{\norm{\nu}{L^2}^2}{\la} \right) \geq  \Pr \left( Y(\x)\leq  \frac{\norm{\nu}{L^2}^2}{\la} \right)\geq \frac{1}{2}.
\end{equation}
The last inequality follows from Markov's inequality and
$\E Y(\x)\leq \frac{\norm{\nu}{L^2}^2}{2\la}$.
Hence, 
\begin{equation}\label{eq:bound_by_symmetrized_process}
\Pr( \sup_{\x \in \T} (Y(\x) - \E Y(\x)) > 2\rho^2/\la  ) \leq 4 \Pr( \sup_{\x \in \T} |Y'(\x)|  > \rho^2/2\la  ).
\end{equation}
Define the event
\begin{equation}
    \mathcal{E}:=\left\{ \sup_{\x\in \T} |\{i\in [m]\; : \; \sign(\langle \a_i,\x \rangle + \nu_i + \tau_i)\neq 
    \sign(\langle \a_i,\x \rangle + \tau_i)\}|\leq c_1 \frac{\rho m }{\la}
    \right\}.
\end{equation}
Since $\norm{\nu}{L^2}\lesssim \rho$, Lemma~\ref{lem:sign_flips_caused_by_noise} implies $\Pr(\mathcal{E}^C)\leq 2 \exp(-c_0\rho m/\la)$. Further, observe that on $\mathcal{E}$,
\begin{equation}
  \sup_{\x \in \T} |Y'(\x)| =  \sup_{\x \in \T}\Big| \frac{1}{m}\sum_{i=1}^m \1_{\{\sign(\langle \a_i,\x \rangle + \nu_i + \tau_i)\neq \sign(\langle \a_i,\x \rangle + \tau_i)\}} \ep_i |\nu_i|  \Big|
    \leq \max_{|I| \leq c_1 \rho m/\lambda} \frac{1}{m}\Big| \sum_{i \in I} \ep_i |\nu_i| \Big| \; .
\end{equation}
Hence,
\begin{align}\label{eq:symmetrized_process_on_event}
\Pr \Big( \sup_{\x \in \T} |Y'(\x)|  > \rho^2/\la  \Big)
&\leq \Pr(\mathcal{E}^C) + 
\Pr \Big(\mathcal{E}, \sup_{\x \in \T} |Y'(\x)|  > \rho^2/\la  \Big) \\
&\leq 
2\exp(-c_0\rho m/\lambda)
+ 
\Pr\Big(  \max_{|I| \leq c_1 \rho m/\la} \Big|\frac{1}{m} \sum_{i \in I} \ep_i |\nu_i| \Big| >\rho^2/\la \Big).
\end{align}
To estimate the second summand, first observe that Hoeffding's inequality implies 
\begin{align*}
    \norm{\frac{1}{\sqrt{|I|}}\sum_{i \in I} \ep_i |\nu_i|}{\psi_2}\lesssim \frac{1}{\sqrt{|I|}}\sqrt{\sum_{i\in I} \norm{\ep_i |\nu_i|}{\psi_2}^2}\leq \norm{\nu}{L^2}
\end{align*}
for every $I\subset [m]$.
By the union bound we obtain for every $k\in [m]$,
\begin{align*}
    \Pr\Big(  \max_{|I| \leq k} \Big| \frac{1}{m}\sum_{i \in I} \ep_i |\nu_i| \Big| > \rho^2/\la \Big)
    &\leq \sum_{l=1}^k
    \Pr\Big(  \max_{|I| = l} \Big|\frac{1}{\sqrt{|I|}} \sum_{i \in I} \ep_i |\nu_i| \Big| > \frac{\rho^2m}{\la\sqrt{l}} \Big)\\
    &\leq 2\sum_{l=1}^k
    \left(\frac{em}{l}\right)^l
    \exp(-c_2 \rho^4 m^2 /(\la^2 l \norm{\nu}{L^2}^2))\\
    &= 2\sum_{l=1}^k
    \exp(l\log(em/l))
    \exp(-c_2 \rho^4 m^2 /(\la^2 l \norm{\nu}{L^2}^2))\\
    &\leq 2k 
    \exp\left(k\log(em/k)-c_2 \rho^4 m^2 /(\la^2 k \norm{\nu}{L^2}^2)\right)\\
    &=\exp\left(\log(2k)+k\log(em/k)-c_2 \rho^4 m^2 /(\la^2 k \norm{\nu}{L^2}^2)\right)\\
    &\leq \exp\left(2k\log(em/k)-c_2 \rho^4 m^2 /(\la^2 k \norm{\nu}{L^2}^2)\right).
\end{align*}
Hence, if 
\begin{equation}\label{eq:lem:concentration_functional_condition}
    m\gtrsim \la \rho^{-2}\norm{\nu}{L^2}k\sqrt{\log(em/k)},
\end{equation}
then
\begin{equation}
   \Pr\Big(  \max_{|I| \leq k} \Big| \frac{1}{m}\sum_{i \in I} \ep_i |\nu_i| \Big| > \rho^2/\la \Big)\leq\exp\left(-c_3 \rho^4 m^2 /(\la^2 k \norm{\nu}{L^2}^2)\right). 
\end{equation}
For $k=c_1\rho m/\la$, condition \eqref{eq:lem:concentration_functional_condition} is satisfied if 
\begin{equation}\label{eq:lem:concentration_functional_condition_2}
    \rho\gtrsim \norm{\nu}{L^2}\sqrt{\log(C\la/\rho)}, %, \quad m\gtrsim \la \rho^{-3}\norm{\nu}{L^2}^2 \log(c\rho m/\la),
\end{equation}
where $C>0$ is a constant that only depends on $c_1$. Therefore, if \eqref{eq:lem:concentration_functional_condition_2} is satisfied, then
\begin{equation}
   \Pr\Big(  \max_{|I| \leq c_1\rho m/\la} \Big| \frac{1}{m}\sum_{i \in I} \ep_i |\nu_i| \Big| > \rho^2/\la \Big)\leq\exp\left(-c_4 \rho^3 m /(\la  \norm{\nu}{L^2}^2)\right)\leq 
   \exp(-c_5 \rho m/\la). 
\end{equation}
Together with \eqref{eq:bound_by_symmetrized_process} and \eqref{eq:symmetrized_process_on_event} this shows the result.
\end{proof}
\subsubsection{Technical Supplement for the High-Noise Regime} \label{sec:ToolsHighNoise}

The following lemma states that for any $\x, \z\in \R^n$, if $\lambda$ is large enough, then 
\begin{equation*}
    \E \left( \L_{\q(\x)}(\z) \right)\approx \frac{\norm{\x-\z}{2}^2+\norm{\nu}{L^2}^2}{2\la}.
\end{equation*}
\begin{lemma}\label{lem:exp_functional}
Let $\la \geq c_0(L) R$.
Then, there exists a constant $C>0$ that only depends on $L$ and an absolute constant $c>0$ such that
for any $\x, \z\in R B^n_2$,
    \begin{equation*}
    \left|\E \left( \L_{\q(\x)}(\z) \right) -\frac{\norm{\x-\z}{2}^2+\norm{\nu}{L^2}^2}{2\la} \right|
    \leq C \, \la \exp\Big(-c \frac{\la^2}{(L^2R^2 + \norm{\nu}{\psi_2}^2)}\Big) \; .
    %F(\la, R, \norm{\nu}{\psi_2}),
\end{equation*}
In particular,
\begin{equation}
    \left|\E \left( \L_{\q(\x)}(\z)- \L_{\q(\x)}(\x) \right) - \frac{\norm{\x-\z}{2}^2}{2\la}\right|
    \leq 2 C \, \la \exp\Big(-c \frac{\la^2}{(L^2R^2 + \norm{\nu}{\psi_2}^2)}\Big)\; .
\end{equation}
\end{lemma}
Lemma~\ref{lem:exp_functional} indicates that for reconstruction via \eqref{eq:P} the assumption of the previous section, namely that the estimation accuracy $\rho>0$ is lower bounded by $\norm{\nu}{L^2}$, is not necessary. Indeed, it implies that given any (arbitrarily small) $\rho>0$, if $\lambda\gtrsim \max\{R, \norm{\nu}{L^2}\}\sqrt{\log(\la/\rho)}$, then 
\begin{equation*}
\E (\L_{\q(\x)}(\z) - \L_{\q(\x)}(\x))\geq \frac{\norm{\x-\z}{2}^2}{4\lambda}
\end{equation*}
for all $\x, \z\in R B^n_2$ with $\norm{\x-\z}{2}\geq \rho$. Since every solution $\x^{\#}$ of \eqref{eq:P} with $\qc=q(\x)$ (here we assume that there are no adversarial bit corruptions) satisfies $\L_{\q(\x)}(\x^{\#}) - \L_{\q(\x)}(\x)\leq 0$, we obtain $\norm{\x-\x^{\#}}{2}\leq \rho$ provided that $\L_{\q(\x)}(\z) - \L_{\q(\x)}(\x)$ concentrates well enough around its expected value uniformly for all $\x, \z\in \T$ with $\norm{\x-\z}{2}\geq \rho$. The latter statement is shown in Theorem~\ref{thm:concencration_functional} below.
\begin{proof}[of Lemma~\ref{lem:exp_functional}]
Define the event
\begin{equation}
    \mathcal{A}:=\{|\inner{\a}{\x}+\nu|\leq \la, |\inner{\a}{\z}|\leq \la \}.
\end{equation}
\begin{align*}
    \E \L_{\q(\x)}(\z)
    &= \E ( 1_{\sign(\inner{\a}{\x}+\nu +\tau)\neq \sign(\inner{\a}{\z}+\tau)}|\inner{\a}{\z}+\tau| )\\
    &=\E ( 1_{\inner{\a}{\x}+\nu \leq -\tau\leq \inner{\a}{\z}}(\inner{\a}{\z}+\tau) )
    +\E ( 1_{ \inner{\a}{\z}\leq -\tau\leq \inner{\a}{\x}+\nu}(-\inner{\a}{\z}-\tau) ).
\end{align*}
By symmetry of $\a$ and $\tau$,
\begin{align*}
   \E ( 1_{\inner{\a}{\x}-\nu \leq -\tau\leq \inner{\a}{\z}}(\inner{\a}{\z}+\tau) )
   &=\E ( 1_{\inner{-\a}{\x}-\nu \leq \tau\leq \inner{-\a}{\z}}(\inner{-\a}{\z}-\tau) )\\
   &=\E ( 1_{\inner{\a}{\x}+\nu \geq  -\tau\geq \inner{\a}{\z}}(\inner{-\a}{\z}-\tau) )\\
   &=\E ( 1_{ \inner{\a}{\z}\leq  -\tau\leq \inner{\a}{\x}+\nu}(\inner{-\a}{\z}-\tau) ).\\
\end{align*}
Consequently, it suffices to show that 
\begin{equation*}
    \left|\E ( 1_{\inner{\a}{\x}+\nu \leq -\tau\leq \inner{\a}{\z}}(\inner{\a}{\z}+\tau) ) -\frac{\norm{\x-\z}{2}^2+\norm{\nu}{L^2}^2}{4\la}\right|\leq C \, \la \exp\Big(-c \frac{\la^2}{(L^2R^2 + \norm{\nu}{\psi_2}^2)}\Big).
\end{equation*}
We have
\begin{align*}
   \E ( 1_{\inner{\a}{\x}+\nu \leq -\tau\leq \inner{\a}{\z}}(\inner{\a}{\z}+\tau) )
   &=\E ( 1_{\inner{\a}{\x}+\nu \leq -\tau\leq \inner{\a}{\z}}(\inner{\a}{\z}-(-\tau)) )\\
   &=\E ( 1_{\mathcal{A}}1_{\inner{\a}{\x}+\nu \leq -\tau\leq \inner{\a}{\z}}(\inner{\a}{\z}-(-\tau)) )\\
   &\quad +
   \E ( 1_{\mathcal{A}^C}1_{\inner{\a}{\x}+\nu \leq -\tau\leq \inner{\a}{\z}}(\inner{\a}{\z}-(-\tau)) ).\\
\end{align*}
Further, we have
\begin{align*}
   \E ( 1_{\mathcal{A}}1_{\inner{\a}{\x}+\nu \leq -\tau\leq \inner{\a}{\z}}(\inner{\a}{\z}-(-\tau)) )
   &=\frac{1}{2\la}\E_{\a,\nnu} \left( 1_{\mathcal{A}}
   \int_{\inner{\a}{\x}+\nu}^{\inner{\a}{\z}}(\inner{\a}{\z}-s)ds \right)\\
   %&=\frac{1}{2\la}\E 1_{\mathcal{A}}
   %\inner{\a}{\z}(\inner{\a}{\z}-(\inner{\a}{\x}+\nu))\\
   %&\quad +\frac{1}{2\la}\E 1_{\mathcal{A}}
   %\frac{1}{2}((\inner{\a}{\x}+\nu)^2-\inner{\a}{\z}^2)\\
   &=\frac{1}{4\la} \E ( 1_{\mathcal{A}}(\inner{\a}{\x-\z}+\nu)^2 ).
\end{align*}
Hence, we obtain the following inequality,
\begin{align*}
    \left| \E ( 1_{\mathcal{A}}1_{\inner{\a}{\x}+\nu \leq -\tau\leq \inner{\a}{\z}}(\inner{\a}{\z}-(-\tau)) ) -\frac{1}{4\la}\E ( (\inner{\a}{\x-\z}+\nu)^2 ) \right|\leq 
    \frac{1}{4\la} \E ( 1_{\mathcal{A}^C}(\inner{\a}{\x-\z}+\nu)^2 ) \; .
\end{align*}
Since $\frac{1}{4\la} \E (\inner{\a}{\x-\z}+\nu)^2 = \frac{1}{4\la} (\norm{\x-\z}{2}^2+\norm{\nu}{L^2}^2)$, 
it follows that
\begin{align*}
   &\left| \E ( 1_{\inner{\a}{\x}+\nu \leq -\tau\leq \inner{\a}{\z}}(\inner{\a}{\z}-(-\tau)) ) -\frac{\norm{\x-\z}{2}^2+\norm{\nu}{L^2}^2}{4\la} \right|\\
   &\leq \E ( 1_{\mathcal{A}^C}|\inner{\a}{\z}+\tau| ) +\frac{1}{4\la}\E ( 1_{\mathcal{A}^C}(\inner{\a}{\x-\z}+\nu)^2 )\\
   &\leq (\Pr(\mathcal{A}^C))^{1/2}
   \norm{\inner{\a}{\z}+\tau}{L^2}+
   \frac{1}{4\la}(\Pr(\mathcal{A}^C))^{1/2}\norm{\inner{\a}{\z}+\tau}{L_4}^2\\
   &\lesssim (\Pr(\mathcal{A}^C))^{1/2}(R+\la + \frac{1}{\la}(L^2R^2 + \la^2)).
\end{align*}
Finally, observe that
\begin{align*}
  \Pr(\mathcal{A}^C)
  &\leq \Pr(|\inner{\a}{\x}|+|\nu|> \la) + \Pr(|\inner{\a}{\z}|>\la)  \\
  &\leq \Pr(|\inner{\a}{\x}|> \la/2) + \Pr(|\nu|>\la/2) +\Pr(|\inner{\a}{\z}|>\la)\\
  &\leq e^{-c\la^2/\norm{\inner{\a}{\x}}{\psi_2}^2}
  + e^{-c\la^2/\norm{\inner{\a}{\z}}{\psi_2}^2}
  + e^{-c\la^2/\norm{\nu}{\psi_2}^2}\\
  &\leq 2e^{-c\la^2/L^2R^2}
  + e^{-c\la^2/\norm{\nu}{\psi_2}^2} \; ,
\end{align*}
where $c>0$ denotes an absolute constant.
To conclude the proof, we observe that since $\la \geq c_0 R$, we have
\begin{align*}
(R+\la + \frac{1}{\la}(L^2R^2 + \la^2)) \left( 2e^{-c\la^2/L^2R^2} + e^{-c\la^2/\norm{\nu}{\psi_2}^2} \right)
&\lesssim \la \left( e^{-c\la^2/L^2R^2} + e^{-c\la^2/\norm{\nu}{\psi_2}^2} \right) \\
&\lesssim \la \exp\Big(- c \frac{\la^2}{\max\{\norm{\nu}{\psi_2}^2,L^2R^2 \}} \Big) \\
&\lesssim \la \exp\Big( - c \frac{\la^2}{2 ( \, (LR)^2 + \norm{\nu}{\psi_2}^2)} \Big) \; .
\end{align*}
This gives the desired bound.
\end{proof}
\begin{theorem}\label{thm:concencration_functional}
There are constants $c_0, c_1>0$ and 
    $C\geq e$ for which the following holds. Let $\T\subset R B^n_2$ denote a convex set and fix $\rho\leq \lambda$.
    %Choose $\lambda \gtrsim R$ and fix $\rho\leq \lambda$.
    Suppose
    \begin{equation}
    \label{eq:conditions_concentration}
m \gtrsim \left( \frac{\lambda}{\rho} \right)^2 \big(\rho^{-2}  w_*((\T-\T)\cap \rho B^n_2)^2 + \log(\mathcal{N}(\T,\eps)\big)
\end{equation}
    for $\eps=c_1 \rho/\log(C\lambda/\rho)$.
    Then, with probability at least $1- 2 \exp(-c_0 m \rho^2/ \lambda^2)$,
\begin{equation}
   \sup_{\substack{\x, \z \in \T, \\ \norm{\x-\z}{2}= \rho}} \Big| (\L_{\q(\x)}(\z)-\L_{\q(\x)}(\x))- \E (\L_{\q(\x)}(\z)-\L_{\q(\x)}(\x)) \Big| \leq \frac{\rho^2}{8\la} \; .
\end{equation}
\end{theorem}
\begin{proof}
We begin by defining the symmetrized random variable
    \begin{equation}
        Z(\x, \z):=\frac{1}{m}\sum_{i=1}^m 
       \eps_i(\relu{-q(\x)_i(\inner{\a_i}{\z} + \tau_i)} -
     \relu{-q(\x)_i(\inner{\a_i}{\x} + \tau_i)}) \;.
    \end{equation}
    By symmetrization (e.g. see \cite[Eq. (6.2)]{ledoux_probability_1991}) for every $t>0$,
\begin{align*}
 a(t) \cdot\Pr\Big(  \sup_{\x, \z \in \T, \norm{\x-\z}{2}= \rho} &\Big|(\L_{\q(\x)}(\z)-\L_{\q(\x)}(\x))- \E (\L_{\q(\x)}(\z)-\L_{\q(\x)}(\x)) \Big| > 2t  \Big) \\
  &\leq 2 \Pr( \sup_{\x, \z \in \T, \norm{\x-\z}{2}= \rho} |Z(\x, \z)|  > t/2  ) \; ,
\end{align*}
where
\begin{equation}
    a(t):=\inf_{\x, \z \in \T, \norm{\x-\z}{2}= \rho}\Pr\Big(|(\L_{\q(\x)}(\z)-\L_{\q(\x)}(\x))- \E (\L_{\q(\x)}(\z)-\L_{\q(\x)}(\x))| \leq t\Big) \; .
\end{equation}
By Hoeffding's inequality,
\begin{align*}
    &\norm{\L_{\q(\x)}(\z)-\L_{\q(\x)}(\x))- \E (\L_{\q(\x)}(\z)-\L_{\q(\x)}(\x))}{\psi_2}\\
    &\quad \lesssim \frac{1}{\sqrt{m}}
    \|(\relu{-q(\x)_i(\inner{\a_i}{\z} + \tau_i)} -
     \relu{-q(\x)_i(\inner{\a_i}{\x} + \tau_i)}) \\
     &\qquad- \E(\relu{-q(\x)_i(\inner{\a_i}{\z} + \tau_i)} -
     \relu{-q(\x)_i(\inner{\a_i}{\x} + \tau_i)})\|_{\psi_2}\\
     &\quad \lesssim \frac{1}{\sqrt{m}}
    \norm{(\relu{-q(\x)_i(\inner{\a_i}{\z} + \tau_i)} -
     \relu{-q(\x)_i(\inner{\a_i}{\x} + \tau_i)})}{\psi_2}\\
     &\quad \lesssim \frac{1}{\sqrt{m}}
    \norm{\inner{\a_i}{\x-\z} }{\psi_2} 
    \lesssim \frac{\norm{\x-\z}{2}}{\sqrt{m}} \; .
\end{align*}
Here, we used the fact that (observe that the function $x \mapsto \relu{x}$ is $1$-Lipschitz)
\begin{equation}
\label{eq:relu_triangle}
   | \relu{-q(\x)_i(\inner{\a_i}{\z} + \tau_i)} -
     \relu{-q(\x)_i(\inner{\a_i}{\x} + \tau_i)}|\leq |\inner{\a_i}{\x-\z}| \; .
\end{equation}
Hence, for every $\x, \z\in \R^n$ with $\norm{\x-\z}{2}=\rho$,
\begin{align*}
    &\Pr\Big(\Big| \L_{\q(\x)}(\z)-\L_{\q(\x)}(\x))-\E (\L_{\q(\x)}(\z)-\L_{\q(\x)}(\x)\Big|
    > \frac{\rho^2}{16\la}\Big)\\
    &\quad\quad \leq 2\exp(-cm\rho^2/\la^2).
\end{align*}
Therefore, if $m\gtrsim (\lambda/\rho)^2$, then $a(\rho^2/16\la)\geq \frac{1}{2}$
which implies 
\begin{align*}
 &\Pr\Big(\sup_{\x, \z \in \T, \norm{\x-\z}{2}= \rho}  \Big|\L_{\q(\x)}(\z)-\L_{\q(\x)}(\x))- \E (\L_{\q(\x)}(\z)-\L_{\q(\x)}(\x) \Big| > \frac{\rho^2}{8\la}  \Big)\\
  &\leq 4 \Pr\Big( \sup_{\x, \z \in \T, \norm{\x-\z}{2}= \rho} |Z(\x, \z)|  > \frac{\rho^2}{32\la}  
  \Big) \; .
\end{align*}
Let $N_\eps \subset \T$ be a minimal $\eps$-net with respect to the Euclidean metric. Clearly, 
\begin{align*}
  \sup_{\substack{\x, \z \in \T, \\ \norm{\x-\z}{2}= \rho}} |Z(\x,\z)| 
  =\max_{\x'\in N_\eps}\; \sup_{\substack{\x,\z\in \T,\; \norm{\x-\x'}{2}
  \leq \eps, \\ \norm{\x-\z}{2} = \rho}} |Z(\x,\z)| \; .
\end{align*}
For every fixed $\x'\in N_\eps$,
\begin{align}
\sup_{\substack{\x,\z\in \T, \; \norm{\x-\x'}{2}\leq \eps, \\ \norm{\x-\z}{2}= \rho}} |Z(\x,\z)| 
       &\leq \sup_{\substack{\x,\z\in \T, \; \norm{\x-\x'}{2}\leq \eps, \\ \norm{\x-\z}{2}= \rho}} \bigg|\frac{1}{m}\sum_{i=1}^m 
       \eps_i(\relu{-q(\x')_i(\inner{\a_i}{\z} + \tau_i)} -
     \relu{-q(\x')_i(\inner{\a_i}{\x} + \tau_i)})\bigg| \label{eq:fixed_part_original}\\
       &\quad\quad +\sup_{\substack{\x,\z\in \T, \; \norm{\x-\x'}{2}\leq \eps, \\ \norm{\x-\z}{2}= \rho}} \bigg|\frac{1}{m}\sum_{i=1}^m 
       \eps_i(\relu{-q(\x)_i(\inner{\a_i}{\z} + \tau_i)} -
     \relu{-q(\x)_i(\inner{\a_i}{\x} + \tau_i)}) \label{eq:error_term}\\
     &\quad\quad-\bigg(\frac{1}{m}\sum_{i=1}^m 
       \eps_i(\relu{-q(\x')_i(\inner{\a_i}{\z} + \tau_i)} -
     \relu{-q(\x')_i(\inner{\a_i}{\x} + \tau_i)})\bigg)\bigg|\; . \nonumber
\end{align}
We will bound the terms \eqref{eq:fixed_part_original} and \eqref{eq:error_term} separately. We start
by controlling the error term \eqref{eq:error_term}. Observe that the terms of the sum in \eqref{eq:error_term} are equal to $0$ if $q(\x)_i = q(\x')_i$. Hence, by using \eqref{eq:relu_triangle}
it follows that
\begin{align*}
    &\sup_{\substack{\x,\z\in \T, \; \norm{\x-\x'}{2}\leq \eps, \\ \norm{\x-\z}{2}= \rho}} 
    \bigg|\frac{1}{m}\sum_{i=1}^m 
       \eps_i(\relu{-q(\x)_i(\inner{\a_i}{\z} + \tau_i)} -
     \relu{-q(\x)_i(\inner{\a_i}{\x} + \tau_i)})\\
     &\quad\quad -\bigg(\frac{1}{m}\sum_{i=1}^m \eps_i ( \relu{-q(\x')_i(\inner{\a_i}{\z} + \tau_i)} -
     \relu{-q(\x')_i(\inner{\a_i}{\x} + \tau_i)})\bigg)\bigg|\\
&\quad \leq 2 \sup_{\substack{\x,\z\in \T, \; \norm{\x-\x'}{2}\leq \eps, \\ \norm{\x-\z}{2}= \rho}} \frac{1}{m}\sum_{q(\x)_i\neq q(\x')_i} |\inner{\a_i}{\x-\z}| \;. 
\end{align*}
On the event 
\begin{equation}
\mathcal{E}_\alpha:=\Big\{\sup_{\substack{\x'\in N_\eps, \; \x\in \T,\\ \norm{\x-\x'}{2}\leq \eps}}
|\{i\in [m] \; : \; q(\x)_i\neq q(\x')_i\}|\leq \frac{\alpha m}{\la}\Big\}
\end{equation}
we have 
\begin{align*}
    &\sup_{\substack{\x,\z\in \T, \; \norm{\x-\x'}{2}\leq \eps, \\ \norm{\x-\z}{2}= \rho}} 
    \bigg|\frac{1}{m}\sum_{i=1}^m 
       \eps_i(\relu{-q(\x)_i(\inner{\a_i}{\z} + \tau_i)} -
     \relu{-q(\x)_i(\inner{\a_i}{\x} + \tau_i)})\\
     &\quad\quad-\bigg(\sum_{i=1}^m 
       \eps_i(\relu{-q(\x')_i(\inner{\a_i}{\z} + \tau_i)} -
     \relu{-q(\x')_i(\inner{\a_i}{\x} + \tau_i)})\bigg)\bigg|\\
&\quad \le 2
\sup_{\substack{\x,\z\in \T, \; \norm{\x-\x'}{2}\leq \eps, \\ \norm{\x-\z}{2}= \rho}} \max_{I\subset [m], |I|\leq \alpha m/\la} \frac{1}{m}\sum_{i\in I} |\inner{\a_i}{\x-\z}|\\
%&\quad \leq \sup_{\substack{\x,\z\in \T,  \\ \norm{\x-\z}{2}= \rho}} \max_{I\subset [m], |I|\leq \alpha m/\la} \frac{1}{m}\sum_{i\in I} |\inner{\a_i}{\x-\z}|\\
&\quad \leq 2 \sup_{\substack{\x,\z\in \T,\\ \norm{\x-\z}{2}= \rho}} \max_{I\subset [m], |I|
\leq \alpha m/\la} \frac{1}{m}\sum_{i \in I} |\inner{\a_i}{\x-\z}|\\
&\quad \leq 2 \sqrt{\frac{\alpha}{m\la}}
 \sup_{\substack{\x,\z\in \T,\\ \norm{\x-\z}{2}= \rho}} \max_{I\subset [m], |I|\leq \alpha m/\la} \Big(\sum_{i \in I} \inner{\a_i}{\x-\z}^2 \Big)^{1/2} \; .
%\\
%&\quad \leq \sqrt{\frac{\rho}{m\la}}\sup_{\x,\z\in \T, \norm{\x-\x'}{2}\leq \eps, \norm{\x-\z}{2}\geq \rho} \max_{I\subset [m], |I|\leq \rho m/\la} \big(\sum_{i\in I} |\inner{\a_i}{\x-\z}|^2\big)^{1/2}
\end{align*}
Therefore, on the event $\mathcal{E}_\alpha$,
\begin{align}\label{eq:thm:concencration_functional_two_summands}
   \sup_{\substack{\x, \z \in \T, \\ \norm{\x-\z}{2}= \rho}} |Z(\x,\z)| 
   &\leq \max_{\x'\in N_\eps}\sup_{\substack{\x,\z\in \T, \norm{\x-\x'}{2}\leq \eps,\\ \norm{\x-\z}{2}= \rho}} \bigg|\frac{1}{m}\sum_{i=1}^m 
       \eps_i(\relu{-q(\x')_i(\inner{\a_i}{\z} + \tau_i)} -
     \relu{-q(\x')_i(\inner{\a_i}{\x} + \tau_i)})\bigg|\\ \nonumber
     &\quad + 2 \sqrt{\frac{\alpha}{m\la}}
 \sup_{\substack{\x,\z\in \T,\\ \norm{\x-\z}{2}= \rho}} \max_{I\subset [m], |I|\leq \alpha m/\la} \Big(\sum_{i \in I} \inner{\a_i}{\x-\z}^2 \Big)^{1/2}.
\end{align}
Applying Theorem~\ref{subgaussian_tail} for $k= \alpha m/\la$ shows
that the second summand on the right hand side above is bounded by $\frac{\rho^2}{64\la}$ with probability at least $1-2\exp(-c_1 m \alpha\log(\lambda/\alpha)/\lambda  )$ provided that \begin{equation}
    \label{eq:conditions_1}
    \sqrt{\frac{\alpha}{m\la}} c_2\; (w_*(\T-\T \cap \rho S^{n-1}) + \rho \sqrt{(\alpha/\lambda) m \log(\lambda/\alpha)}) \leq \frac{\rho^2}{64 \lambda} \; .
\end{equation}
Inequality \eqref{eq:conditions_1} holds, if 
$\alpha=c\rho/\sqrt{\log(\la/\rho)}$ for a constant $c>0$ small enough
and $m\gtrsim \la\rho^{-3} w_*(\T-\T \cap \rho S^{n-1})^2$.
Let us turn our attention to the first summand on the right hand side of 
\eqref{eq:thm:concencration_functional_two_summands}. Fix $\x'\in N_\eps$. For 
$\z \in \T$ and $Y_i = (\a_i,\tau_i,\nu_i)$ set 
$f_{\z}(Y_i) := \relu{-q(\x') (\inner{\a_i}{\z} + \tau_i)}$. 
Since $\eps\leq \rho$,
we clearly have
\begin{align}\label{eq:fixed_part}
    &\sup_{\substack{\x,\z\in \T,\norm{\x-\x'}{2}\leq \eps, \\ \norm{\x-\z}{2}= \rho}} \bigg|\frac{1}{m}\sum_{i=1}^m 
       \eps_i(\relu{-q(\x')_i(\inner{\a_i}{\z} + \tau_i)} -
     \relu{-q(\x')_i(\inner{\a_i}{\x} + \tau_i)})\bigg|\\ \nonumber
     &\leq \sup_{\z\in \T,  \norm{\z-\x'}{2}\leq 2 \rho} \bigg|\frac{1}{m}\sum_{i=1}^m 
       \eps_i f_{\z}(Y_i) 
    \bigg| + 
   \sup_{\x\in \T,  \norm{\x-\x'}{2}\leq \rho} \bigg|\frac{1}{m}\sum_{i=1}^m 
       \eps_i f_{\x}(Y_i)
    \bigg| .
\end{align}
The increments of the two processes above satisfy
\begin{align*}
    \left\|
    \frac{1}{m}\sum_{i=1}^m 
       \eps_i(f_{\mathbf{u}}(Y_i) -f_{\z}(Y_i)) \right\|_{\psi_2}
    \lesssim \frac{1}{\sqrt{m}}
    \|f_{\mathbf{u}}(Y_i) -f_{\z}(Y_i)\|_{\psi_2}
    \leq \frac{1}{\sqrt{m}} 
    \norm{\inner{\a_i}{\mathbf{u} - \z}}{\psi_2}
     \lesssim_L
    \frac{1}{\sqrt{m}} 
    \norm{\mathbf{u}-\z}{2}.
\end{align*}
Hence, the processes in \eqref{eq:fixed_part} are subgaussian with respect to a rescaled Euclidean metric. By a chaining argument (see e.g. \cite[Theorem 3.2]{dirksen_tail_2015}) and the majorizing measures theorem \cite{talagrand_upper_2014} we obtain that for any $t\geq 1$ with probability at least $1-2e^{-t}$, 
\begin{align*}
    \sup_{\z\in \T,  \norm{\z-\x'}{2}\leq 2 \rho} \bigg|\frac{1}{m}\sum_{i=1}^m 
       \eps_i f_{\z}(Y_i) 
    \bigg| + 
   \sup_{\x\in \T,  \norm{\x-\x'}{2}\leq \rho} \bigg|\frac{1}{m}\sum_{i=1}^m 
       \eps_i f_{\x}(Y_i)
    \bigg| 
     \lesssim \frac{w_*((\T-\T)\cap 2\rho B^n_2)}{\sqrt{m}} + \rho \sqrt{\frac{t}{m}} \; .
\end{align*}
Choosing $t = c\frac{\rho^2}{\lambda^2}m$, for $c>0$ an absolute constant small enough, 
a union bound together with inequality \eqref{eq:fixed_part} shows that with probability at least $1-\exp(-c m(\rho^2/\lambda^2))$,
\begin{align*}
    \max_{\x'\in N_\eps}\sup_{\substack{\x,\z\in \T, \\ \norm{\x-\z}{2}= \rho}} \bigg|\frac{1}{m}\sum_{i=1}^m 
       \eps_i(\relu{-q(\x')_i(\inner{\a_i}{\z} + \tau_i)} -
     \relu{-q(\x')_i(\inner{\a_i}{\x} + \tau_i)})\bigg|  \leq  
      \frac{\rho^2}{64\lambda}  \; 
\end{align*}
provided that
\begin{equation*}
m \gtrsim \lambda^2 \rho^{-4}  w_*((\T-\T)\cap \rho S^{n-1})^2 \quad \text{ and } \quad \log(|N_\eps|)\lesssim \frac{\rho^2m}{\la^2}.
\end{equation*}
Finally, by Lemma~\ref{lem:upper} the event $\mathcal{E}_\alpha$ 
satisfies $\Pr(
\mathcal{E}_\alpha^C) \leq 2\exp(-c_0 m \alpha /\lambda)$ provided that 
\begin{equation}
       m \gg \lambda \alpha^{-3} w_*((\T-\T)\cap \ep B^n_2)^2\quad \text{ and }\quad \log \N(\T,\ep) \lesssim \frac{\alpha}{\lambda}m
    \end{equation}
for a number $\eps>0$ with $\ep \lesssim \alpha/\sqrt{\log(C\lambda/\alpha)}$.
For $\alpha=c\rho/\sqrt{\log(\la/\rho)}$ these conditions are implied by \eqref{eq:conditions_concentration}. A union bound over the above events yields
\begin{align*}
    \Pr\Big( \sup_{\x, \z \in \T, \norm{\x-\z}{2}= \rho} |Z(\x, \z)|  > \frac{\rho^2}{32\la}  
   \Big) \le 6 e^{-c m(\rho^2/\lambda^2)}
\end{align*}{}
and thus proves the claim.
\end{proof}

\subsubsection{Proof of Theorem \ref{thm:noisy_recovery}} \label{sec:ProofOneBit}

We finally have all the necessary tools to prove Theorem \ref{thm:noisy_recovery}.

\begin{proof}[of Theorem \ref{thm:noisy_recovery}]
 Since every $\x\in \T$ is feasible for the program \eqref{eq:P}, any solution $\x^{\#}$ of \eqref{eq:P} satisfies $\L_{\qc}(\x^{\#})- \L_{\qc}(\x)\leq 0$. Therefore, 
if we can show that with high probability
\begin{equation}\label{eq:excess_risk}
     \inf_{\substack{\x,\z\in \T\\
      \norm{\x-\z}{2}=\rho \\ d_H(\qc, q(\x))\leq \beta m}}(\L_{\qc}(\z)- \L_{\qc}(\x) )>0 \; ,
\end{equation}
then on this event we have $\norm{\x-\x^\#}{2}\leq \rho$ for every $\x\in \T$ and every minimizer $\x^{\#}$ of \eqref{eq:P}. Indeed, if 
$\norm{\x-\x^{\#}}{2}>\rho$, then there exists $\z\in \T\cap \text{conv}(\x,\x^{\#})$ with 
$\norm{\x-\z}{2}=\rho$ ($\x,\x^{\#}\in \T$ and $\T$ is convex). From \eqref{eq:excess_risk} it follows $\L_{\qc}(\z)>\L_{\qc}(\x)$. Since $\z\mapsto \L_{\qc}(\z)$ is a convex function, we can conclude that $\L_{\qc}(\x^{\#})\geq \L_{\qc}(\z)>\L_{\qc}(\x)$, which contradicts that $\x^{\#}$ is a minimizer of $\L_{\qc}$ on $\T$. 

Next, we show that \eqref{eq:excess_risk}
holds with high probability. We distinguish two noise regimes, since these regimes require different arguments:
\begin{enumerate}
    \item[1)] $\norm{\nu}{L^2}> c_0 \rho/\sqrt{\log(C\lambda/\rho)}$ (high additive noise)
    \item[2)] $\norm{\nu}{L^2}\leq c_0 \rho/\sqrt{\log(C\lambda/\rho)}$ (low additive noise) 
\end{enumerate}
The argument requires us to use two different decompositions of the excess risk $(\L_{\qc}(\z)-\L_{\qc}(\x))$. The
bounds for the parts of these decompositions rely on different technical tools developed in this section.
We begin by studying the \textit{high additive noise} regime.
\par
\vspace{0.5cm}
\par
\textbf{1) High additive noise.}
For every $\x, \z\in \R^n$ we can write
\begin{align*}
    \L_{\qc}(\z)- \L_{\qc}(\x)
    &=
    (\L_{\qc}(\z)- \L_{\qc}(\x))-(\L_{q(\x)}(\z)- \L_{q(\x)}(\x))\\
    &\quad + \E (\L_{q(\x)}(\z)- \L_{q(\x)}(\x)) + (\L_{q(\x)}(\z)- \L_{q(\x)}(\x)) - \E (\L_{q(\x)}(\z)- \L_{q(\x)}(\x)) \; .
\end{align*}
If we can show that for all $\x,\z \in \T$ with $\norm{\x-\z}{2} = \rho$ we have
\begin{align*}
    (A) &:= |(\L_{\qc}(\z)- \L_{\qc}(\x))-(\L_{q(\x)}(\z)- \L_{q(\x)}(\x))| \leq \frac{\rho^2}{16\la}, \\
    (B) &:= \E (\L_{q(\x)}(\z)- \L_{q(\x)}(\x)) \geq \frac{\rho^2}{4\la}, \\
    (C) &:= |(\L_{q(\x)}(\z)- \L_{q(\x)}(\x)) - \E (\L_{q(\x)}(\z)- \L_{q(\x)}(\x))| \leq \frac{\rho^2}{8\la},
\end{align*}
then
\begin{equation}
     \inf_{\substack{\x,\z\in \T\\
      \norm{\x-\z}{2}=\rho \\ d_H(\qc, q(\x))\leq \beta m}}(\L_{\qc}(\z)- \L_{\qc}(\x) )
      \geq \frac{\rho^2}{16\la} > 0 \; .
\end{equation}
Let us observe that if $(\qc)_i = q(\x)_i$, then the $i$-th summand in $(\L_{\qc}(\z)- \L_{\qc}(\x))-(\L_{q(\x)}(\z)- \L_{q(\x)}(\x))$ is zero. Hence,
\begin{align*}
    |(\L_{\qc}(\z)&- \L_{\qc}(\x))-(\L_{q(\x)}(\z)- \L_{q(\x)}(\x))| \\
    &= 
    \Big| \frac{1}{m} \sum_{i : (\qc)_i \neq q(\x)_i} 
    \Big(\relu{-(\qc)_i (\inner{\a_i}{\z} + \tau_i)}
    - \relu{-(\qc)_i (\inner{\a_i}{\x} + \tau_i)} \\
    &\quad- (\relu{-q(\x)_i (\inner{\a_i}{\z} + \tau_i)}
    -\relu{-q(\x)_i (\inner{\a_i}{\x} + \tau_i)})\Big)
    \Big| \\
    &\leq \max_{|I| \leq \beta m} \frac{1}{m} \sum_{i \in I} |\inner{\a_i}{\x-\z}| \; ,
\end{align*}
where we used once more \eqref{eq:relu_triangle} in the last step. Since $\beta\in (0,1)$ satisfies $\beta\log(e/\beta)\leq c_2\rho/\la$, there exists $\beta'\in (0,1)$ with 
$\beta'\geq \beta$ such that $\beta'\log(e/\beta')= c_2\rho/\la$.
Hence, to bound the term (A) for all $\x,\z \in \T$ with $\norm{\x-\z}{2} = \rho$ we can make use of Lemma~\ref{lem:sup_max_subg}. Indeed, applying Lemma~\ref{lem:sup_max_subg} for $\alpha=\beta'$ shows that 
with probability at least $1-2\exp(-c_0 \rho m/\la  )$,
\begin{equation}
    \sup_{\substack{\x,\z \in \T \\ \norm{\x-\z}{2} = \rho}}
   \max_{|I| \leq \beta' m} \frac{1}{m} \sum_{i \in I} |\inner{\a_i}{\x-\z}| \leq \frac{\rho^2}{16\la} \; 
\end{equation}
if the constant $c_2>0$ is chosen small enough and $m\gtrsim \la \rho^{-3} w_*((\T-\T)\cap \rho \S^{n-1})^2$.
Consequently, on this event we have $(A)\leq \frac{\rho^2}{16\la}$ for all $\x, \z\in \T$ with $\norm{\x-\z}{2}=\rho$.
Term $(B)$ can directly be handled using Lemma~\ref{lem:exp_functional} which yields
\begin{equation}
    \E \left( \L_{\q(\x)}(\z)- \L_{\q(\x)}(\x) \right) \geq  \frac{\rho^2}{2\la}
    - 2 C \la \exp\Big(-c \frac{\la^2}{(L^2R^2 + \norm{\nu}{\psi_2}^2)}\Big)\\
\end{equation}
for any $\x, \z\in \T$ with
$\norm{\x-\z}{2}=\rho$. Since $\norm{\nu}{\psi_2}\lesssim \norm{\nu}{L^2}$, it follows that
$(B)\geq \frac{\rho^2}{4\la}$ for any $\x, \z\in \T$ with
$\norm{\x-\z}{2}=\rho$ if
$\la\gtrsim (R+\norm{\nu}{L^2})\sqrt{\log(\la/\rho)}$.
We are left with estimating the term (C). Theorem~\ref{thm:concencration_functional} provides the necessary bound. Indeed, if \eqref{eq:conditions_concentration}  holds and assuming that $\rho \lesssim \la$, then with probability at least $1-2\exp(-c' m(\rho/\la)^2)$ we find that for all $\x,\z \in \T$ such that
$\norm{\x-\z}{2}=\rho$,
\begin{equation}
    (C) \leq \frac{\rho^2}{8 \la} \; .
\end{equation}
This finishes the proof in the high additive noise regime.
\par
\vspace{0.5cm}
\par

\textbf{2) Low additive noise.}
Observe that we can decompose $\L_{\qc}(\z)- \L_{\qc}(\x)$ as follows:
\begin{align}
        \L_{\qc}(\z) - \L_{\qc}(\x) 
        &=\frac{1}{m}\sum_{i=1}^m  
       ( \relu{-(q_{\text{corr}})_i(\inner{\a_i}{\z} + \tau_i)} - \relu{-(q_{\text{corr}})_i(\inner{\a_i}{\x} + \tau_i)})
       \nonumber\\
        &=\frac{1}{m}\sum_{(q_{\text{corr}})_i= q(\x)_i} 
       \relu{-q(\x)_i(\inner{\a_i}{\z} + \tau_i)} -
       \frac{1}{m}\sum_{(q_{\text{corr}})_i= q(\x)_i} \relu{-q(\x)_i(\inner{\a_i}{\x} + \tau_i)} 
       \nonumber\\ 
       &\quad +\frac{1}{m}\sum_{(q_{\text{corr}})_i\neq q(\x)_i} 
       ( \relu{-(q_{\text{corr}})_i(\inner{\a_i}{\z} + \tau_i)} - \relu{-(q_{\text{corr}})_i(\inner{\a_i}{\x} + \tau_i)})
       \label{eq:three_summands}
\end{align}
Further, by applying the inequality 
\begin{equation*}
    \relu{a} - \relu{b}\geq - |a-b| , \quad a, b \in \R,
\end{equation*}
to the last summand of \eqref{eq:three_summands}
and by extending the sum in the second summand over all indices $i\in [m]$ we obtain the lower bound
\begin{align}
        \L_{\qc}(\z) - \L_{\qc}(\x) 
        &\geq 
        \frac{1}{m}\sum_{(q_{\text{corr}})_i= q(\x)_i} 
       \relu{-q(\x)_i(\inner{\a_i}{\z} + \tau_i)} -
       \underbrace{\frac{1}{m}\sum_{i=1}^m \relu{-q(\x)_i(\inner{\a_i}{\x} + \tau_i)}}_{=\L_{\q(x)}(\x)}
       \nonumber\\
       &\quad -\frac{1}{m}\sum_{(q_{\text{corr}})_i\neq q(\x)_i} 
        |\inner{\a_i}{\z-\x}| \; .
\end{align}
Introduce the random variable
    \begin{equation*}
        Y(\x):=\frac{1}{m} \sum_{i=1}^m 1_{\{\sign(\langle \a_i,\x \rangle + \nu_i + \tau_i)\neq \sign(\langle \a_i,\x \rangle + \tau_i)\}}|\nu_i| \; ,
    \end{equation*}
    and observe that 
    \begin{align*}
    \L_{\q(x)}(\x)
    &= 
     \frac{1}{m} \sum_{i=1}^m 
     \left[ -\sign(\langle \a_i,\x \rangle + \nu_i + \tau_i) (\langle \a_i,\x \rangle + \tau_i) \right]_+\\
     &=\frac{1}{m} \sum_{i=1}^m 1_{\{\sign(\langle \a_i,\x \rangle + \nu_i + \tau_i)\neq \sign(\langle \a_i,\x \rangle + \tau_i)\}}|\langle \a_i,\x \rangle + \tau_i|\\
     &\leq \frac{1}{m} \sum_{i=1}^m 1_{\{\sign(\langle \a_i,\x \rangle + \nu_i + \tau_i)\neq \sign(\langle \a_i,\x \rangle + \tau_i)\}}|\nu_i|\\ 
     &= Y(\x) \; .
    \end{align*}
Therefore, for every sequence $\qc$ and all $\x,\z \in \T$,
 \begin{align}
         \L_{\qc}(\z) - \L_{\qc}(\x) 
        &\geq \frac{1}{m}\sum_{(q_{\text{corr}})_i= q(\x)_i} 
       \relu{-q(\x)_i(\inner{\a_i}{\z} + \tau_i)} - Y(\x)
       -\frac{1}{m}\sum_{(q_{\text{corr}})_i\neq q(\x)_i} 
        |\inner{\a_i}{\z-\x}| 
        \nonumber\\
        &=\frac{1}{m}\sum_{(q_{\text{corr}})_i= q(\x)_i} 
       \relu{-q(\x)_i(\inner{\a_i}{\z} + \tau_i)} - \E Y(\x) - ( Y(\x) - \E Y(\x))
       \nonumber\\
       &\quad-\frac{1}{m}\sum_{(q_{\text{corr}})_i\neq q(\x)_i} 
        |\inner{\a_i}{\z-\x}| \label{eq:four_summands}\;.
\end{align}
%
%Consequently, 
%
%\begin{align*}
%     \inf_{\substack{\x,\z\in \T\\
%      \norm{\x-\z}{2}\geq \rho}}(\L_{\qc}(\z)- \L_{\qc}(\x) )
%      &\geq \inf_{\substack{\x,\z\in \T\\
%      \norm{\x-\z}{2}\geq \rho}} \Big( \frac{1}{m}\sum_{(q_{\text{corr}})_i= q(\x)_i} 
%       \relu{-q(\x)_i(\inner{\a_i}{\z} + \tau_i)} \Big)
%        - 
%       \sup_{\x\in \T}\E Y (\x) \\
%      &\quad - \sup_{\x\in \T} ( Y(\x) - \E Y(\x) )
%       - \sup_{\substack{\x,\z\in \T\\
%      \norm{\x-\z}{2}\geq \rho}}\frac{1}{m}\sum_{(q_{\text{corr}})_i\neq q(\x)_i} 
%        \relu{(q_{\text{corr}})_i(\inner{\a_i}{\z-\x}} \; .
%\end{align*}
%
Next, we will bound each of the four summands above. Regarding the first summand observe that for any $\theta>0$,
\begin{align*}
   \frac{1}{m}\sum_{(q_{\text{corr}})_i= q(\x)_i} 
       \relu{-q(\x)_i(\inner{\a_i}{\z} + \tau_i)}
       &\geq 
       \frac{1}{m}\sum_{\substack{(q_{\text{corr}})_i= q(\x)_i,\\ i\in I(\x, \z, \nnu, \theta) }} 
       \relu{-q(\x)_i(\inner{\a_i}{\z} + \tau_i)}\\
       &\geq 
       \frac{|\{i\in [m] \; : \; i \in I(\x, \z, \nnu, \theta) \text{ and } (q_{\text{corr}})_i= q(\x)_i \}|}{m}\, \theta \norm{\x-\z}{2}.
\end{align*}
Here we have used that by definition, if $i\in I(\x,\z, \nnu,\theta)$, then    
\begin{equation*}
        \relu{-q(\x)_i(\inner{\a_i}{\z} + \tau_i)}=
        |\inner{\a_i}{\z} + \tau_i|\geq \theta \norm{\x-\z}{2}.
\end{equation*}
By Theorem~\ref{thm:noisy_uniform} 
there exist constants $c, c', \theta>0$ such that with probability at least $1-2\exp(-c \rho m / \la)$ for all $\x,\z \in \T$ with $\norm{\x-\z}{2}= \rho$ we have
$|I(\x, \z, \nnu, \theta)|\geq c' \frac{\rho}{\la} m$. Hence, 
if $\beta\leq \frac{c' \rho }{2 \la}$, then on the same event for all $\x,\z \in \T$ with $\norm{\x-\z}{2}= \rho$ we have
$|\{i\in [m] \; : \; i \in I(\x, \z, \nnu, \theta) \text{ and } (q_{\text{corr}})_i= q(\x)_i \}|\geq \frac{c'}{2}\frac{\rho}{\la} m$ and therefore
\begin{equation}
\label{eq:lower_bound_functional}
    \frac{1}{m}\sum_{(q_{\text{corr}})_i= q(\x)_i} 
       \relu{-q(\x)_i(\inner{\a_i}{\z} + \tau_i)}
       \geq c'' \frac{\rho^2}{\la} \;.
\end{equation}
The second summand on the right hand side of 
\eqref{eq:four_summands} is bounded by $\norm{\nu}{L^2}^2/2\la$, see \eqref{eq:mean_Y}.

Further, Lemma~\ref{lem:concentration_functional} shows that 
\begin{equation}
\label{eq:concentration_Y}
    \sup_{\x\in \T} ( Y(\x) - \E Y(\x) )\leq \frac{c'' \rho^2}{4\lambda} \; 
\end{equation}
with probability at least $1 - 2 \exp(- c_0 m \rho /\lambda)$.
Regarding the fourth summand of \eqref{eq:four_summands} we can clearly estimate 
\begin{equation*}
 \frac{1}{m}\sum_{(q_{\text{corr}})_i\neq q(\x)_i} 
        |\inner{\a_i}{\z-\x}|
\leq \max_{I\subset [m], |I|\leq \beta' m} \frac{1}{m}\sum_{i\in I} 
       |\inner{\a_i}{\z-\x}|  ,      
\end{equation*}
where $\beta'\in (0,1)$ with 
$\beta'\geq \beta$ such that $\beta'\log(e/\beta')= c_2\rho/\la$.
Applying Lemma~\ref{lem:sup_max_subg} for $\alpha=\beta'$ shows that 
with probability at least $1-2\exp(-c_0 \rho m/\la  )$,
\begin{equation}\label{eq:adversarial_part}
    \sup_{\substack{\x,\z \in \T \\ \norm{\x-\z}{2} = \rho}}
   \max_{|I| \leq \beta' m} \frac{1}{m} \sum_{i \in I} |\inner{\a_i}{\x-\z}| \leq \frac{c''\rho^2}{4\la} \; 
\end{equation}
if the constant $c_2>0$ is chosen small enough and $m\gtrsim \la \rho^{-3} w_*((\T-\T)\cap \rho \S^{n-1})^2$.
In total, the events 
\eqref{eq:lower_bound_functional}, \eqref{eq:concentration_Y} and \eqref{eq:adversarial_part} occur at once with probability at least $1-6\exp(-c\rho m /\la)$ and 
on the intersection of theses events by \eqref{eq:four_summands} the following holds: for all $\x, \z\in \T$ with $\norm{\x-\z}{2}= \rho$,
\begin{align*}
  \L_{\qc}(\z) - \L_{\qc}(\x)\geq   c'' \frac{\rho^2}{\la} - \frac{1}{2\la}\norm{\nu}{L^2}^2 -
  \frac{c'' \rho^2}{4\lambda} - 
  \frac{c''\rho^2}{4\la}\geq \frac{c''\rho^2}{4\la},
\end{align*}
where the last inequality follows if $\norm{\nu}{L^2}\leq \sqrt{\frac{c''}{2}}\rho$. This finishes the proof in the low additive noise regime.

\end{proof}

%%%%%%%%%%%%%%%%%%%%%%%%%%%%%%%%%%%%%%%%%%%%%%%%%%%%%%%%%%%%%%%%%%%%%%%%%%%%%%
%%% Multi Bit
%%%%%%%%%%%%%%%%%%%%%%%%%%%%%%%%%%%%%%%%%%%%%%%%%%%%%%%%%%%%%%%%%%%%%%%%%%%%%%

\subsection{Proof of Theorem \ref{thm:MultiBitMain}}

%\begin{figure}
%    \centering
%    \includegraphics[width=0.8\textwidth]{Figures/TranslationInvariance.pdf}
%    \caption{Two points remain separated by the multi-bit measurement $\a_i$ when translated as long as the translation remains within the range of the quantizer. \red{(put to TikZ)}}
%    \label{fig:translation}
%\end{figure}

The proof is related to the 1-bit case. The main difference is that
in the multi-bit case we distinguish between reconstruction accuracies $\rho$
above and below the quantizer resolution $\Delta$:
\begin{itemize}
    \item[Case $\rho\gtrsim \Delta$:] The argument for the high quantizer resolution regime begins with Lemma \ref{lem:multi_bit_separation} which adapts Theorem \ref{thm:noise_well-separating} to the multi-bit setting. In the second step, Theorem \ref{thm:LocalMultiBit} extends Lemma \ref{lem:multi_bit_separation} uniformly to a localized signal set of size $\approx \Delta$. This localization prevents using a covering of $\T$ on resolution $\Delta$ and hence leads to less measurement complexity. In the last step, the basic property of a multi-bit quantizer 
%depicted in Figure~\ref{fig:translation}
is used: if for $\x, \y \in \T$ the projections $\inner{\a_i}{\x}$ and $\inner{\a_i}{\y}$ are at least $2\Delta$ apart, then $\x$ and $\y$ and any shift of the two points will be distinguished as well. Consequently, the uniform result on the localized set is extended to $\T$ in Theorem \ref{thm:MultiBit} and yields a simple proof of the first part of Theorem \ref{thm:MultiBitMain}.
\item[Case $\rho<\Delta$:] 
The argument for the second part of Theorem \ref{thm:MultiBitMain} is a tedious but conceptually straight-forward adaption of Theorem \ref{thm:noisy_recovery} and omitted here. Note that basically the dithering parameter $\lambda$ is replaced by the quantizer refinement $\Delta$.\\
\end{itemize}
Before we start, let us generalize the concept of well-separating hyperplanes from the one-bit to the multi-bit setting.

\begin{definition} \label{def:MultiBit_Separation}
    For $\theta>0$ we say that $\x,\y \in \R^n$ are $\theta$-well-separated in direction $\a_i\in \R^n$
    if there exists 
    $j \in \{-(2^{B-1}-1),...,(2^{B-1}-1)\}$ such that
    the hyperplane $H_{\a_i,\tau_i+j\Delta}$ $\theta$-well-separates $\x$ and $\z$ in the sense of Definition~\ref{def:well_sep}, i.e.,
    \begin{enumerate}
        \item $Q_{i,j}(\x) \neq Q_{i,j}(\y)$,
        \item $|\langle \a_i,\x \rangle + (\tau_i + j\Delta)| \ge \theta \norm{\x - \y}{2}$,
        \item $|\langle \a_i,\y \rangle + (\tau_i + j\Delta)| \ge \theta \norm{\x - \y}{2}$.
    \end{enumerate}
   Further, we set
   \begin{equation*}
       I(\x,\y,\theta):=\{i\in [m]\; : \; \x,\y \text{ are } \theta\text{-well-separated in direction } \a_i\}.
   \end{equation*}
\end{definition}

\subsubsection{Multi-Bit Compressive Sensing above the Quantizer Resolution}

Throughout the section we assume $\rho \gtrsim \Delta$. The following lemma, which adapts Theorem \ref{thm:noise_well-separating} to multi-bit 
measurements, is the core of the argument. It states that as long as two points 
have a distance of at least $C\Delta$ there is no influence of the resolution onto 
the number of measurements sufficient to have well-separation.

\begin{lemma} \label{lem:multi_bit_separation}
For $R>0$ let $\x,\y \in R B_2^n \subset \R^n$ which satisfy $\norm{\x-\y}{2}\geq 16\Delta$. Assume that $R\leq \frac{3}{64L^2}M_{\Delta, B}$ where 
$M_{\Delta,B}=(2^{B-1}-1)\Delta-\frac{\Delta}{2}$ characterizes the range of the quantizer. There exist constants $c_1, c_2>0$ 
such that with probability at least $1-2\exp(-c_1 m)$,
\begin{equation*}
    |I(\x,\y,\frac{1}{16})| \ge c_2 m.
\end{equation*}
\end{lemma}

To prove this lemma we rely on the relation $ q(\x) = Q(\x) \mathbbm{1} \frac{\Delta}{2}$. 
Indeed, for $\a, \x,\y\in \R^m$ the function $|q_{\mathcal{A}_{\Delta,B}}(\inner{\a}{\x}) - q_{\mathcal{A}_{\Delta,B}}(\inner{\a}{\y})|$ 
essentially (up to a factor $1/\Delta$) counts the number of hyperplanes with normal 
direction $\a$ separating the points $\x$ and $\y$. This observation is vital to 
the following argument.

\begin{lemma} \label{lem:multi_bit_sep}
    For $R>0$ let $\x,\y \in R B_2^n \subset \R^n$ which satisfy $\norm{\x-\y}{2}\geq \Delta$. Assume that $R\leq \frac{3}{64L^2}M_{\Delta, B}$ where 
$M_{\Delta,B}=(2^{B-1}-1)\Delta-\frac{\Delta}{2}$.
    Let $N(\x,\y,\a_i)$ denote the number of hyperplanes separating $\x$ and $\y$ 
    in direction $\a_i$, that is,
    \begin{equation*}
      N(\x,\y,\a_i):=\{
      j\in \{-(2^{B-1}-1),...,(2^{B-1}-1)\}\; : \;
      H_{\a_i,\tau_i+j\Delta} \text{ separates } \x \text{ and } \y\}  
    \end{equation*}
    and for $\beta>0$ define
    \begin{equation}
	    \label{eq:definition_J_set}
        J(\x,\y,\beta) := \{i \in [m] : 
        N(\x,\y,\a_i) \geq \beta \Delta^{-1} \norm{\x-\y}{2} \}.
    \end{equation}
    There are 
    constants $c_1,c_2 > 0$ such that
    with probability at least $1-2\exp(-c_1 m)$,
    \begin{equation}
        |J(\x,\y,\frac{1}{4})| \geq c_2 m.
    \end{equation}
\end{lemma}

In order to prove Lemma \ref{lem:multi_bit_sep}
we will make use of the following 
result, which appeared in similar form in \cite[Lemma A.1]{xu2018quantized}. For convenience a proof is included in Appendix~\ref{app:noise_stability}.

\begin{lemma} \label{lem:expectation_quantizer}
Let $\mathcal{A}_{\Delta,B}$ denote the quantization alphabet 
defined in \eqref{eq:Alphabet} and let
$q_{\mathcal{A}_{\Delta,B}}: \R 
\to \mathcal{A}_{\Delta,B}$ be the 
one-dimensional quantizer defined in \eqref{eq:Alphabet_Quantizer}.
Set $M_{\Delta,B}=(2^{B-1}-1)\Delta-\frac{\Delta}{2}$ and let
$\tau \sim \unif([-\Delta,\Delta])$.
Then
\begin{align}\label{eq:lem:expectation_quantizer}
        \E [q_{\mathcal{A}_{\Delta,B}}(x+\tau)] = x\quad \text{ for all } \; x\in [-M_{\Delta,B},M_{\Delta,B}]. 
    \end{align}
In particular, for all $x,y \in [-M_{\Delta,B},M_{\Delta,B}]$,
    \begin{align*}
        \E [| q_{\mathcal{A}_{\Delta,B}}(x+\tau) - q_{\mathcal{A}_{\Delta,B}}(y+\tau) |] = |x-y| \; .
    \end{align*}
\end{lemma}

\begin{proof}[of Lemma \ref{lem:multi_bit_sep}]
For $i\in [m]$, define the random variables $Z_i:=\frac{1}{\Delta}|q(\x)_i-
q(\y)_i|$. Then $Z_i=N(\x,\y,\a_i)$
and
by Lemma~\ref{lem:expectation_quantizer}, $$\E_{\tau_i} [Z_i] 
=\frac{1}{\Delta}\E_{\tau_i} [|q_{\mathcal{A}_{\Delta,B}}(\inner{\a_i}{\x}+\tau_i)-q_{\mathcal{A}_{\Delta,B}}(\inner{\a_i}{\y}+\tau_i)|]
= \frac{1}{\Delta} |\inner{\a_i}{\x} - \inner{\a_i}{\y}|$$ provided that 
$|\inner{\a_i}{\x}|,|\inner{\a_i}{\y}|\leq M_{\Delta,B}= (2^{B-1}-1)\Delta-\frac{\Delta}{2}$. We are interested in the independent events
\begin{equation}
    E_i = \{Z_i \geq \frac{1}{4}\Delta^{-1} \norm{\x-\y}{2}  \} \; .
\end{equation}
As a matter of fact, the result follows once we can show that $J(\x,\y,\frac{1}{4})=\sum_{i=1}^m \1_{E_i} \geq c_2 m$. 
To do so, we first restrict ourselves to measurements $\a_i$ 
which behave well in the sense that the projections of $\x$ and $\y$ onto $\a_i$ 
are bounded by $cR$ and that the hyperplane $H_{\a_i}$ is quite orthogonal to the 
connecting line between $\x$ and $\y$. Let us denote this set by
\begin{equation*}
I:=\Big\{i\in [m]\; : \; \max\{|\inner{\a_i}{\x}|,|\inner{\a_i}{\y}|\} \leq 
\frac{4R}{\sqrt{\ep}}\text{ and } |\inner{\a_i}{\x-\y}| \geq 
\gamma \norm{\x-\y}{2} \Big\},
\end{equation*} 
where $\eps \in (0,1)$ will be determined below. By the Paley-Zygmund inequality, 
for any $\gamma\in (0,1)$ we have
\begin{displaymath}
 \mathbb{P}(|\inner{\a_i}{\x-\y}|^2 \geq \gamma^2 
\norm{\inner{\a_i}{\x-\y}}{L^2}^2) \geq (1-\gamma^2)^2 
\frac{\norm{\inner{\a_i}{\x-\y}}{L^2}^4}{\norm{\inner{\a_i}{\x-\y}}{L^4}^4}.
\end{displaymath}
Since $\a_i, i=1,...,m$ are subgaussian we have $\norm{\inner{\a_i}{\x-\y}}{L^4} 
\leq 2 L\norm{\inner{\a_i}{\x-\y}}{L^2}$ and therefore by isotropy
\begin{displaymath}
 \mathbb{P}(|\inner{\a_i}{\x-\y}| \geq \gamma 
\norm{\x-\y}{2}) \geq (1-\gamma^2)^2 \frac{1}{(2L)^4} \; .
\end{displaymath}
Setting $\ep:=(1-\gamma^2)^2 \frac{1}{(2L)^4}\in (0,1)$ we obtain
\begin{equation} \label{eq/small_ball_condition}
  \mathbb{P}(|\inner{\a_i}{\x-\y}| \geq \gamma 
\norm{\inner{\a_i}{\x-\y}}{L^2}) \geq \ep .
\end{equation}
Further, note that the subgaussian random variables $\inner{\a_i}{\x}$ 
are square integrable. Hence, by isotropy and Chebychev's inequality, the 
following holds for $\ep' \in (0,1)$
\begin{displaymath}
 \mathbb{P}(|\inner{\a_i}{\x}| \geq \norm{\x}{2} /\sqrt{\ep'}) \leq \ep' \; .
\end{displaymath}
Considering the $m$ independent rows $\a_i$ of $\A$ a Chernoff bound reveals that
\begin{displaymath}
 | \{ i \in [m] \; : \; |\inner{\a_i}{\x}| \geq \norm{\x}{2}/\sqrt{\ep'} 
\}| 
\leq 2 \ep' m
\end{displaymath}
holds with probability at least $1 - 2 \exp(-c_1 \ep' m)$. 
Applying the small ball condition in \eqref{eq/small_ball_condition} to 
a sequence $(\a_i)_{i \in [m]}$ another Chernoff bound yields, that with probability 
at least $1- 2 \exp(-c_1\ep m)$ we have
\begin{equation}
 |\{i \in [m] \; : \; |\inner{\a_i}{\x-\y}| \geq \gamma 
\norm{\x-\y}{2}\}| \geq \frac{\ep m}{2} \; .
\end{equation}
Set $\ep'= \ep/16$. Then with probability at least $1-4\exp(-c_1\eps m)$,
\begin{align*}
    |I|&\geq |\{i\in [m]\; : \; |\inner{\a_i}{\x-\y}| \geq 
\gamma \norm{\x-\y}{2}\}|\\
&\quad -|\{i\in [m]\; :\; |\inner{\a_i}{\x}|>\frac{4R}{\sqrt{\eps}}  \}|
    -|\{i\in [m]\; : \;|\inner{\a_i}{\y}|>\frac{4R}{\sqrt{\eps}} \}|\\
    &\geq |\{i\in [m]\; : \; |\inner{\a_i}{\x-\y}| \geq 
\gamma \norm{\x-\y}{2}\}|\\
&\quad -|\{i\in [m]\; :\; |\inner{\a_i}{\x}|>\frac{\norm{\x}{2}}{\sqrt{\eps'}}  \}|
    -|\{i\in [m]\; : \;|\inner{\a_i}{\y}|>\frac{\norm{\y}{2}}{\sqrt{\eps'}} \}|\\
    &\geq \frac{\eps m}{2} - 4\eps' m =\frac{\eps m}{4}.
\end{align*}
Set $\gamma=\frac{1}{2}$ which implies $\eps = \frac{9}{256 L^4}$. By the above, $|I|\geq c_2 m$ with probability at least $1-2\exp(-c_1m)$ where $c_1, c_2>0$ are constants. In the following, we condition on the event $|I|\geq c_2m$.
By the Payley-Zygmund inequality,  
 \begin{equation} \label{eq:Ptau}
  \mathbb{P}_{\tau_i}\Big(Z_i \geq \frac{1}{2} \E_{\tau_i} [Z_i] \Big) \geq \frac{1}{4} 
  \frac{(\E_{\tau_i} [Z_i])^2}{\E_{\tau_i} [Z_i^2]} \; .
 \end{equation}
Since $|q(\x)_i-(\inner{\a_i}{\x}+\tau_i)|\leq \frac{\Delta}{2}$ for every $\x\in \R^n$, it follows that
\begin{equation}
    |q(\x)_i-q(\y)_i- (\inner{\a_i}{\x}-\inner{\a_i}{\y})|\leq \Delta
\end{equation}
and 
\begin{equation}
    |q(\x)_i-q(\y)_i|
    \leq |\inner{\a_i}{\x}-\inner{\a_i}{\y}|+\Delta
\end{equation}
for all $\x, \y\in \R^n$.
Therefore,
\begin{equation}\label{eq:bound_Z}
    Z_i\leq \frac{|\inner{\a_i}{\x}-\inner{\a_i}{\y}|}{\Delta} +1.
\end{equation}
Let $\frac{4R}{\sqrt{\eps}}\leq M_{\Delta, B}$. Then $\Delta\E_{\tau_i} [Z_i]=|\inner{\a_i}{\x-\y}|$ for every $i\in I$. 
Hence, for every $i\in I$, 
\begin{align*}
   \mathbb{P}_{\tau_i}\Big(Z_i \geq \frac{1}{2} \frac{\gamma \norm{\x-\y}{2}}{\Delta} \Big)
   \geq\mathbb{P}_{\tau_i}\Big(Z_i \geq \frac{1}{2} \frac{1}{\Delta} |\inner{\a_i}{\x} - \inner{\a_i}{\y}| \Big)
   &=\mathbb{P}_{\tau_i}\Big(Z_i \geq \frac{1}{2} \E_{\tau_i} [Z_i] \Big)\\
   &\geq \frac{1}{4} 
  \frac{(\E_{\tau_i} [Z_i])^2}{\E_{\tau_i} [Z_i^2]}.
\end{align*}
Using \eqref{eq:bound_Z} and $| \inner{\a_i}{\x} - \inner{\a_i}{\y} | \ge \gamma \norm{\x - \y}{2} \ge \frac{\Delta}{2}$ we obtain 
\begin{equation}
    Z_i^2\leq Z_i \Big(\frac{|\inner{\a_i}{\x}-\inner{\a_i}{\y}|+\Delta}{\Delta} \Big) 
    \le Z_i \Big(\frac{3 |\inner{\a_i}{\x}-\inner{\a_i}{\y}|}{\Delta} \Big)
\end{equation}
which implies for every $i\in I$,
\begin{equation}
   \frac{(\E_{\tau_i} [Z_i])^2}{\E_{\tau_i} [Z_i^2]}\geq \frac{\Delta\E_{\tau_i} [Z_i]}{3 |\inner{\a_i}{\x}-\inner{\a_i}{\y}|} = \frac{1}{3}.
\end{equation}
In total, if 
$R\leq \frac{3}{64L^2}M_{\Delta, B}$, $\Delta\leq \norm{\x-\y}{2}$, then
for every $i\in I$,
\begin{equation}
   \mathbb{P}_{\tau_i}\Big(Z_i \geq  \frac{ \norm{\x-\y}{2}}{4\Delta} \Big)\geq \frac{1}{12}.
\end{equation}
By the Chernoff bound with $\ttau$-probability at least 
$1-\exp(-c_1|I|)$,
\begin{equation}
    \sum_{i\in I} 1_{ \left\{Z_i \geq \frac{ \norm{\x-\y}{2}}{4\Delta} \right\} } \geq  \frac{|I|}{24} \; .
\end{equation}
Hence, conditioned on the event $|I|\geq c_2m$ (which occurs with probability at least $1-2\exp(-c_1m)$), with $\ttau$-probability at least $1-\exp(-c_1c_2m)$,
\begin{equation}
 \sum_{i=1}^m 1_{ \left\{ Z_i \geq \frac{ \norm{\x-\y}{2}}{4\Delta} \right\} }\geq
   \sum_{i\in I} 1_{ \left\{ Z_i \geq \frac{ \norm{\x-\y}{2}}{4\Delta} \right\} }\geq \frac{c_2 m}{24}.
\end{equation}
This shows the result.
\end{proof}

\begin{proof}[of Lemma \ref{lem:multi_bit_separation}]
If $\norm{\x-\y}{2}\geq 16\Delta$, then
\begin{equation}
    \{i\in [m] \; : \; 
    N(\x,\y,\a_i) \geq \frac{1}{4}  \Delta^{-1} \norm{\x-\y}{2}\}\subset I(\x,\y,\frac{1}{16}).
\end{equation}
Indeed,
if $N(\x,\y,\a_i) \geq \frac{1}{4}  \Delta^{-1} \norm{\x-\y}{2} \geq 3$, then 
there exists $j\in \{-(2^{B-1}-1),...,(2^{B-1}-1)\}$ such that 

\begin{equation}
    \sign(\inner{\a_i}{\x}+\tau_i+j\Delta)\neq \sign(\inner{\a_i}{\y}+\tau_i+j\Delta),
\end{equation}

\begin{equation}
    |\inner{\a_i}{\x}+\tau_i+j\Delta| \ge \frac{|\inner{\a_i}{\x}-\inner{\a_i}{\y}| - \Delta}{2},
    %|\inner{\a_i}{\x}-\inner{\a_i}{\y}|\leq 2|\inner{\a_i}{\x}+\tau_i+j\Delta| +\Delta,
\end{equation}
and
\begin{equation}
    |\inner{\a_i}{\y}+\tau_i+j\Delta| \ge \frac{|\inner{\a_i}{\x}-\inner{\a_i}{\y}| - \Delta}{2}.
    %|\inner{\a_i}{\x}-\inner{\a_i}{\y}|\leq 2|\inner{\a_i}{\y}+\tau_i+j\Delta| +\Delta.
\end{equation}
Therefore
\begin{align}
  &2|\inner{\a_i}{\x}+\tau_i+j\Delta| +\Delta
  \geq |\inner{\a_i}{\x}-\inner{\a_i}{\y}|\geq (N(\x,\y,\a_i)-1)\Delta\geq \frac{1}{4}\norm{\x-\y}{2}-\Delta.
\end{align}
Hence,
\begin{equation}
|\inner{\a_i}{\x}+\tau_i+j\Delta|\geq \frac{1}{8}\norm{\x-\y}{2}-\Delta\geq \frac{1}{16}\norm{\x-\y}{2}
\end{equation}
as $\norm{\x-\y}{2} \ge 16\Delta$. Analogously, 
\begin{equation}
    |\inner{\a_i}{\y}+\tau_i+j\Delta|\geq \frac{1}{16}\norm{\x-\y}{2}.
\end{equation}
The result now follows from 
Lemma~\ref{lem:multi_bit_sep}.
\end{proof}

\paragraph{Localized uniform result}
For $\rho>0$ define the localized set $\T_\rho = (\T-\T)\cap \rho B_2^n$. By arguing similarly to the one-bit case, we can obtain the following uniform result.

\begin{theorem} \label{thm:LocalMultiBit}
    There exist constants $c_\ast, c_1, c_2\in (0,1)$ such that the follwing holds. 
    Let $\T\subset \R^n$ be convex.
    For every $\rho\geq \Delta$ and every $c\leq c_\ast$, 
    if
    \begin{align}\label{eq:meas_assumptions_thm:LocalMultiBit}
        m \gtrsim \rho^{-2} w_*(\T_{c\rho})^2 + \log(\mathcal{N}(\T_{17\rho}, c\rho))
    \end{align}
    then with probability at least $1-2e^{-c_1m}$, 
    \begin{align*}
     \inf_{\substack{\x,\y \in T_{17\rho}\\\norm{\x - \y}{2} \ge 17 \rho}}|I(\x,\y,\frac{1}{32})| \ge c_2 m.
    \end{align*}
\end{theorem}

\begin{proof}
For $c\in (0, \frac{1}{2}]$ 
let $N_{c \rho}\subset \T_{17\rho}$ be a minimal $c\rho$-net of $\T_{17\rho}$.
By Lemma~\ref{lem:multi_bit_separation} and
a union 
bound we obtain that with probability at least $1 - |N_{c\rho}|^2 2 e^{-c_1m}$
the following holds:
for all $\x',\y' 
\in N_{c\rho}$ such that $\norm{\x'-\y'}{2}\geq 16 \Delta$,
    \begin{align*}
        |I(\x',\y',\frac{1}{16})| \ge c'm.
    \end{align*}
Let $\x, \y\in \T_{17\rho}$ such that $\norm{\x-\y}{2} \geq 17\rho$. There exist $\x', \y' \in \N_{c\rho}$ with 
$\norm{\x-\x'}{2}\leq c\rho$ and $\norm{\y-\y'}{2}\leq c\rho$. 
By the triangle inequality, 
\begin{equation}
\norm{\x'-\y'}{2}\geq \norm{\x-\y}{2}-2c\rho\geq
17\rho-\rho= 
16\rho\geq 16\Delta.
\end{equation}
By Lemma~\ref{lem:noise_stability},
\begin{align*}
    |I(\x,\y,\frac{1}{32})|
    &\geq |I(\x',\y',\frac{1}{16})| - 2\sup_{\z\in (\T_{17\rho}-\T_{17\rho})\cap c\rho B^n_2}|\{i\in [m]\; : \; |\inner{\a_i}{\z}|>\frac{\rho}{4} \}.\\
\end{align*}
By Lemma~\ref{lem:bound_for_spoiled_hyperplanes},
if
\begin{equation} \label{eq:meas_LocalizedArgument}
   c_1 \frac{1}{\sqrt{m}}\Big(w_*((\T_{17\rho}-\T_{17\rho})\cap c\rho B^n_2) +  c\rho \sqrt{m \log(e)} \Big) \leq  \frac{1}{4} \rho \; ,
\end{equation}
then 
\begin{equation}
\sup_{\z\in (\T_{17\rho}-\T_{17\rho})\cap c\rho B^n_2}|\{i\in [m]\; : \; |\inner{\a_i}{\z}|>\frac{1}{4} \rho \} | \leq \frac{c'}{4}m
\end{equation}
with probability at least $1-2e^{-c_2 m}$. 
Observe that since $\T$ is convex, the set $\T-\T$ is symmetric and convex. Therefore, $\T_{17\rho}-\T_{17\rho}\subset 2\T_{17\rho}$ which implies $(\T_{17\rho}-\T_{17\rho})\cap c\rho B^n_2\subset 2\T_{17\rho}\cap c\rho B^n_2=2\T_{\frac{c}{2}\rho}\subset 2\T_{c\rho}$. Hence, if $c\leq \frac{1}{2}$ is chosen small
enough, then \eqref{eq:meas_assumptions_thm:LocalMultiBit} implies \eqref{eq:meas_LocalizedArgument}.
The result now follows on the intersection of the two events of Lemma \ref{lem:noise_stability} \& \ref{lem:bound_for_spoiled_hyperplanes}.

\begin{comment}
First note that if $\frac{c_1}{\sqrt{m}} (w_*(T_\Delta) + \Delta \sqrt{m}) \le \rho$,
for $\rho > 0$, then with probability at least $1 - e^{-c_0m}$ one has by 
Lemma~\ref{lem:bound_for_spoiled_hyperplanes} that
    \begin{align*}
        \sup_{\z \in T_\Delta} |\{ i \in [m] \colon |\langle \a_i,\z \rangle| 
	\ge \rho \}| \le \frac{c'}{4} m.
    \end{align*}
Moreover, if $N_\Delta$ is a $2\Delta$-net of $T_\Delta$, we get by a simple union 
bound that with probability at least $1 - |N_\Delta|^2 e^{-cm}$, for all $\x',\y' 
\in N_\Delta$ such that $\norm{\x'-\y'}{2}\geq 12 \Delta$,
    \begin{align*}
        |I(\x',\y',\frac{\Delta}{\norm{\x'-\y'}{2}})| \ge c'm.
    \end{align*}
We can now apply Lemma \ref{lem:noise_stability} with $\rho = 12\Delta$ and $\eps = 3\Delta$ 
to get for all $\x',\y' \in N_\Delta$ with $\norm{\x' - \y'}{2} \ge 12\Delta$ and 
all $\x,\y$ with $\norm{\x - \x'}{2} \le 3\Delta$ and $\norm{\y - \y'}{2} \le 3\Delta$ 
that
    \begin{align*}
        |I(\x,\y,\frac{\Delta}{2\norm{\x-\y}{2}})| &\ge |I(\x',\y',\frac{\Delta}{\norm{\x-\y}{2}})| - \frac{c'}{4} m \\
        &\ge  |I(\x',\y',\frac{\Delta}{\norm{\x'-\y'}{2}})| - \frac{c'}{4} m \\
        &\geq \frac{c'}{2}m.
    \end{align*}
which yields the claim by a union bound over the two above events.
\end{comment}

\end{proof}

\paragraph{Extension to the whole set} We will use now the elementary properties 
of our uniform multi-bit quantizer to extend Theorem \ref{thm:LocalMultiBit} to 
the complete signal set $T$.

\begin{theorem} \label{thm:MultiBit}
    There exist constants $c_1, c_2\in (0,1)$ and $\Gamma\geq 1$ 
such that the following holds. 
Let $\T\subset R B^n_2$ be a convex set. Assume that $\Delta > 0$ and $B\in \mathbb{N}$ are chosen such that $\Gamma R\leq (2^{B-1}-1)\Delta - \frac{\Delta}{2}$. 
    For every $\rho \geq 2\Delta$, if
    \begin{align*}
        m \gtrsim \rho^{-2} w_*(\T_{17\rho})^2 + R^{-2} w_*(\T)^2
    \end{align*}
    then with probability at least $1 - 2 e^{-c_1 m}$ 
    the following holds:
    for all $\x,\y \in \T$ 
    with $\norm{\x - \y}{2} \ge 17 \rho$ we have
    \begin{align*}
        \left| I\left(\x,\y,\frac{\rho}{4 \norm{\x - \y}{2}} \right) \right| 
	\ge c_2 m.
    \end{align*}
\end{theorem}

\begin{proof}[of Theorem \ref{thm:MultiBit}]
    It suffices to show that 
    with probability at least $1 - 4e^{-c_1 m}$ the following holds:
    for all $\x, \z \in \T$ with $\norm{\x - \z}{2} = 17 \rho$,
    \begin{equation}
       \left| I\left(\x,\z,\frac{\rho}{4 \norm{\x - \z}{2}} \right) \right| = \left| I\left(\x,\z,\frac{1}{ 68} \right) \right|
	\ge c_2 m. 
    \end{equation}
    Indeed, if $\x, \y\in \T$ with $\norm{\x - \y}{2} > 17 \rho$, then by convexity of $\T$ there exists $\z\in \T$ on the connecting line between $\x$ and $\y$ such that $\norm{\x - \z}{2} = 17 \rho$. 
    Since $\inner{\a_i}{\z}$
    then also lies between $\inner{\a_i}{\x}$ and $\inner{\a_i}{\y}$ for every $i\in [m]$, it follows that 
    \begin{equation}
     I\left(\x,\z,\frac{1}{68} \right) \subset I\left(\x,\y,\frac{\rho}{4 \norm{\x - \y}{2}}\right).
    \end{equation}
    By Theorem \ref{thm:LocalMultiBit}     there exist constants $c, c_1, c_2\in (0,1)$ 
    such that 
    with probability at least $1-2e^{-c_1 m}$,
    \begin{align*}
        \inf_{\substack{\x, \z\in \T\\
        \norm{\x-\z}{2}=17\rho}}\left| I \left( \0,\x-\z,\frac{1}{32} \right)\right| \ge c_2m,
    \end{align*}
    provided that $m \gtrsim \rho^{-2} w_*(\T_{c\rho})^2 + \log(\mathcal{N}(\T_{17\rho}, c\rho))$. 
    Sudakov's inequality implies that this condition on $m$ is satisfied if $m\gtrsim \rho^{-2}w_*(\T_{17\rho})^2$. 
    By Lemma \ref{lem:bound_for_spoiled_hyperplanes}, if $m \gtrsim \Gamma^{-2}R^{-2}w_*(\T)^2$,
    where $\Gamma \ge 1$ is a suitable constant that only depends on $L$, 
    then with probability at least $1 - 2e^{-c_3 m}$,
    \begin{align*}
        \left| \left\{ i \in [m] \colon \sup_{\w \in \T} |\langle \a_i,\w \rangle| > \Gamma R \right\}\right| < \frac{c_2}{2} m.
    \end{align*}{}
    A union bound shows that with probability at least $1 - 4e^{-c_4 m}$,
    for all $\x, \z\in \T$ with $\norm{\x-\z}{2} = 17 \rho$
    the set
    \begin{align*}
        \mathcal{I}_{\x,\z} = I \left( \0,\x-\z,\frac{1}{32} \right) \cap \left\{ i \in [m] 
	\colon \sup_{\w \in \T} |\langle \a_i,\w \rangle| \le \Gamma R \right\}
    \end{align*}
    satisfies $|\mathcal{I}_{\x,\z}| \geq \frac{c_2}{2}m$. It remains to show that $\mathcal{I}_{\x,\z} \subset I\left(\x,\z,\frac{1}{68} \right)$. By definition, for every $i \in \mathcal{I}_{\x,\z}$ there exists $j_i \in \{-(2^{B-1}-1),...,(2^{B-1}-1)\}$ such that
    \begin{itemize}
        \item $Q_{i,j_i}(\0) \neq Q_{i,j_i}(\x-\z)$,
        \item $|\langle \a_i,\x - \z \rangle + (\tau_i + j_i \Delta)| 
		\ge \frac{1}{32} \norm{\x-\z}{2}=\frac{17}{32}\rho\geq \frac{\rho}{2}$,
        \item $|\tau_i + j_i \Delta| \ge \frac{1}{32} \norm{\x-\z}{2}=\frac{17}{32}\rho\geq \frac{\rho}{2}$.
    \end{itemize}
    This implies for all $i \in \mathcal{I}_{\x,\z}$ that
    \begin{align} \label{eq:distance}
        |\langle \a_i,\x \rangle - \langle \a_i,\z \rangle| 
	    = |\langle \a_i,\x - \z \rangle + (\tau_i + j_i\Delta) 
		- (\tau_i + j_i\Delta)| \ge \rho\geq 2\Delta,
    \end{align}
    where we used that $\sign(\langle \a_i,\x - \z \rangle + (\tau_i + j_i \Delta)) \neq \sign(\tau_i + j_i \Delta)$. 
    If $\Gamma R\leq M_{B;\Delta}$
    then for every 
    $i\in \mathcal{I}_{\x,\z}$ we have  \eqref{eq:distance} together with 
    $\max\{|\langle \a_i,\x \rangle|, |\langle \a_i,\z \rangle|\} \leq M_{B;\Delta}$ and the quantizer resolution $\Delta$ which implies 
    that there exists $j_i'\in \{-(2^{B-1}-1),...,(2^{B-1}-1)\}$ such that
    \begin{itemize}
        \item $Q_{i,j_i'}(\x) \neq Q_{i,j_i'}(\z)$,
        \item $|\langle \a_i,\x \rangle + (\tau_i + j_i' \Delta)| \ge \frac{\rho}{4}=\frac{1}{68}\norm{\x-\z}{2}$,
        \item $|\langle \a_i,\z \rangle + (\tau_i + j_i' \Delta)| \ge \frac{\rho}{4}=\frac{1}{68}\norm{\x-\z}{2}$.
    \end{itemize}
    %To see this, assume the contrary and check that this would imply that the projected values between $\langle \a_i,\x \rangle + \tau_i + \frac{\rho}{4}$ and $\langle \a_i,\z \rangle + \tau_i - \frac{\rho}{4}$ (assuming w.l.o.g. that $\langle \a_i,\x \rangle < \langle \a_i,\z \rangle$) cannot be distinguished by the quantizer contradicting the resolution of $\Delta$ as $|\langle \a_i,\x \rangle + \tau_i + \frac{\rho}{4} - \langle \a_i,\z \rangle + \tau_i - \frac{\rho}{4}| \ge \frac{3}{2} \rho \ge \frac{3}{2} \Delta$.\\
    This implies that $\mathcal{I}_{\x, \z} \subset I\left(\x,\z,\frac{1}{68} \right)$ and hence proves the claim.
\end{proof}

\begin{proof}[of Theorem \ref{thm:MultiBitMain}]
Let $\rho \geq 2\Delta$.
    By Theorem~\ref{thm:MultiBit}, if
    \begin{align*}
        m \gtrsim \rho^{-2} w_*(\T_{17\rho})^2 + R^{-2} w_*(\T)^2
    \end{align*}
    then the following holds
    with probability at least $1 - 4e^{-c_1 m}$: 
    for all $\x,\y \in \T$ 
    with $\norm{\x - \y}{2} \ge 17 \rho$ we have
    \begin{align*}
        \left| I\left(\x,\y,\frac{\rho}{4 \norm{\x - \y}{2}} \right) \right| 
	\ge c_2 m.
    \end{align*}
    On the same event for all $\x,\y \in \T$ 
    with $\norm{\x - \y}{2} \ge 17 \rho$,
    \begin{align*}
        \L_{Q(\x)}(\y) - \L_{Q(\x)}(\x) = \L_{Q(\x)}(\y) 
        &\geq
        \frac{1}{m} \sum_{i\in I\left(\x,\y,\frac{\rho}{4 \norm{\x - \y}{2}} \right)}\sum_{j = -(2^{B-1}-1)}^{(2^{B-1}-1)} \left[ -Q(\x)_{i,j} (\langle \a_i,\y \rangle + (\tau_i + j\Delta)) \right]_+\\
        &\geq \frac{1}{m} \sum_{i\in I\left(\x,\y,\frac{\rho}{4 \norm{\x - \y}{2}} \right)}\frac{\rho}{4}\\
        &\geq \frac{c_2\rho}{4}>0.
    \end{align*}{}
    In particular, on this event the following holds: 
    for all $\x\in \T$ and  
every minimizer $\x^\#$ of the program \eqref{eq:Pmulti} 
with $Q(\x)_{i,j} = \sign (\langle \a_i,\x \rangle + \tau_i + j\Delta)$,
 $$\norm{\x-\x^\#}{2}\leq 17 \rho.$$
 The result follows by rescaling $\rho$.
\end{proof}

\section{Numerical Experiments} \label{sec:Numerics}
In this section we present different numerical simulations substantiating the theoretical considerations above. After discussing beneficial properties of the functional $\L$ with respect to efficient minimization and briefly introducing the competitors we use to compare its performance, we present simulations both in the one-bit setting analyzed in Section \ref{sec:OneBit} and the multi-bit setting analyzed in Section \ref{sec:MultiBit}.

\paragraph{Formulation as a linear program} For signal sets $\T$ which are convex polytopes, the minimization in \eqref{eq:P} becomes a linear program and can be solved more efficiently than general convex programs. We illustrate this by choosing $\T = RB^n_1$. First note that the ReLU-function can be written as
\begin{align*}
    [z]_+ = \frac{z + |z|}{2}.
\end{align*}{}
Consequently, by introducing the vectors $\u_+,\u_- \in \R^n$ and $\w_+,\w_- \in \R^m$ we may re-formulate
\begin{align*}
    \min_{\norm{\z}{1} \le R} \frac{1}{m} \sum_{i=1}^m \relu{-(q_{\text{corr}})_i(\inner{\a_i}{\z} + \tau_i)}  
    &= \min_{\substack{\u_+,\u_- \ge \0 \\ [\u_+^T,\u_-^T] \1 \le R}} \frac{1}{m} \sum_{i=1}^m \relu{-(q_{\text{corr}})_i(\inner{\a_i}{\u_+ - \u_-} + \tau_i)} \\
    &= \min_{\substack{\u_+,\u_-,\w_+,\w_- \ge \0 \\ [\u_+^T,\u_-^T] \cdot \1 \le R \\ \w_+ - \w_- = \left( -(q_{\text{corr}})_i(\inner{\a_i}{\u_+ - \u_-} + \tau_i) \right)_{i\in [m]} }} \1^T \left( \frac{(\w_+ - \w_-) + |\w_+ - \w_-|}{2} \right) \\
    &= \min_{\substack{\u_+,\u_-,\w_+,\w_- \ge \0 \\ [\u_+^T,\u_-^T] \cdot \1 \le R \\ \w_+ - \w_- = \left( -(q_{\text{corr}})_i(\inner{\a_i}{\u_+ - \u_-} + \tau_i) \right)_{i\in [m]} }} \1^T \w_+
\end{align*}{}
which is a linear program in $\R^{2n+2m}$. In the last step we used that by shape of the objective function corresponding entries of the minimizing $\w_+$ and $\w_-$ are never simultaneously non-zero, \emph{i.e.}, if $(\u_+,\u_-,\w_+,\w_-)$ minimizes the program in the second line then $(w_+)_i = 0$ or $(w_-)_i = 0$, for all $i \in [m]$. The same applies to the multi-bit generalization of $\L$ presented in Section \ref{sec:MultiBit}.\\

Let us now discuss the competing algorithms we use to evaluate the performance of our program:

\paragraph{Single back-projection} As already mentioned in the introduction, the program \eqref{eq:Dirksen} has been shown to approximate signals from an near-optimal number of noisy one-bit measurements. Moreover, the minimization is equivalent to performing a single projected back-projection, \emph{i.e.},
\begin{align} \label{eq:SingleProjection}
    \x^\# = \P_\T \left( \frac{\lambda}{m} \A^T \q \right).
\end{align}{}
In the case of $\T$ being the set of $s$-sparse vectors, the projection becomes a simple hard-thresholding step which in spite of the non-convexity of \eqref{eq:Dirksen} is fast to compute. Since we are not able to perform the minimization of \eqref{eq:P} on a non-convex set and are thus forced to use its convex relaxation $\sqrt{s} B_1^n$ (here the set of $s$-sparse vectors has been restricted to the unit ball $B_2^n$ before taking the convex hull), we as well consider \eqref{eq:SingleProjection} with $\T = \sqrt{s} B_1^n$ to have a fairer comparison. We transfer \eqref{eq:SingleProjection} to the multi-bit setting by replacing $\q$ with its multi-bit version in \eqref{eq:Alphabet}, a setting which has been examined in \cite{xu2018quantized} for infinite uniform alphabets.

\paragraph{Iterative Thresholding} More sophisticated recovery schemes are adapted iterative thresholding algorithms which still lack thorough theoretical analysis, but perform exceedingly well in practice. Having been introduced in \cite{jacques_robust_2013,jacques_quantized_2013} as \emph{binary iterative hard-thresholding (BIHT)} for one-bit quantized measurements and \emph{quantized iterative hard-thresholding (QIHT)} for multi-bit quantized measurements, both algorithms follow the same concept of iterating between gradient descent steps of an $\ell_2$-data fidelity term (including projection of the involved quantities to the quantization alphabet) and hard-thresholding projections onto the set of $s$-sparse vectors, \emph{i.e.},
\begin{align} \label{eq:IHT}
    \x^{k+1} = \P_\T \left( \x^k + \mu \A^T ( \q - \P_{\mathcal{A}_{\Delta,B}^m} (\A\x^k) ) \right),
\end{align}{}
where $\T$ is chosen as the set of $s$-sparse vectors, $\mu > 0$ is a step-size parameter, and $\mathcal{A}_{\Delta,B}^m$ denotes the corresponding quantization alphabet. Obviously, \eqref{eq:SingleProjection} is the first iteration of \eqref{eq:IHT}. Similar to the single back-projection, we use versions of \eqref{eq:IHT} with $\T = \sqrt{s} B_1^n$ as well to have a fairer comparison and dub the resulting recovery schemes \emph{binary iterative soft-thresholding (BIST)} and \emph{quantized iterative soft-thresholding (QIST)}. The step-size $\lambda = \frac{1}{m} \left( 1 - \sqrt{\frac{2s}{m}} \right)$ is chosen according to the results in \cite{jacques_quantized_2013}.

%\paragraph{The advantage of using CVX.} To perform the minimization of \eqref{eq:L}, we used the CVX toolbox of MATLAB which is specialized on solving convex optimization problems. Comparing the solutions obtained by CVX to those obtained by linear programming, we observe a notable difference in approximation quality, see Figure \ref{fig:CVXvsLP}. While the solutions obtained by CVX are similar to the ones obtained by BIHT/QIHT (methods enforcing sparsity), the solutions obtained by linear programming behave similar to the ones obtained by BIST/QIST (methods minimizing the $\ell_1$-norm). This effect is seemingly caused by CVX enforcing sparsity of minimizers as both CVX and linear programming yield feasible solutions while the former one leads to solutions with smaller support. The reader should keep this side-effect of CVX in mind when studying the experiments below.

%\begin{figure}
%    \begin{subfigure}{0.5\textwidth}
%        \centering
%        \includegraphics[width=\textwidth]{Figures/CVXvsLP_1_Bit.pdf}
%        \caption{One bit quantization}
%        \end{subfigure} 
%    \begin{subfigure}{0.5\textwidth}
%        \centering
%        \includegraphics[width=\textwidth]{Figures/CVXvsLP_2_Bit.pdf}
%        \caption{Two bit quantization}
%    \end{subfigure}
%    \caption{Comparison of CVX and linear programming.} \label{fig:CVXvsLP}
%\end{figure}{}

\subsection{One-Bit Experiments}

%In the one-bit quantized measurement setting, 
For a comprehensive comparison, we ran experiments for all one-bit algorithms described above as well as single hard-thresholding step \eqref{eq:SingleProjection} with the optimal scaling for one-bit measurements without dithering (cf.\ \cite{foucart_flavors_2017}). For BIHT and BIST, we set the maximum number of iterations to 1000.
The measurement matrices have independent standard Gaussian entries. Depicted are averages over $500$ realizations. In all experiments shown here, we solved \eqref{eq:P} using the "linprog" function with the "interior-point" algorithm in Matlab. \\
%performed three different experiments to evaluate the performance of \eqref{eq:P} and its competitors. 

We first studied the setting of 10-sparse unit norm signals in $\R^{200}$ that were drawn uniformly at random.
Figure \ref{fig:OneBitAlgComp}a compares the resulting median $\ell_2$-error (in dB) for the case of noiseless measurements. We see that BIHT outperforms all other algorithms while the median error for \eqref{eq:P} is similar to BIST. This is to be expected since \eqref{eq:P} does not minimize the support size of $\x$ and thus does not enforce sparsity as strictly as BIHT. 

Next, we studied a setting of approximately sparse signals. For the signals in $\R^{200}$, we randomly chose $s=10$ components to be Gaussian with a variance of 1 and the remaining ones to be Gaussian with a variance of $10^{-3}$. All signals were then normalized to lie in the set $\T = \sqrt{s} B_1^n$. The results are shown in Figure~\ref{fig:OneBitAlgComp}b. We see that in contrast to the last experiment \eqref{eq:P} outperforms BIHT, i.e.\ using in \eqref{eq:P} the convex formulation with $\T = \sqrt{s} B_1^n$ instead of $\T = \{\z \colon |\mathrm{supp}(\z)| \le s \}$ apparently increases stability. 

\begin{figure}
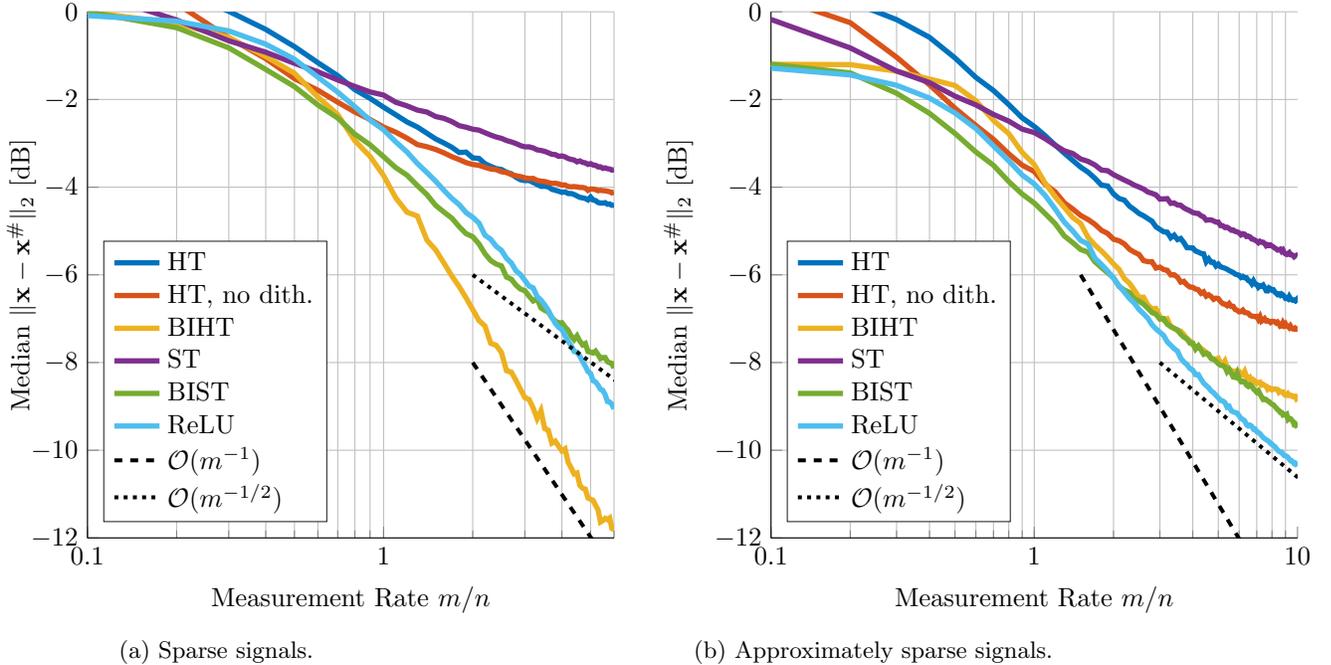

    \begin{subfigure}{0.35\textwidth}
        \centering
        \input{1bit_clean_comp.tex}
        \caption{Sparse signals.}
        \end{subfigure} \qquad\qquad\qquad\qquad
    \begin{subfigure}{0.35\textwidth}
        \centering
        \input{1bit_clean_approx_sparse_comp.tex}
        \caption{Approximately sparse signals.}
    \end{subfigure}
    \caption{Comparison of different recovery schemes.} \label{fig:OneBitAlgComp}
\end{figure}{}

Motivated by Theorem~\ref{thm:noisy_recovery}, we consider additive noise before quantizing the measurements. Here, we again wish to recover 10-sparse unit norm signals in $\R^{200}$ that were drawn uniformly at random. The noise vectors have independent Gaussian entries with zero mean and a given variance. Figure~\ref{fig:OneBitAddNoiseComp} compares the different algorithms for a fixed measurement rate of $m/n=3$ and a varying noise variance. By strict sparsity of the input, the performance of \eqref{eq:P} is again similar to BIST. In Figure~\ref{fig:OneBitAddNoiseLoglog}, we show a doubly logarithmic plot to investigate the order of decay for different noise variances (given in the legend). We clearly see that for smaller noise variances, the approximation error decays faster than for large noise. However, we cannot clearly observe the phase transition predicted by Theorem~\ref{thm:noisy_recovery} for the worst case error.

\begin{figure}
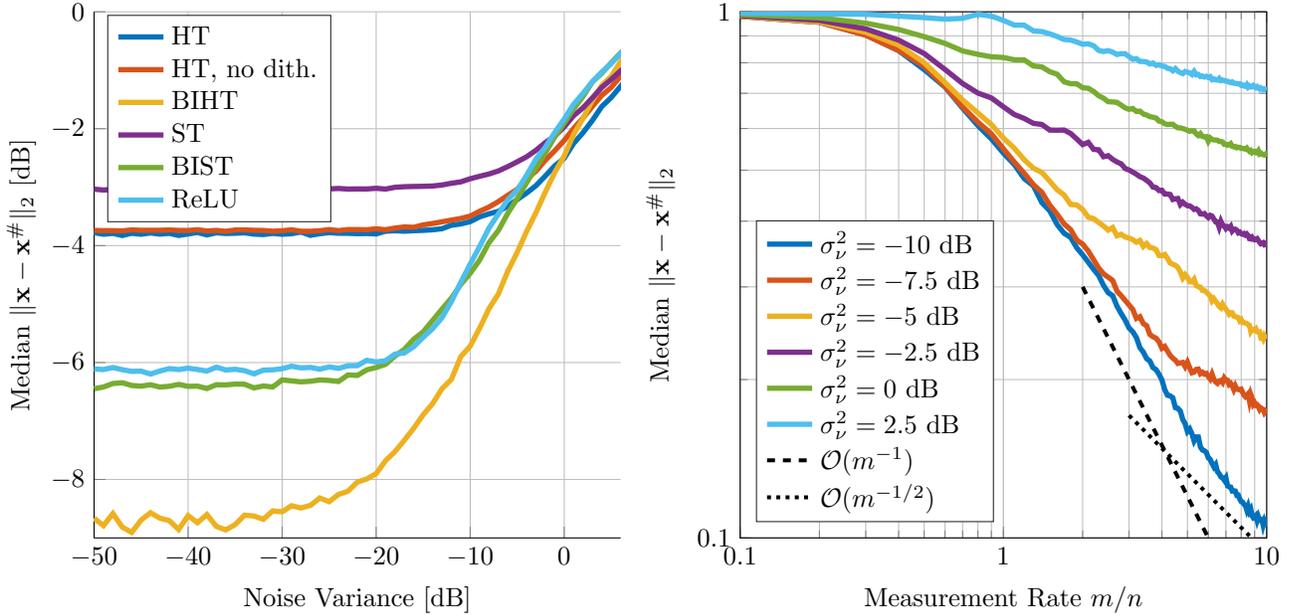

	%\hspace{-1cm}
   \begin{subfigure}{0.35\textwidth}
        \centering
        \input{1bit_additive_noise_comp.tex}
        \caption{Comparison of recovery algorithms with additive noise for $m/n = 3$.} \label{fig:OneBitAddNoiseComp}
        \end{subfigure} 
        \hspace{7em}
    \begin{subfigure}{0.35\textwidth}
        \centering
        \input{1bit_additive_noise_loglog.tex}
        \caption{Median error for \eqref{eq:P} for different noise variances.}
        \label{fig:OneBitAddNoiseLoglog}
    \end{subfigure}
    \caption{Comparison of different noise levels.}
    \label{fig:NoiseComp}
\end{figure}{}

%$5$ bit flips occur.  As can be clearly seen, the ReLU formulation exhibits its strengths when there is no bit-corruption after quantization while the single back-projection steps which have been shown to be very robust do not optimally use the uncorrupted measurement information. As soon as bit corruptions appear, the performance of the ReLU dramatically deteriorates and is astride the single back-projections. This might be explained by the $\ell_1$-norm not sufficiently enforcing sparsity in contrast to the hard sparsity constraints of BIHT, a observation verified by the data. Curiously, the accuracy improvement of ReLU behaves for small and large numbers of measurements exactly like BIHT. Only for rates between one and two (corresponding to $m = 200$ and $m = 400$) it completely fails profiting from the additional measurements and falls back behind BIHT. We lack a satisfactory explanation for this phenomenon so far.\\

\blue{Let us briefly comment on the dashed lines in Figures \ref{fig:OneBitAlgComp} and \ref{fig:NoiseComp}, highlighting decays of order $\mathcal{O}(m^{-1})$ and $\mathcal{O}(m^{-\frac{1}{2}})$ for reference. The ReLU program apparently achieves an error decay which is beyond the theoretically predicted $\mathcal{O}(m^{-\frac{1}{3}})$. Nevertheless, the reader should view the empirical outcome with a grain of salt. Recall that our theoretical results give uniform performance bounds for \emph{all} minimizers of $\mathcal{L}$ while in experiments the numerical solver decides which specific optimal point is chosen. In particular, this might add additional priors to the reconstruction process. For instance, in Figure \ref{fig:OneBitAlgComp}a the computed minimizers of the ReLU program turn out to be sparse meaning that the linear solver of Matlab implicitly restricts the signal prior from a mere $\ell_1$-ball to exactly sparse vectors and allows to reach the same decay as BIHT.}

To check the applicability of the theoretical results, we compared the recovery performance of \eqref{eq:P} for different measurement ensembles. Comparing Gaussian, Rademacher and Hadamard matrices, we could not find any notable difference in the approximation quality.
%Figure \ref{fig:OneBitMatComp} shows that recovery is not affected by the choice of the sub-gaussian distribution and that numerically even randomly subsampled orthonormal systems can be used; in this case no measurement rates beyond $1$ (corresponding to $m = 200$) can be reached.\\

%\begin{figure}
%    \centering
%    \input{Figures/Fig_ReLUCompMatrices_MSE.tex}
%    \caption{Comparison of different measurement ensembles.} \label{fig:OneBitMatComp}
%\end{figure}{}

%The third experiment shows the convenience of not requiring well-chosen parameters for recovery. While with oracle knowledge of the norm of $\x$ the single back-projection works well, Figure \ref{fig:OneBitLamComp} illustrates how strong the performance depends on a good choice of $\lambda$.\\

%\begin{figure}
%    \centering
%    \input{Figures/Fig_HTCompLambda_MSE.tex}
%    \caption{Comparison of different $\lambda$ for HT.} \label{fig:OneBitLamComp}
%\end{figure}{}

\subsection{Multi-Bit Experiments}
To demonstrate the trade-offs between measurement rates and bit rates, we performed experiments on the 10-sparse signals used in the 1-bit experiments. For our quantizer, we chose $\Delta$ to satisfy $\Gamma R\leq (2^{B-1}-1)\Delta - \frac{\Delta}{2}$ with equality for $\Gamma = 1$ in Theorem~\ref{thm:MultiBitMain}.

Figure~\ref{fig:Multibit_Measurements} shows the approximation errors at bit rates between one and five bits per measurement. We see that for larger quantizer depths, a stronger phase transition emerges that resembles the characteristics of the compressive sensing problem. For larger measurement rates, the error decreases only slowly.

In Figure~\ref{fig:Multibit_Bits}, we compare the different quantizers for a fixed bit rate $Bm/n$. We observe what is predicted by Theorem~\ref{thm:MultiBitMain}: to obtain a certain accuracy $\rho$ with a minimal bit budget, one has to choose a $B$-bit quantizer such that $\Delta \approx \rho$. For instance, if $R = 1$ and $\rho = -2 \text{dB} \approx 0.6$ one has to choose $B$ such that $\Delta \approx 2R\cdot 2^{-B} \approx 0.6$ which is fulfilled for $B = 2$. This can be verified in Figure~\ref{fig:Multibit_Bits} by checking that the red curve is minimal in bit rate for $\norm{\x - \x^\#}{2} = -2\text{dB}$.

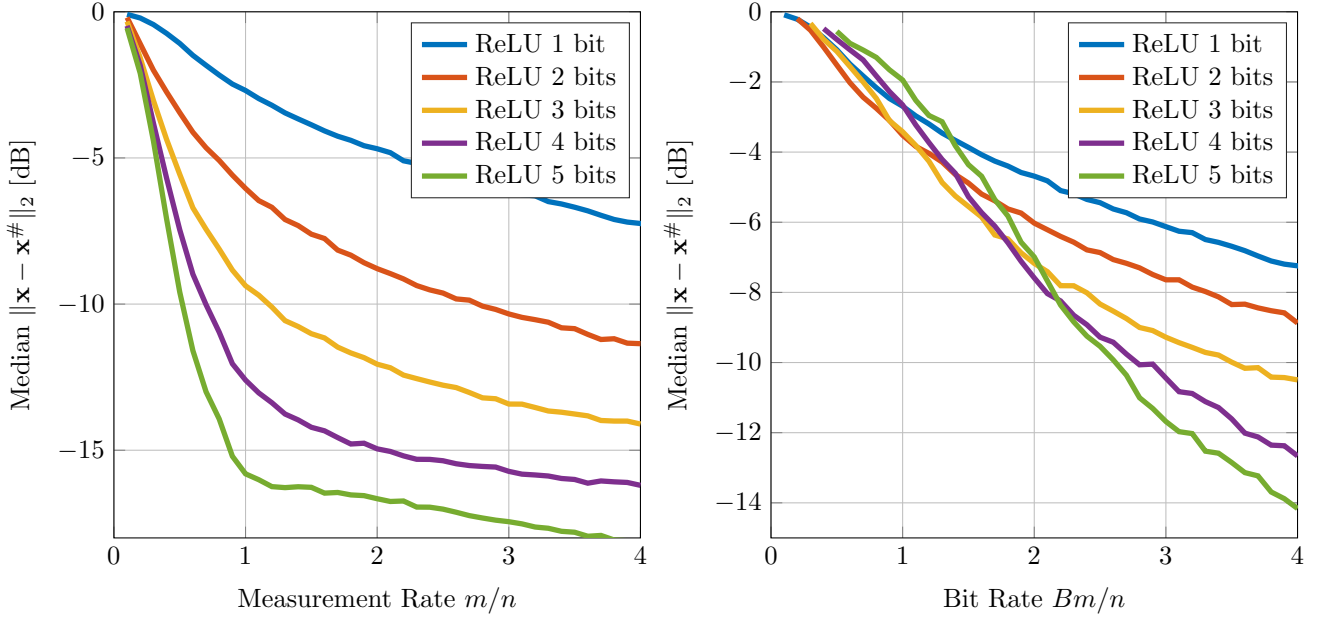
\begin{figure}
    \begin{subfigure}{0.35\textwidth}
        \centering
% Minimal pgfplots version: 1.3
%  http://www.mathworks.com/matlabcentral/fileexchange/22022-matlab2tikz
\definecolor{mycolor1}{rgb}{0.00000,0.44700,0.74100}%
\definecolor{mycolor2}{rgb}{0.85000,0.32500,0.09800}%
\definecolor{mycolor3}{rgb}{0.92900,0.69400,0.12500}%
\definecolor{mycolor4}{rgb}{0.49400,0.18400,0.55600}%
\definecolor{mycolor5}{rgb}{0.46600,0.67400,0.18800}%
\begin{tikzpicture}

\begin{axis}[%
width=7cm,
height=7cm,
at={(1.142917in,0.750139in)},
scale only axis,
xmin=0,
xmax=4,
xlabel={Measurement Rate $m/n$},
xmajorgrids,
ymin=-18,
ymax=0,
ylabel={Median $\|\x-\x^\# \|_2$ [dB]},
ymajorgrids,
title style={font=\bfseries},
legend style={legend cell align=left,align=left,draw=white!15!black}
]
\addplot [color=mycolor1,solid,line width=2.0pt]
  table[row sep=crcr]{%
0.1	-0.0868622876325475\\
0.2	-0.213030776758782\\
0.3	-0.43648721395331\\
0.4	-0.738337877845197\\
0.5	-1.08111537566997\\
0.6	-1.49037630456025\\
0.7	-1.83148332587685\\
0.8	-2.16389284886991\\
0.9	-2.47197041299883\\
1	-2.69284075729762\\
1.1	-2.96868098547425\\
1.2	-3.19796957541235\\
1.3	-3.45855966597129\\
1.4	-3.66535224208819\\
1.5	-3.86610556678547\\
1.6	-4.07470257633254\\
1.7	-4.25665365433131\\
1.8	-4.3983054483903\\
1.9	-4.58198063452251\\
2	-4.68575380405216\\
2.1	-4.82018700657654\\
2.2	-5.09765124622462\\
2.3	-5.18745240569793\\
2.4	-5.34956093987163\\
2.5	-5.4421126613625\\
2.6	-5.62154199535338\\
2.7	-5.7282418302046\\
2.8	-5.90747939385908\\
2.9	-5.99267494522218\\
3	-6.12597089379308\\
3.1	-6.25672681597876\\
3.2	-6.29678784620608\\
3.3	-6.49313357609317\\
3.4	-6.57466577095447\\
3.5	-6.68570942288786\\
3.6	-6.81063657296415\\
3.7	-6.96266585071039\\
3.8	-7.10517504548411\\
3.9	-7.19462811121707\\
4	-7.24230601516119\\
4.1	-7.38756493505878\\
4.2	-7.51495044189995\\
4.3	-7.61485719377017\\
4.4	-7.74758192555511\\
4.5	-7.76533332902537\\
4.6	-7.85335075740959\\
4.7	-7.97743795206387\\
4.8	-8.08603310305991\\
4.9	-8.20520822567154\\
5	-8.24912784074491\\
5.1	-8.41486327349796\\
5.2	-8.44304341684889\\
5.3	-8.51874366831534\\
5.4	-8.61025429125206\\
5.5	-8.63355180356\\
5.6	-8.73176210615389\\
5.7	-8.88786664234006\\
5.8	-8.90070372697159\\
5.9	-8.95101309540341\\
6	-9.05907901458857\\
};
\addlegendentry{ReLU 1 bit};

\addplot [color=mycolor2,solid,line width=2.0pt]
  table[row sep=crcr]{%
0.1	-0.19535084088376\\
0.2	-1.09722498325805\\
0.3	-1.99744810596598\\
0.4	-2.74356015288875\\
0.5	-3.44925347335476\\
0.6	-4.12110883355541\\
0.7	-4.65061480305943\\
0.8	-5.08605060692251\\
0.9	-5.58709822442718\\
1	-6.04323710236534\\
1.1	-6.45826903733004\\
1.2	-6.68927681977219\\
1.3	-7.10708509261273\\
1.4	-7.31530115460118\\
1.5	-7.60929014591733\\
1.6	-7.75743163402312\\
1.7	-8.15332861362189\\
1.8	-8.32650562182027\\
1.9	-8.58406304151407\\
2	-8.78786328148151\\
2.1	-8.95718187895858\\
2.2	-9.13897000369593\\
2.3	-9.36647729795702\\
2.4	-9.51459807364531\\
2.5	-9.6227259307919\\
2.6	-9.82330923899115\\
2.7	-9.86527872205595\\
2.8	-10.0703394039641\\
2.9	-10.1821459892125\\
3	-10.3374845425833\\
3.1	-10.4546960140301\\
3.2	-10.5333521430396\\
3.3	-10.6206998861883\\
3.4	-10.8116214239686\\
3.5	-10.8406519979314\\
3.6	-11.0364655094137\\
3.7	-11.2113921078256\\
3.8	-11.1858148171342\\
3.9	-11.3368502782162\\
4	-11.3547071339195\\
};
\addlegendentry{ReLU 2 bits};

\addplot [color=mycolor3,solid,line width=2.0pt]
  table[row sep=crcr]{%
0.1	-0.331726313436056\\
0.2	-1.57761449371437\\
0.3	-3.03408088474451\\
0.4	-4.36477328265204\\
0.5	-5.55962337741871\\
0.6	-6.70751597677168\\
0.7	-7.43985979106959\\
0.8	-8.11959176657431\\
0.9	-8.837611160525\\
1	-9.37418962024841\\
1.1	-9.69370993365973\\
1.2	-10.1018163756506\\
1.3	-10.5661254120764\\
1.4	-10.7681709197604\\
1.5	-11.0136806659553\\
1.6	-11.1585343654295\\
1.7	-11.4722366670213\\
1.8	-11.6721473750243\\
1.9	-11.832694590081\\
2	-12.056501459847\\
2.1	-12.1780844802156\\
2.2	-12.4290929592143\\
2.3	-12.5498728139084\\
2.4	-12.6659315908904\\
2.5	-12.7745110440294\\
2.6	-12.8563632462557\\
2.7	-13.024566914277\\
2.8	-13.2075100037266\\
2.9	-13.2386142916229\\
3	-13.4199890882384\\
3.1	-13.4230740850267\\
3.2	-13.5376435539605\\
3.3	-13.6574447967251\\
3.4	-13.6992682873945\\
3.5	-13.7606547582277\\
3.6	-13.8238597896381\\
3.7	-13.9824091885361\\
3.8	-14.0024471671894\\
3.9	-14.0030629622093\\
4	-14.104421413397\\
};
\addlegendentry{ReLU 3 bits};

\addplot [color=mycolor4,solid,line width=2.0pt]
  table[row sep=crcr]{%
0.1	-0.479378424431281\\
0.2	-1.83903566105029\\
0.3	-3.8069119835482\\
0.4	-5.6878166952865\\
0.5	-7.45653399283292\\
0.6	-8.97865885391003\\
0.7	-10.021170143191\\
0.8	-10.9624680049735\\
0.9	-12.0410488729671\\
1	-12.6010232177316\\
1.1	-13.0376477750502\\
1.2	-13.3743567886307\\
1.3	-13.761446454093\\
1.4	-13.9629199093877\\
1.5	-14.2152910908951\\
1.6	-14.3360366223415\\
1.7	-14.5669318014224\\
1.8	-14.7869527816906\\
1.9	-14.7602043677896\\
2	-14.9527597673569\\
2.1	-15.0419377376934\\
2.2	-15.1932968414949\\
2.3	-15.3099176925987\\
2.4	-15.31371613083\\
2.5	-15.3600034712348\\
2.6	-15.4622560798354\\
2.7	-15.5242675690984\\
2.8	-15.5542953537361\\
2.9	-15.5758041688587\\
3	-15.7224982238152\\
3.1	-15.8182567463983\\
3.2	-15.8461978132732\\
3.3	-15.8817707375128\\
3.4	-15.96604674137\\
3.5	-16.0003403237093\\
3.6	-16.1271874226428\\
3.7	-16.0510157661654\\
3.8	-16.0811528640016\\
3.9	-16.1021874818133\\
4	-16.2009901409631\\
};
\addlegendentry{ReLU 4 bits};

\addplot [color=mycolor5,solid,line width=2.0pt]
  table[row sep=crcr]{%
0.1	-0.557262172184168\\
0.2	-2.10647615190832\\
0.3	-4.41312436208134\\
0.4	-7.08164470261485\\
0.5	-9.58835413715016\\
0.6	-11.5935168173484\\
0.7	-13.0004539544976\\
0.8	-13.9318537336976\\
0.9	-15.2064219694956\\
1	-15.8081663004294\\
1.1	-16.0078299451493\\
1.2	-16.2450743781735\\
1.3	-16.2792962774263\\
1.4	-16.2490902742802\\
1.5	-16.2690490014593\\
1.6	-16.4674710604561\\
1.7	-16.4450016125117\\
1.8	-16.5286219800828\\
1.9	-16.5536090473568\\
2	-16.6588450868727\\
2.1	-16.7532236818454\\
2.2	-16.7365726496958\\
2.3	-16.9442206084443\\
2.4	-16.9487251627059\\
2.5	-17.0153821332712\\
2.6	-17.1212975986077\\
2.7	-17.2369750672111\\
2.8	-17.3185000618126\\
2.9	-17.3922136080748\\
3	-17.4449531836622\\
3.1	-17.5196097816838\\
3.2	-17.6273417351051\\
3.3	-17.6696157690286\\
3.4	-17.7736837520471\\
3.5	-17.8056205339321\\
3.6	-17.9431449610275\\
3.7	-17.9147581381764\\
3.8	-18.0709561644182\\
3.9	-18.0939034488035\\
4	-18.1371606528301\\
};
\addlegendentry{ReLU 5 bits};
\end{axis}
\end{tikzpicture}%
        \caption{Error vs. measurement rate.}
        \label{fig:Multibit_Measurements}
        \end{subfigure} \qquad\qquad\qquad\qquad
    \begin{subfigure}{0.35\textwidth}
        \centering
% Minimal pgfplots version: 1.3
%  http://www.mathworks.com/matlabcentral/fileexchange/22022-matlab2tikz
\definecolor{mycolor1}{rgb}{0.00000,0.44700,0.74100}%
\definecolor{mycolor2}{rgb}{0.85000,0.32500,0.09800}%
\definecolor{mycolor3}{rgb}{0.92900,0.69400,0.12500}%
\definecolor{mycolor4}{rgb}{0.49400,0.18400,0.55600}%
\definecolor{mycolor5}{rgb}{0.46600,0.67400,0.18800}%
\begin{tikzpicture}

\begin{axis}[%
width=7cm,
height=7cm,
at={(1.011111in,0.641667in)},
scale only axis,
xmin=0,
xmax=4,
xlabel={Bit Rate $Bm/n$},
xmajorgrids,
ymin=-15,
ymax=0,
ylabel={Median $\|\x-\x^\#\|_2$ [dB]},
ymajorgrids,
legend style={legend cell align=left,align=left,draw=white!15!black}
]
\addplot [color=mycolor1,solid,line width=2.0pt]
  table[row sep=crcr]{%
0.1	-0.0868622876325475\\
0.2	-0.213030776758782\\
0.3	-0.43648721395331\\
0.4	-0.738337877845197\\
0.5	-1.08111537566997\\
0.6	-1.49037630456025\\
0.7	-1.83148332587685\\
0.8	-2.16389284886991\\
0.9	-2.47197041299883\\
1	-2.69284075729762\\
1.1	-2.96868098547425\\
1.2	-3.19796957541235\\
1.3	-3.45855966597129\\
1.4	-3.66535224208819\\
1.5	-3.86610556678547\\
1.6	-4.07470257633254\\
1.7	-4.25665365433131\\
1.8	-4.3983054483903\\
1.9	-4.58198063452251\\
2	-4.68575380405216\\
2.1	-4.82018700657654\\
2.2	-5.09765124622462\\
2.3	-5.18745240569793\\
2.4	-5.34956093987163\\
2.5	-5.4421126613625\\
2.6	-5.62154199535338\\
2.7	-5.7282418302046\\
2.8	-5.90747939385908\\
2.9	-5.99267494522218\\
3	-6.12597089379308\\
3.1	-6.25672681597876\\
3.2	-6.29678784620608\\
3.3	-6.49313357609317\\
3.4	-6.57466577095447\\
3.5	-6.68570942288786\\
3.6	-6.81063657296415\\
3.7	-6.96266585071039\\
3.8	-7.10517504548411\\
3.9	-7.19462811121707\\
4	-7.24230601516119\\
4.1	-7.38756493505878\\
4.2	-7.51495044189995\\
4.3	-7.61485719377017\\
4.4	-7.74758192555511\\
4.5	-7.76533332902537\\
4.6	-7.85335075740959\\
4.7	-7.97743795206387\\
4.8	-8.08603310305991\\
4.9	-8.20520822567154\\
5	-8.24912784074491\\
5.1	-8.41486327349796\\
5.2	-8.44304341684889\\
5.3	-8.51874366831534\\
5.4	-8.61025429125206\\
5.5	-8.63355180356\\
5.6	-8.73176210615389\\
5.7	-8.88786664234006\\
5.8	-8.90070372697159\\
5.9	-8.95101309540341\\
6	-9.05907901458857\\
};
\addlegendentry{ReLU 1 bit};

\addplot [color=mycolor2,solid,line width=2.0pt]
  table[row sep=crcr]{%
0.2	-0.19535084088376\\
0.3	-0.525269127253631\\
0.4	-1.00761102039019\\
0.5	-1.51831622252314\\
0.6	-2.02363297181542\\
0.7	-2.43501478493551\\
0.8	-2.75143797962175\\
0.9	-3.09792208535654\\
1	-3.51843401038173\\
1.1	-3.8390207766418\\
1.2	-4.03693194328257\\
1.3	-4.28392000890834\\
1.4	-4.63820769486561\\
1.5	-4.88325624538992\\
1.6	-5.19775953443985\\
1.7	-5.38900726480181\\
1.8	-5.62075190516029\\
1.9	-5.73874515839831\\
2	-6.02511316733119\\
2.1	-6.21382284545017\\
2.2	-6.40358814129874\\
2.3	-6.56909021350589\\
2.4	-6.78796416629083\\
2.5	-6.86499727818024\\
2.6	-7.06415683327829\\
2.7	-7.17121849314339\\
2.8	-7.30040479720962\\
2.9	-7.4905156814975\\
3	-7.64138607323757\\
3.1	-7.64157121803472\\
3.2	-7.84970913258104\\
3.3	-7.97586949799439\\
3.4	-8.12897808708159\\
3.5	-8.34697755385806\\
3.6	-8.33566452442664\\
3.7	-8.44191883240305\\
3.8	-8.5228973419662\\
3.9	-8.58374458170933\\
4	-8.87340070001868\\
};
\addlegendentry{ReLU 2 bits};

\addplot [color=mycolor3,solid,line width=2.0pt]
  table[row sep=crcr]{%
0.3	-0.331726313436056\\
0.4	-0.783380325380924\\
0.5	-1.12723211557044\\
0.6	-1.56050576995362\\
0.7	-1.99941974975024\\
0.8	-2.46845857622198\\
0.9	-3.10382822419457\\
1	-3.42047547059847\\
1.1	-3.81122519684371\\
1.2	-4.26057907459977\\
1.3	-4.86711365214588\\
1.4	-5.25066684088291\\
1.5	-5.55061822648903\\
1.6	-5.85163310551397\\
1.7	-6.36343010214119\\
1.8	-6.47567948134448\\
1.9	-6.87109114378118\\
2	-7.17354203859935\\
2.1	-7.40835241149304\\
2.2	-7.80889470979832\\
2.3	-7.81110517994249\\
2.4	-8.01176735541934\\
2.5	-8.33262710426594\\
2.6	-8.53361460401285\\
2.7	-8.7452071722102\\
2.8	-9.00047641349719\\
2.9	-9.08842254434062\\
3	-9.27824819553703\\
3.1	-9.43279132308664\\
3.2	-9.56807203601112\\
3.3	-9.71187754242742\\
3.4	-9.78694851229895\\
3.5	-9.98888160461631\\
3.6	-10.1637872870874\\
3.7	-10.1457606820039\\
3.8	-10.4154594738157\\
3.9	-10.4285755186196\\
4	-10.4944309351105\\
};
\addlegendentry{ReLU 3 bits};

\addplot [color=mycolor4,solid,line width=2.0pt]
  table[row sep=crcr]{%
0.4	-0.479378424431281\\
0.5	-0.78834704425448\\
0.6	-1.08534914395436\\
0.7	-1.37359544998297\\
0.8	-1.82915528481427\\
0.9	-2.27551899626895\\
1	-2.65570460289479\\
1.1	-3.22686541942665\\
1.2	-3.72958776138545\\
1.3	-4.20987043942108\\
1.4	-4.60440387517962\\
1.5	-5.27342263015489\\
1.6	-5.72659656279622\\
1.7	-6.11007518552359\\
1.8	-6.58445349303928\\
1.9	-7.1144551283507\\
2	-7.58683091748864\\
2.1	-8.03517712024541\\
2.2	-8.24395004634598\\
2.3	-8.65242421816705\\
2.4	-8.92283225778831\\
2.5	-9.27733163358099\\
2.6	-9.42071557080525\\
2.7	-9.75854097934333\\
2.8	-10.0643966658821\\
2.9	-10.0453155133439\\
3	-10.442945494226\\
3.1	-10.8311247690416\\
3.2	-10.8828212873372\\
3.3	-11.1136440965315\\
3.4	-11.2837901540113\\
3.5	-11.6071520951813\\
3.6	-12.0125230038108\\
3.7	-12.1218276503952\\
3.8	-12.3548054338356\\
3.9	-12.373127911284\\
4	-12.662534041873\\
};
\addlegendentry{ReLU 4 bits};

\addplot [color=mycolor5,solid,line width=2.0pt]
  table[row sep=crcr]{%
0.5	-0.557262172184168\\
0.6	-0.898839038180063\\
0.7	-1.0908426476125\\
0.8	-1.30029350960269\\
0.9	-1.65457538661407\\
1	-1.94568517329487\\
1.1	-2.52447061995973\\
1.2	-2.9563177646541\\
1.3	-3.13463331822628\\
1.4	-3.81646997544289\\
1.5	-4.35922470114133\\
1.6	-4.68863538156572\\
1.7	-5.35258076853805\\
1.8	-5.82388798784677\\
1.9	-6.56976613222525\\
2	-6.98922186379593\\
2.1	-7.67498383633915\\
2.2	-8.36062145669928\\
2.3	-8.8410155860933\\
2.4	-9.24751344539595\\
2.5	-9.53765353895885\\
2.6	-9.9236532757996\\
2.7	-10.3601197330121\\
2.8	-11.0040197317908\\
2.9	-11.3096244096484\\
3	-11.6809205834931\\
3.1	-11.9705525216822\\
3.2	-12.025384682044\\
3.3	-12.5268975452233\\
3.4	-12.584916184974\\
3.5	-12.8502292736038\\
3.6	-13.138231273289\\
3.7	-13.231398790056\\
3.8	-13.6904730192818\\
3.9	-13.8822950259007\\
4	-14.1617279262075\\
};
\addlegendentry{ReLU 5 bits};

\end{axis}
\end{tikzpicture}%
        \caption{Error vs. bit rate}
        \label{fig:Multibit_Bits}
    \end{subfigure}
    \caption{Comparison of different bit levels.}
\end{figure}{}

%%%%%%%%%%%%%%%%%%%%%%%%%%%%%%%%%%%%%%%%
%%% Acknowledgements
%%%%%%%%%%%%%%%%%%%%%%%%%%%%%%%%%%%%%%%%	

\section*{Acknowledgments}

H.C.J., J.M.\ and L.P. acknowledge funding by the Deutsche Forschungsgemeinschaft (DFG, German Research Foundation) under SPP 1798. A.S.\ acknowledges funding by the Deutsche Forschungsgemeinschaft (DFG, German Research Foundation) under Germany's Excellence Strategy – MATH+ : The Berlin Mathematics Research Center, EXC-2046/1 – project ID: 390685689.

%%%%%%%%%%%%%%%%%%%%%%%%%%%%%%%%%%%
%%% Appendix
%%%%%%%%%%%%%%%%%%%%%%%%%%%%%%%%%%%

\appendix

%%%%%%%%%%%%%%%%%%%%%%%%%%%%%%%%%%%%%%%
%%% Stability of separating hyperplanes
%%%%%%%%%%%%%%%%%%%%%%%%%%%%%%%%%%%%%%%

\section{Additional Proofs} \label{app:noise_stability}
\begin{proof}[of Lemma \ref{lem:noise_stability}]
By the triangle inequality we have
\begin{equation*}
    \norm{\x - \z}{2} \geq \norm{\x'-\z'}{2} - 2\ep \geq \norm{\x'-\z'}{2} - \frac{\rho}{2}\geq \frac{\norm{\x'-\z'}{2}}{2}\; ,
\end{equation*}
and similarly $\norm{\x'-\z'}{2}\geq \frac{2}{3} \norm{\x - \z}{2}$.
If $i \in I(\x',\z',\kappa)$ and $|\inner{\a_i}{\x-\x'}| \leq \kappa \rho /4$, then
\begin{align*}
|\inner{\a_i}{\x}  + \nu_i+ \tau_i| &\geq \kappa \norm{\x'-\z'}{2} - \kappa \rho/4 \geq \kappa \norm{\x'-\z'}{2} - \frac{\kappa}{4} \norm{\x'-\z'}{2} =\frac{3\kappa}{4} \norm{\x'-\z'}{2}\geq 
\frac{\kappa}{2} \norm{\x-\z}{2}.
\end{align*}
Moreover, since $|\inner{\a_i}{\x'}  + \nu_i + \tau_i| \geq \kappa \rho$ and $|\inner{\a_i}{\x-\x'}| \leq \kappa \rho /4$, it follows that 
$\sign(\inner{\a_i}{\x'}  + \nu_i + \tau_i)=\sign(\inner{\a_i}{\x}  + \nu_i + \tau_i)$.
Analogously, if $i \in I(\x',\z', \kappa)$ and $|\inner{\a_i}{\z-\z'}| \leq \kappa \rho /4$, then $|\inner{\a_i}{\z} + \tau_i|\geq \frac{\kappa}{2} \norm{\x-\z}{2}$,
and $\sign(\inner{\a_i}{\z'} + \tau_i)=\sign(\inner{\a_i}{\z} + \tau_i)$. Hence, if $i \in I(\x',\z', \kappa), |\inner{\a_i}{\x-\x'}| \leq \kappa \rho /4$ and
$|\inner{\a_i}{\z-\z'}| \leq \kappa \rho /4$, then $i \in I(\x,\z,\kappa/2)$.
Hence, 
\begin{align*}
    |I(\x',\z',\kappa)|&\leq
    | I(\x',\z',\kappa) \cap \{i\in [m]\; :\; |\inner{\a_i}{\z-\z'}|\leq \kappa \rho /4
    \text{ and } |\inner{\a_i}{\x-\x'}|\leq \kappa \rho /4\}|\\
    &\quad + |\{i\in [m]\; :\; |\inner{\a_i}{\z-\z'}|> \kappa \rho /4
    \text{ or } |\inner{\a_i}{\x-\x'}|> \kappa \rho /4\}|.
\end{align*}

This observation shows that
\begin{align*}
    |I(\x,\z,\kappa/2)| &\geq |I(\x',\z',\kappa)| - 2\sup_{\y\in (\T-\T)\cap \ep B^n_2}|\{i\in [m]\; : \; |\inner{\a_i}{\y}|>\kappa \rho/4\}| \; \\
    &\geq |I(\x',\z',\kappa)| - 2k.
\end{align*}

\end{proof}
%%%%%%%%%%%%%%%%%%%%%%%%%%%%%%%%%%%%%%%
%%% Bound for spoiled hyperplanes
%%%%%%%%%%%%%%%%%%%%%%%%%%%%%%%%%%%%%%%

\begin{proof}[of Lemma \ref{lem:bound_for_spoiled_hyperplanes}]

For a vector $\w \in \R_+^m$ let $\w^*$ denote its non-increasing rearrangement. We observe that for 
any $1 \leq k \leq m$ we have
\begin{equation}
    w_k^* \leq \max_{ |I| \leq k} \Big( \frac{1}{k} \sum_{i \in I} w_i^2 \Big)^{\frac{1}{2}}.
\end{equation}
In particular, if $\max_{|I|\leq k}( \frac{1}{k} \sum_{i \in I} w_i^2 )^{\frac{1}{2}} \leq \rho$ then $w_k^*\leq \rho$ and therefore
\begin{displaymath}
 |\{i\in [m]\; : \; |w_i|>\rho\}| <  k \; .
\end{displaymath}
Therefore, 
\begin{displaymath}
\Big\{ \S_{k} := \sup_{\z \in \mathcal{K}}\max_{ |I| \leq k} \Big( \frac{1}{k} \sum_{i \in I} \inner{\a_i}{\z}^2 \Big)^{\frac{1}{2}} \leq \rho  \Big\} \subset 
\Big\{\sup_{\z\in \mathcal{K}}|\{i\in [m]\; : \; |\inner{\a_i}{\z}|>\rho\}|
 < k \Big\}.
\end{displaymath}
By Theorem~\ref{subgaussian_tail} there exist constants $c_1, c_2>0$ that only depend on $L$ such that
for every $u\geq 1$ 
with probability at least $1 - 2 \exp(-c_1 u^2 k \log(em/k))$,
\begin{equation}
\sqrt{k} \S_{k} \leq c_2 \Big( w_*(\mathcal{K}) + u\, d_{\mathcal{K}} \sqrt{k \log(em/k)} \Big) \; .
\end{equation}
Hence, we only need to ensure that 
\begin{equation}
    \frac{c_2}{\sqrt{k}} \Big( w_*(\mathcal{K}) + u\, d_{\mathcal{K}} \sqrt{k \log(em/k)} \Big)\leq \rho.
\end{equation}
  
\end{proof}

%%%%%%%%%%%%%%%%%%%%%%%%%%%%%%%%%%%%%%%
%%% Noisy Uniform
%%%%%%%%%%%%%%%%%%%%%%%%%%%%%%%%%%%%%%%
\begin{proof}[of Theorem \ref{thm:noisy_uniform}] 
Let $\ep > 0$ be a real number with $\eps\leq \frac{\rho}{8}$ and let $\N_\ep$ denote a minimal $\ep$-net in $\T\subset R B^n_2$ with respect
to the Euclidean norm. By Theorem~\ref{thm:noise_well-separating} there exist absolute
constants $\kappa, c_0, c_1>0$ such that 
for each pair $(\x',\z') \in \N_\ep$ with probability at least $1 - 4 \exp(-c_0 m \norm{\x' - \z'}{2}/\lambda )$ we have
\begin{equation*}
    |I(\x',\z', \nnu, \kappa)| \geq c_1 m \frac{\norm{\x'-\z'}{2}}{\lambda} \; .
\end{equation*}
By a union bound and the assumption $\log|\N_\ep| \lesssim \frac{\rho}{\lambda}m$ it follows that with probability at least  $1 - 4 \exp(- c_0 m \rho/4\lambda )$ the following event occurs,
\begin{equation}
\label{eqn:controll_on_ep_net}
    \forall (\x',\z') \in \N_\ep^2 \text{ with } \norm{\x'-\z'}{2}\geq \rho/2 \; : \; |I(\x',\z',\nnu, \kappa)| \geq c_1 m \frac{\norm{\x'-\z'}{2}}{\lambda}  \; .
\end{equation}
With this in mind consider arbitrary $\x,\z \in \T$ such that $\norm{\x-\z}{2} \geq \rho$. There are $\x',\z'\in \N_\ep$ such that $\norm{\x-\x'}{2}\leq \ep$ and $\norm{\z-\z'}{2}\leq \ep$.
Since $2\eps\leq \frac{\rho}{2}\leq \frac{\norm{\x - \z}{2}}{2}$ it follows
\begin{equation}
\label{eqn:size_of_distance}
    \norm{\x' - \z'}{2} \geq \norm{\x-\z}{2} - 2\eps \geq \frac{\norm{\x - \z}{2}}{2}  \geq \frac{\rho}{2},
\end{equation}
which implies that on the event \eqref{eqn:controll_on_ep_net} we have
\begin{equation*}
    |I(\x',\z', \nnu, \kappa)| \geq c_1 m \frac{\norm{\x'-\z'}{2}}{\lambda} \; .
\end{equation*}
Applying Lemma \ref{lem:noise_stability} to this situation with $\rho/2$ instead of $\rho$, shows that 
\begin{align}
    |I(\x,\z,\nnu,\kappa/2)| 
    &\geq |I(\x',\z',\nnu,\kappa)| 
    -
     2\sup_{\y\in (\T-\T)\cap \ep B^n_2}|\{i\in [m]\; : \; |\inner{\a_i}{\y}|>\kappa \rho/8\}|\\
    &\geq c_1 m \frac{\norm{\x'-\z'}{2}}{\lambda} -
     2\sup_{\y\in (\T-\T)\cap \ep B^n_2}|\{i\in [m]\; : \; |\inner{\a_i}{\y}|>\kappa \rho/8\}| \; \\
     &\geq c_1 m \frac{\norm{\x-\z}{2}}{2\lambda} -
     2\sup_{\y\in (\T-\T)\cap \ep B^n_2}|\{i\in [m]\; : \; |\inner{\a_i}{\y}|>\kappa \rho/8\}| \;.
\end{align}
The theorem follows once we can ensure that
\begin{equation}
\label{eqn:spoiled_hyperplane}
    \sup_{\y\in (\T-\T)\cap \ep B^n_2}|\{i\in [m]\; : \; |\inner{\a_i}{\y}|>\kappa \rho/8\}| \leq c_1 m \frac{\rho}{8\lambda}
\end{equation}
with high probability. In order to accomplish this we invoke Lemma \ref{lem:bound_for_spoiled_hyperplanes}. This lemma shows that there exist constants 
$c_2, c_3>0$ such that for every $1\leq k \leq m$ and $u\geq 1$
we have
\begin{displaymath}
 \sup_{\y\in (\T-\T)\cap \ep B^n_2}|\{i\in [m]\; : \; |\inner{\a_i}{\y}|>\kappa \rho/8\}| < k
\end{displaymath}
with probability at least $1-2\exp(-c_2 u^2 k \log(em/k))$ provided that
\begin{equation}
    \frac{c_3}{\sqrt{k}}\Big(w_*((\T-\T)\cap \ep B^n_2) + u \ep \sqrt{k \log(em/k)} \Big) \leq  \frac{\kappa \rho}{8}\; .
\end{equation}
Choosing $k = c_1 m \rho/8\lambda$ and $u=1$ we obtain  that \eqref{eqn:spoiled_hyperplane} can achieved with probability at least
$1-2\exp(-c_4 m \rho\log(e\lambda/\rho)/\lambda)$ provided that
\begin{equation}
    m \gg \lambda \rho^{-3} w_*((\T-\T)\cap \ep B^n_2)^2 \quad \text{ and }\quad \ep \lesssim \frac{\rho }{\sqrt{\log(C\lambda/\rho)}} \; ,
\end{equation}
where $c_4>0$ and $C\geq e$ denote absolute constants.
With probability at least $1-2\exp(-c_4 m \rho\log(e\lambda/\rho)/\lambda)- 4 \exp(- c_0 m \rho/4\lambda )$ the events 
\eqref{eqn:controll_on_ep_net} and
\eqref{eqn:spoiled_hyperplane} both occur and 
by the calculation above for all $\x,\z \in \T$ with $\norm{\x-\z}{2} \geq \rho$ we have
\begin{align}
    |I(\x,\z,\nnu,\kappa/2)| &\geq c_1 m \frac{\norm{\x-\z}{2}}{2\lambda} - c_1 m \frac{\rho}{4\lambda}\\
    &\geq c_1 m \frac{\norm{\x-\z}{2}}{4\lambda},
\end{align}
which shows the result.
\end{proof}

\begin{proof}[of Lemma \ref{lem:expectation_quantizer}]
  In order to simplify notation we define $q:=q_{\mathcal{A}_{\Delta,B}}: \R 
\to \mathcal{A}_{\Delta,B}$.
  The second statement immediately follows from the first. Indeed, suppose
  $$x,y \in [-((2^{B-1}-1)\Delta-\frac{\Delta}{2}),(2^{B-1}-1)\Delta-\frac{\Delta}{2}]$$ and w.l.o.g. let $x<y$. Using \eqref{eq:lem:expectation_quantizer} we obtain
  \begin{align*}
      \E [| q(x+\tau) - q(y+\tau) |]=\E [ - q(x+\tau)+  q(y+\tau) ]
      =-x+y=|x-y|,
  \end{align*}
  where we have additionally used that the quantizer $q$ is a monotonically increasing function.
  Next, we show \eqref{eq:lem:expectation_quantizer}. 
  For $j\in \mathbb{Z}$ set 
  \begin{equation}
      q_j=\frac{\Delta}{2}+j\Delta.
  \end{equation}
  Then 
  $$\mathcal{A}_{\Delta,B}=
  \{q_j\; : \; j\in \{-2^{B-1},\ldots, -1,0,1,\ldots, 2^{B-1}-1 \}\}.$$
  If 
  $\x\in [-((2^{B-1}-1)\Delta-\frac{\Delta}{2}),(2^{B-1}-1)\Delta-\frac{\Delta}{2}]=
  [q_{-2^{B-1}+1}, q_{2^{B-1}-2}]$, then there exists $j\in 
  \{-2^{B-1}+1,\ldots, -1,0,1,\ldots, 2^{B-1}-3 \}$ such that $q_j\leq x\leq q_{j+1}$. Let us first assume that $x\in [q_j,q_j+\frac{\Delta}{2}]$. Then $-\frac{\Delta}{2}\leq q_j-x\leq 0$ and
  \begin{align*}
     \E [q(x+\tau)]
     &= q_{j-1}\Pr(q_{j-1}-\frac{\Delta}{2}\leq x+\tau\leq q_{j-1}+\frac{\Delta}{2})\\
     &\quad+q_j\Pr(q_{j}-\frac{\Delta}{2}\leq x+\tau\leq q_{j}
     +\frac{\Delta}{2})
     +q_{j+1}\Pr(q_{j+1}-\frac{\Delta}{2}\leq x+\tau\leq q_{j+1}+\frac{\Delta}{2})\\
    &= q_{j-1}\Pr(q_{j-1}-\frac{\Delta}{2}-x\leq \tau\leq q_{j-1}+\frac{\Delta}{2}-x)\\
     &\quad+q_j\Pr(q_{j}-\frac{\Delta}{2}-x\leq \tau\leq q_{j}
     +\frac{\Delta}{2}-x)
     +q_{j+1}\Pr(q_{j+1}-\frac{\Delta}{2}-x\leq \tau\leq q_{j+1}+\frac{\Delta}{2}-x)\\
     &= q_{j-1}\Pr(q_{j}-\frac{3\Delta}{2}-x\leq \tau\leq q_{j}-\frac{\Delta}{2}-x)\\
     &\quad+q_j\Pr(q_{j}-\frac{\Delta}{2}-x\leq \tau\leq q_{j}
     +\frac{\Delta}{2}-x)
     +q_{j+1}\Pr(q_{j}+\frac{\Delta}{2}-x\leq \tau\leq q_{j}+\frac{3\Delta}{2}-x)\\
     &= q_{j-1}\Pr(-\Delta\leq \tau\leq q_{j}-x-\frac{\Delta}{2})\\
     &\quad+q_j\Pr(q_{j}-x-\frac{\Delta}{2}\leq \tau\leq q_{j}-x
     +\frac{\Delta}{2})
     +q_{j+1}\Pr(q_{j}-x
     +\frac{\Delta}{2}\leq \tau\leq \Delta). 
  \end{align*}
Further, 
\begin{align*}
  \Pr(-\Delta\leq \tau\leq q_{j}-x-\frac{\Delta}{2})  &=\frac{q_j-x+\frac{\Delta}{2}}{2\Delta},\\
  \Pr(q_{j}-x-\frac{\Delta}{2}\leq \tau\leq q_{j}-x
     +\frac{\Delta}{2})&=\frac{1}{2},\\
   \Pr(q_{j}-x
     +\frac{\Delta}{2}\leq \tau\leq \Delta)&=\frac{\frac{\Delta}{2}-q_j+x}{2\Delta}.  
\end{align*}
Therefore,
\begin{align*}
    \E [q(x+\tau)]&=q_{j-1}\big(\frac{q_j-x+\frac{\Delta}{2}}{2\Delta})+\frac{q_j}{2}
    +q_{j+1}\big(\frac{\frac{\Delta}{2}-q_j+x}{2\Delta}\big)\\
    &=(q_{j}-\Delta)\big(\frac{q_j-x+\frac{\Delta}{2}}{2\Delta})+\frac{q_j}{2}
    +(q_{j}+\Delta)\big(\frac{\frac{\Delta}{2}-q_j+x}{2\Delta}\big)\\
    &=q_j\big(\frac{q_j-x+\frac{\Delta}{2}}{2\Delta}+\frac{\Delta}{2\Delta}+\frac{\frac{\Delta}{2}-q_j+x}{2\Delta}\big)-\frac{q_j-x+\frac{\Delta}{2}}{2} + \frac{\frac{\Delta}{2}-q_j+x}{2}\\
    &=q_j-\frac{q_j}{2}+\frac{x}{2}-\frac{q_j}{2}+\frac{x}{2}\\
    &=x.
\end{align*}
Analogously, one can show that
$\E [q(x+\tau)]=x$ if $x\in [q_j+\frac{\Delta}{2}, q_{j+1}]$.
\end{proof}

%\printbibliography
\bibliographystyle{IEEEtranS}
\bibliography{mbit}

% Generated by IEEEtranS.bst, version: 1.14 (2015/08/26)
\begin{thebibliography}{10}
\providecommand{\url}[1]{#1}
\csname url@samestyle\endcsname
\providecommand{\newblock}{\relax}
\providecommand{\bibinfo}[2]{#2}
\providecommand{\BIBentrySTDinterwordspacing}{\spaceskip=0pt\relax}
\providecommand{\BIBentryALTinterwordstretchfactor}{4}
\providecommand{\BIBentryALTinterwordspacing}{\spaceskip=\fontdimen2\font plus
\BIBentryALTinterwordstretchfactor\fontdimen3\font minus
  \fontdimen4\font\relax}
\providecommand{\BIBforeignlanguage}[2]{{%
\expandafter\ifx\csname l@#1\endcsname\relax
\typeout{** WARNING: IEEEtranS.bst: No hyphenation pattern has been}%
\typeout{** loaded for the language `#1'. Using the pattern for}%
\typeout{** the default language instead.}%
\else
\language=\csname l@#1\endcsname
\fi
#2}}
\providecommand{\BIBdecl}{\relax}
\BIBdecl

\bibitem{herman2008high}
M.~A.~Herman and T.~Strohmer, ``High-resolution radar via compressed sensing,''
  \emph{IEEE Transactions on Signal Processing}, vol.~57, pp. 2275 -- 2284, 07
  2009.

\bibitem{ai_one-bit_2014}
A.~Ai, A.~Lapanowski, Y.~Plan, and R.~Vershynin, ``One-bit compressed sensing
  with non-{Gaussian} measurements,'' \emph{Linear Algebra and its
  Applications}, vol. 441, pp. 222--239, 2014.

\bibitem{amelunxen14}
D.~Amelunxen, M.~Lotz, M.~B. McCoy, and J.~A. Tropp, ``{Living on the edge:
  phase transitions in convex programs with random data},'' \emph{Information
  and Inference: A Journal of the IMA}, vol.~3, no.~3, pp. 224--294, 06 2014.

\bibitem{baraniuk_2017_exponential}
R.~G. Baraniuk, S.~Foucart, D.~Needell, Y.~Plan, and M.~Wootters, ``Exponential
  decay of reconstruction error from binary measurements of sparse signals,''
  \emph{IEEE Transactions on Information Theory}, vol.~63, no.~6, pp.
  3368--3385, 2017.

\bibitem{bennett2007netflix}
J.~Bennett, S.~Lanning \emph{et~al.}, ``The netflix prize,'' in
  \emph{Proceedings of KDD cup and workshop}, vol. 2007.\hskip 1em plus 0.5em
  minus 0.4em\relax New York, NY, USA., 2007, p.~35.

\bibitem{boufounos_greedy_2009}
P.~T. Boufounos, ``Greedy sparse signal reconstruction from sign
  measurements,'' in \emph{{Conference} {Record} of the {Forty}-{Third}
  {Asilomar} {Conference} on {Signals}, {Systems} and {Computers}}.\hskip 1em
  plus 0.5em minus 0.4em\relax IEEE, 2009, pp. 1305--1309.

\bibitem{boufounos2010reconstruction}
------, ``Reconstruction of sparse signals from distorted randomized
  measurements,'' in \emph{International Conference on Acoustics, Speech and
  Signal Processing}.\hskip 1em plus 0.5em minus 0.4em\relax IEEE, 2010, pp.
  3998--4001.

\bibitem{boufounos_1-bit_2008}
P.~T. Boufounos and R.~G. Baraniuk, ``1-bit compressive sensing,'' in
  \emph{42nd {Annual} {Conference} on {Information} {Sciences} and
  {Systems}}.\hskip 1em plus 0.5em minus 0.4em\relax IEEE, 2008, pp. 16--21.

\bibitem{boufounos_quantization_2015}
P.~T. Boufounos, L.~Jacques, F.~Krahmer, and R.~Saab, ``Quantization and
  compressive sensing,'' in \emph{Compressed sensing and its applications},
  ser. Applied and {Numerical} {Harmonic} {Analysis}.\hskip 1em plus 0.5em
  minus 0.4em\relax Birkhäuser/Springer, Cham, 2015, pp. 193--237.

\bibitem{candes_robust_2006}
E.~J. Candès, {J.}, T.~Tao, and J.~K. Romberg, ``Robust uncertainty
  principles: exact signal reconstruction from highly incomplete frequency
  information,'' \emph{IEEE Transactions on Information Theory}, vol.~52,
  no.~2, pp. 489--509, 2006.

\bibitem{candes_near_2006}
E.~J. Candès and T.~Tao, ``Near optimal signal recovery from random
  projections: universal encoding strategies?'' \emph{IEEE Transactions on
  Information Theory}, vol.~52, no.~12, pp. 5406--5425, 2006.

\bibitem{dirksen_tail_2015}
S.~Dirksen, ``Tail bounds via generic chaining,'' \emph{Electronic Journal of
  Probability}, vol.~20, pp. no. 53, 1--29, 2015.

\bibitem{dirksen2019quantized}
------, ``Quantized compressed sensing: a survey,'' in \emph{Compressed Sensing
  and Its Applications}.\hskip 1em plus 0.5em minus 0.4em\relax Springer, 2019,
  pp. 67--95.

\bibitem{dirksen2017one}
S.~Dirksen, H.~C. Jung, and H.~Rauhut, ``One-bit compressed sensing with
  partial gaussian circulant matrices,'' \emph{arXiv preprint
  arXiv:1710.03287}, 2017.

\bibitem{dirksen_robust_2018}
S.~Dirksen and S.~Mendelson, ``Non-{Gaussian} hyperplane tessellations and
  robust one-bit compressed sensing,'' \emph{arXiv preprint arXiv:1805.09409},
  2018.

\bibitem{dirksen_robust_circulant_2018}
------, ``Robust one-bit compressed sensing with partial circulant matrices,''
  \emph{arXiv preprint arXiv:1812.06719}, 2018.

\bibitem{donoho_compressed_2006}
D.~L. Donoho, ``Compressed sensing,'' \emph{IEEE Transactions on Information
  Theory}, vol.~52, no.~4, pp. 1289--1306, 2006.

\bibitem{foucart_flavors_2017}
S.~Foucart, ``Flavors of {Compressive} {Sensing},'' in \emph{Approximation
  {Theory} {XV}: {San} {Antonio} 2016}, G.~E. Fasshauer and L.~L. Schumaker,
  Eds.\hskip 1em plus 0.5em minus 0.4em\relax Cham: Springer International
  Publishing, 2017, pp. 61--104.

\bibitem{foucart_mathematical_2013}
S.~Foucart and H.~Rauhut, \emph{A {Mathematical} {Introduction} to
  {Compressive} {Sensing}}, ser. Applied and {Numerical} {Harmonic}
  {Analysis}.\hskip 1em plus 0.5em minus 0.4em\relax Birkhäuser, 2013.

\bibitem{gray1998quantization}
R.~M. Gray and D.~L. Neuhoff, ``Quantization,'' \emph{IEEE Transactions on
  Information Theory}, vol.~44, no.~6, pp. 2325--2383, 1998.

\bibitem{haghighatshoar2018low}
S.~Haghighatshoar and G.~Caire, ``Low-complexity massive mimo subspace
  estimation and tracking from low-dimensional projections,'' \emph{IEEE
  Transactions on Signal Processing}, vol.~66, no.~7, pp. 1832--1844, 2018.

\bibitem{haldar2010compressed}
J.~P. Haldar, D.~Hernando, and Z.-P. Liang, ``Compressed-sensing mri with
  random encoding,'' \emph{IEEE Transactions on Medical Imaging}, vol.~30,
  no.~4, pp. 893--903, 2010.

\bibitem{iwen2018recovery}
M.~A. Iwen, F.~Krahmer, S.~Krause-Solberg, and J.~Maly, ``On recovery
  guarantees for one-bit compressed sensing on manifolds,'' \emph{arXiv
  preprint arXiv:1807.06490}, 2018.

\bibitem{jacques2016error}
L.~Jacques, ``Error decay of (almost) consistent signal estimations from
  quantized gaussian random projections,'' \emph{IEEE Transactions on
  Information Theory}, vol.~62, no.~8, pp. 4696--4709, 2016.

\bibitem{jacques_quantized_2013}
L.~Jacques, K.~Degraux, and C.~De~Vleeschouwer, ``Quantized iterative hard
  thresholding: {Bridging} 1-bit and high-resolution quantized compressed
  sensing,'' \emph{arXiv preprint arXiv:1305.1786}, 2013.

\bibitem{jacques_dequantizing_2011}
L.~Jacques, D.~K. Hammond, and J.~M. Fadili, ``Dequantizing compressed sensing:
  when oversampling and non-{Gaussian} constraints combine,'' \emph{IEEE
  Transactions on Information Theory}, vol.~57, no.~1, pp. 559--571, 2011.

\bibitem{jacques_robust_2013}
L.~Jacques, J.~N. Laska, P.~T. Boufounos, and R.~G. Baraniuk, ``Robust 1-bit
  compressive sensing via binary stable embeddings of sparse vectors,''
  \emph{IEEE Transactions on Information Theory}, vol.~59, no.~4, pp.
  2082--2102, 2013.

\bibitem{knudson_one-bit_2014}
K.~Knudson, R.~Saab, and R.~Ward, ``One-bit compressive sensing with norm
  estimation,'' \emph{IEEE Transactions on Information Theory}, vol.~62, no.~5,
  pp. 2748--2758, 2016.

\bibitem{krause2017tractable}
S.~Krause-Solberg and J.~Maly, ``A tractable approach for one-bit compressed
  sensing on manifolds,'' in \emph{2017 International Conference on Sampling
  Theory and Applications (SampTA)}.\hskip 1em plus 0.5em minus 0.4em\relax
  IEEE, 2017, pp. 667--671.

\bibitem{laska_regime_2012}
J.~N. Laska and R.~G. Baraniuk, ``Regime change: bit-depth versus
  measurement-rate in compressive sensing,'' \emph{IEEE Transactions on Signal
  Processing}, vol.~60, no.~7, pp. 3496--3505, 2012.

\bibitem{ledoux_probability_1991}
M.~Ledoux and M.~Talagrand, \emph{Probability in {Banach} spaces}.\hskip 1em
  plus 0.5em minus 0.4em\relax Berlin: Springer-Verlag, 1991.

\bibitem{moshtaghpour_consistent_2016}
A.~Moshtaghpour, L.~Jacques, V.~Cambareri, K.~Degraux, and C.~D. Vleeschouwer,
  ``Consistent {Basis} {Pursuit} for {Signal} and {Matrix} {Estimates} in
  {Quantized} {Compressed} {Sensing},'' \emph{IEEE Signal Processing Letters},
  vol.~23, no.~1, pp. 25--29, Jan. 2016.

\bibitem{murphy2012fast}
M.~Murphy, M.~Alley, J.~Demmel, K.~Keutzer, S.~Vasanawala, and M.~Lustig,
  ``Fast l1-spirit compressed sensing parallel imaging mri: Scalable parallel
  implementation and clinically feasible runtime,'' \emph{IEEE Transactions on
  Medical Imaging}, vol.~31, no.~6, pp. 1250--1262, 2012.

\bibitem{plan_one-bit_2013}
Y.~Plan and R.~Vershynin, ``One-bit compressed sensing by linear programming,''
  \emph{Communications on Pure and Applied Mathematics}, vol.~66, no.~8, pp.
  1275--1297, 2013.

\bibitem{plan_robust_2013}
------, ``Robust 1-bit compressed sensing and sparse logistic regression: a
  convex programming approach,'' \emph{IEEE Transactions on Information
  Theory}, vol.~59, no.~1, pp. 482--494, 2013.

\bibitem{plan2016generalized}
------, ``The generalized lasso with non-linear observations,'' \emph{IEEE
  Transactions on Information Theory}, vol.~62, no.~3, pp. 1528--1537, 2016.

\bibitem{plan2016high}
Y.~Plan, R.~Vershynin, and E.~Yudovina, ``High-dimensional estimation with
  geometric constraints,'' \emph{Information and Inference: A Journal of the
  IMA}, vol.~6, no.~1, pp. 1--40, 2016.

\bibitem{roberts_picture_1962}
L.~Roberts, ``Picture coding using pseudo-random noise,'' \emph{IRE
  Transactions on Information Theory}, vol.~8, no.~2, pp. 145--154, Feb. 1962.

\bibitem{shi_methods_2016}
H.~J.~M. Shi, M.~Case, X.~Gu, S.~Tu, and D.~Needell, ``Methods for quantized
  compressed sensing,'' in \emph{2016 {Information} {Theory} and {Applications}
  {Workshop} ({ITA})}, Jan. 2016, pp. 1--9.

\bibitem{talagrand_upper_2014}
M.~Talagrand, \emph{Upper and lower bounds for stochastic processes}, ser.
  Ergebnisse der {Mathematik} und ihrer {Grenzgebiete}. 3. {Folge}. {A}
  {Series} of {Modern} {Surveys} in {Mathematics}.\hskip 1em plus 0.5em minus
  0.4em\relax Springer, Heidelberg, 2014, vol.~60.

\bibitem{vershynin_estimation_2014}
R.~Vershynin, ``Estimation in high dimensions: a geometric perspective,'' in
  \emph{Sampling theory, a renaissance}.\hskip 1em plus 0.5em minus 0.4em\relax
  Springer, 2015, pp. 3--66.

\bibitem{xu2018quantized}
C.~Xu and L.~Jacques, ``Quantized compressive sensing with rip matrices: The
  benefit of dithering,'' \emph{arXiv preprint arXiv:1801.05870}, 2018.

\end{thebibliography}
\end{document}